\documentclass[a4paper,fleqn]{cas-dc}
\usepackage{amsmath}
\usepackage{amsthm}
\usepackage[english]{babel}
\usepackage{enumitem}
\usepackage[font=small]{caption} 
\usepackage{subcaption}
\usepackage{newtxtext} 
\usepackage{graphicx,epstopdf}
\usepackage{multirow}
\usepackage{algorithm}
\usepackage{algpseudocode}
\newtheorem{theorem}{Theorem}

\newcommand{\Sectref}[1]{Section~\ref{#1}}

\newcommand\floor[1]{\lfloor#1\rfloor}
\newcommand\ceil[1]{\lceil#1\rceil}
\DeclareMathOperator*{\argmin}{argmin}

\usepackage[numbers,sort&compress]{natbib}

\def\tsc#1{\csdef{#1}{\textsc{\lowercase{#1}}\xspace}}
\tsc{WGM}
\tsc{QE}
\tsc{EP}
\tsc{PMS}
\tsc{BEC}
\tsc{DE}

\begin{document}
\let\WriteBookmarks\relax
\def\floatpagepagefraction{1}
\def\textpagefraction{.001}

\shorttitle{Node Cardinality Estimation in a Heterogeneous Wireless Network}

\shortauthors{S. Kadam et~al.}

\title [mode = title]{Node Cardinality Estimation in a Heterogeneous Wireless Network Deployed Over a Large Region Using a Mobile Base Station}                       
\tnotemark[1,2]

\tnotetext[1]{A preliminary version of this paper was presented at the IEEE SPCOM 2020 conference~\cite{kadam2020SPCOM}.}

\tnotetext[2]{S. Kadam worked on this research while he was with IIT Bombay. The contributions of G. Kasbekar have been supported by the grant RD/0222-EETIHBY-014.
}

%
\author[1]{Sachin Kadam}[orcid=0000-0001-7085-3365]
\ead{sachinkadam@skku.edu}
\affiliation[1]{organization={Technology Innovation Hub Foundation for IoT and IoE},
    addressline={IIT Bombay campus}, 
    city={Mumbai},
    postcode={400076}, 
    country={India}}
\author[2]{Kaustubh S. Bhargao}
\ead{kbhargao@ee.iitb.ac.in} 
\author[2]{Gaurav S. Kasbekar}
\ead{gskasbekar@ee.iitb.ac.in}
\affiliation[2]{organization={Department of Electrical Engineering},
     addressline={Indian Institute of Technology (IIT) Bombay}, 
    city={Mumbai},
    postcode={400076}, 
    country={India}}

\begin{abstract}
We consider the problem of estimation of the node cardinality of each node type in a heterogeneous wireless network with $T$ types of nodes deployed over a large region, where $T \ge 2$ is an integer. A mobile base station (MBS), such as that mounted on an unmanned aerial vehicle, is used in such cases since a single static base station is not sufficient to cover such a large region. The MBS moves around in the region and makes multiple stops, and at the last stop, it is able to estimate the node cardinalities for the entire region. {\color{black}In this paper, two schemes, viz., HSRC-M1 and HSRC-M2, are proposed to rapidly estimate the number of nodes of each type.} Both schemes have two phases, and they are performed at each stop. We prove that the node cardinality estimates computed using our proposed schemes are equal to, and hence as accurate as, the estimates that would have been obtained if a well-known estimation protocol designed for homogeneous networks in prior work were separately executed $T$ times. {\color{black}Closed-form expressions for the expected number of slots required by HSRC-M1 to execute and the expected energy consumption of a node under HSRC-M1 are computed. The problem of finding the optimal tour of the MBS around the region, which covers all the nodes and minimizes the travel cost of the MBS, is formulated and shown to be NP-complete, and a greedy algorithm is provided to solve it. Using simulations, it is shown that the numbers of slots required by the proposed schemes, HSRC-M1 and HSRC-M2, for computing node cardinality estimates are significantly less than the number of slots required for $T$ separate executions of the above estimation protocol for homogeneous networks.}
\end{abstract}





\begin{keywords}
Mobile Base Station \sep Node Cardinality Estimation \sep Heterogeneous Wireless Networks \sep Optimal MBS Tour Problem 
\end{keywords}

\maketitle

\section{Introduction}\label{sec:introduction}
Mobile base stations (MBSs), such as those mounted on unmanned aerial vehicles (UAVs),  are robust, highly mobile and agile, have wide coverage capability and high battery  backup capacity, and are being  extensively deployed in various military and civilian applications~\cite{mozaffari2019tutorial}. MBSs can be used for estimating the number of nodes such as  moving vehicles in traffic control systems~\cite{kanistras2013survey, ke2016real}, agricultural field monitoring sensors~\cite{giambene20195g}, people affected by disasters such as floods, wild fires~\cite{dinh2019flying}, etc. In~\cite{kanistras2013survey}, MBSs are used to estimate the number of vehicles moving on some congested roads per hour so that an efficient traffic controller can be designed. {\color{black}An MBS mounted on a UAV traverses the given set of busy roads and stops at specific spots for this purpose. The estimated data helps in dynamically fixing the traffic signal ON/OFF durations based on vehicle density.}
Consider an agricultural field in which sensors that measure various parameters (e.g., soil moisture, temperature, etc.) are deployed. 
Before collecting the actual data from the active sensors,\footnote{Active sensors are those that have some data to send to the MBS.} an MBS moves around the agricultural field and stops at predetermined spots to estimate the number of active sensors~\cite{giambene20195g}. This enhances the efficiency of the data collection process since the MBS can decide the amount of time it needs to spend at each stop when it again visits the same spots to collect the actual data, and it can inform the active sensors when to be available for sending the data based on the estimates. During natural disasters such as floods, wildfires, etc., MBSs hover over the affected region to find estimates of the number of affected people. These estimates are used to achieve efficient management of disaster relief teams and the distribution of relief materials~\cite{dinh2019flying}. Node cardinality estimation schemes are also useful in the design of medium access control protocols for wireless networks; in particular, the optimal contention probabilities, contention period durations, data transmission period durations, etc., can be computed as functions of the computed estimates~\cite{kadam2017fast, kadam2020fast, liew2019probability}. Various  Radio-Frequency IDentification (RFID) systems make use of node cardinality estimation schemes for inventory management, tag identification, detection, finding missing tags, etc.~\cite{arjona2018TagID, liu2019TagSearch, liu2020TagSearch, yu2018TagSearch1, yu2018TagSearch2, liu2022MissingTagID, fahim2018TagIDError, zhu2019MissingTagID, zang2018MissingTagID, liu2020TagIDContention}.  

A ``heterogeneous wireless network'' (HWN) is one in which different types of nodes are present~\cite{kadam2019rapid}. Different types of nodes may have different quality of service (QoS) requirements, e.g., one type of node may need to periodically transmit small amounts of data, another type may have strict delay requirements or need priority access to communicate alerts (e.g., in medical and security applications), another type may require high throughput (e.g., in multimedia surveillance applications) or highly reliable communication (e.g., in remote payment systems)~\cite{3gpp, MACM2M, M2MmobileInternet}, etc. Different types of nodes in an HWN may also have different hardware and software capabilities, e.g., different processor speeds, storage capabilities, transmission power, battery capacities, etc.~\cite{intelligentIOT}. {\color{black}In this paper, an HWN containing $T$ types of nodes, which are referred to as Type 1 ($\mathcal{T}_1$), $\dots$, Type $T$ ($\mathcal{T}_T$) nodes, is considered, where $T\geq 2$ is an arbitrary integer.}

Node cardinality estimation schemes for homogeneous Machine-to-Machine (M2M) networks and RFID systems have been proposed in~\cite{qian2011cardinality, arjona2017scalable, hou2015PLACE, liu2019NCENetworked, lin2019NCETash, zhou2018NCETimeVarying, ng2020NCEBeacon, deng2015NCEAlohaFrameCollisonRate, shahzad2014NCE(ART), shahzad2012NCE(ART), ferreira2019NCEAlohaFrame,  han2010NCEAnonymity, xi2020NCEInAndOut, nguyen2019NCEAlohaExpectationMaximization, hasan2018NCEGaussianDistribution, zheng2014NCE(ZOE), zheng2013NCE(ZOE), li2010NCEEnergyMin, zheng2011NCE(PET), wang2022NCENoisyChannel, xiao2016NCEInAndOut, kodialam2007NCE(EZB), lodialam2006NCE, bui2017novel, lugo2021NCEM2MLoadEstimationMPRExpectationMaximization}. Also, an estimation problem has been addressed for homogeneous RFID systems deployed over a large region in which, when the entire region is not in the coverage range of a reader,  it sequentially moves to different locations in order to estimate the total number of active tags present in the entire region~\cite{zhou2014understanding, han2010counting, xiao2019estimating}. However, the above schemes cannot be used to efficiently estimate the node cardinalities of different types of nodes in a \emph{heterogeneous} network with $T$ types of nodes, where $T \ge 2$ is an integer.

Node cardinality estimation schemes for heterogeneous RFID systems and M2M networks have been proposed in~\cite{lee2019NCEHeteroIDFree, xiao2019NCEHeteroSnapshot, zhang2020NCEHeteroSnapshotAnonymity,liu2016NCEHeteroFrameBinarySearch, kadam2017fast, kadam2020fast, sesha2019rapid, kadam2019rapid}. In these works, the authors have considered the case where a single static reader or base station (BS) can cover the entire region. {\color{black}In contrast, in this paper, the problem of rapidly estimating the node cardinality of each node type in an HWN with $T$ types of nodes, which are distributed over a large region, is addressed; in this case, a single static BS is not sufficient. Hence, an MBS is considered, which moves around making multiple stops to cover the region, so that the union of the coverage ranges of the MBS at the set of all stops covers all the nodes (see Section~\ref{nwmodel_SRCM}).} To the best of our knowledge, in prior work, no scheme has been designed for node cardinality estimation in an HWN with nodes distributed over a large region. (Note that separately executing a node cardinality estimation scheme designed for a homogeneous network multiple times to estimate the cardinality of each node type in an HWN is inefficient in general.) {\color{black}This is the area to which this paper contributes.}   

{\color{black}The multiple-set Simple RFID Counting (SRC$_M$) protocol proposed in~\cite{zhou2014understanding} for tag cardinality estimation in a homogeneous RFID system is briefly reviewed in Section~\ref{SubSec_SRCM},  and it is extended to design the proposed node cardinality estimation schemes for an HWN.} We propose two node cardinality estimation schemes, viz., HSRC-M1 and HSRC-M2; using any one of these, the MBS, after covering all the nodes in the region, can find the active node cardinality estimates of all $T$ types. Note that a challenge that needs to be overcome is that the coverage regions at different stops may overlap; hence, the proposed schemes are designed to prevent multiple counting of the same active nodes. The schemes, HSRC-M1 and HSRC-M2, proposed in this paper, are generalizations of the HSRC-1 and HSRC-2 schemes proposed in~\cite{kadam2019rapid}. We prove that the node cardinality estimates computed using any one of our proposed schemes, HSRC-M1 and HSRC-M2, \emph{equal and hence are as accurate as} the estimates that would have been obtained if the SRC$_M$  protocol were separately executed $T$ times to estimate the cardinalities of the $T$ node types. {\color{black}Closed-form expressions for the expected number of time slots required by HSRC-M1 to execute and the expected energy consumption of a node under HSRC-M1 are computed.}
Using simulations, we show that the number of time slots required by each of the proposed schemes, viz., HSRC-M1 and HSRC-M2, for computing the node cardinality estimates is significantly less compared to the number of slots required for $T$ separate executions of the SRC$_M$ protocol.

Now, an MBS is usually battery-powered and, once deployed, starts at a charging station, which is its initial stop,  covers the desired region and returns to the initial stop. Since it is battery-operated, its path must be optimized in order to save energy. Consider a set of possible stops for the MBS. There is a travel cost incurred by the MBS when it moves from one given stop to another. Also, the nodes may be active-powered (having a power source, e.g., a battery) or passive-powered (having no power source, temporarily powered by the MBS). In both cases, the energy expenditure of nodes must be minimized, since frequently replacing the batteries of a large number of nodes would be expensive and since passive-powered nodes have a small amount of available energy. From the set of all possible stops, the MBS must choose the number of stops and their locations such that it covers every node in the network at least once. Also, the MBS must visit a given stop at most once. {\color{black}Taking these facts into account, the optimal MBS tour (OMT) problem is formulated,} which is the problem of finding the optimal tour of the MBS around the region with the following objective and constraints: (a) minimizing the travel cost of the MBS, (b) ensuring that the total energy spent by the nodes during the estimation process does not exceed a given upper limit, (c) ensuring that all nodes are covered at least once, and (d) ensuring that a given stop is visited by the MBS at most once. {\color{black}The OMT problem is formulated as an optimization problem and proven to be NP-complete, and a greedy algorithm is proposed to solve it.} 

The rest of the paper is organized as follows. {\color{black}In Section~\ref{Sec:RelWork}, a review of related prior literature is provided. The network model and problem formulation are presented in Section~\ref{nwmodel_SRCM}. In Section~\ref{SubSec_SRCM}, a brief review of the SRC$_M$ protocol is provided. The proposed node cardinality estimation schemes, viz., HSRC-M1 and HSRC-M2, are described in Section~\ref{EstScheme_SRCM}. In Sections~\ref{p2_3stage} and~\ref{subsec_energy}, the expected number of time slots and the expected energy consumption of a node under HSRC-M1, respectively, are mathematically analyzed. In Section~\ref{OMT}, the OMT problem is formulated, proven to be NP-complete, and a greedy algorithm is provided to solve it. Simulation results are provided in Section~\ref{Sec:Simu} and conclusions and directions for future work are provided in Section~\ref{conclusion}.}

\section{Related Work}\label{Sec:RelWork}
Node cardinality estimation is a crucial problem in several domains, including  RFID systems, the Internet of Things, and data networks. In~\cite{cohen2017CECompNetsOnlineML, ullah2021CESwitchingnets, wang2022NCEDataNetsOnlineCE2, wang2022NCEDataNetsOnlineCE1}, the authors discuss the problem of node cardinality estimation in packet-switching networks for network traffic monitoring, popularity tracking on social media, and network security. 
Online cardinality estimation schemes are proposed in~\cite{wang2022NCEDataNetsOnlineCE2, wang2022NCEDataNetsOnlineCE1} for automatically adapting to different stream sizes in data networks. A similar sampling-based cardinality estimation algorithm is proposed in~\cite{cohen2017CECompNetsOnlineML}, which uses online machine learning to adapt to changes in flow size distribution. In contrast, in this paper, we focus on a wireless network and calculate the node cardinality estimates independently in each frame.

RFID systems consist of two components: tags and reader(s). A tag has its unique ID pre-stored, and it transmits this ID whenever interrogated by a reader. Reader(s) communicate with the tags and perform fundamental tasks such as RFID tag identification, tag search, missing tag identification, tag counting (estimation), etc. These tasks are crucial operations for inventory management~\cite{liu2020TagSearch}. {\color{black} Tag identification is the task of collecting the IDs of all unknown tags in a set~\cite{hou2018NCEPHY1,lin2022compact}.} Various tag identification protocols based on Slotted ALOHA~\cite{arjona2018TagID}, filtering vectors~\cite{zhu2019MissingTagID}, and binary tree~\cite{hou2018NCEPHY1, fahim2018TagIDError}, have been proposed. {\color{black} Tag search is the task of searching the system to find out as to which wanted tags are currently present within the interrogation region of the reader~\cite{yu2018TagSearch1, yu2018TagSearch2, liu2020TagSearch, liu2019TagSearch,liu2023efficient}. Missing tag identification is the task of identifying the tags that have entered or exited the interrogation region of the reader and various protocols are proposed in~\cite{liu2022MissingTagID, zang2018MissingTagID,liu2022revisiting,chu2022efficient} to get the IDs of the missing tags. Algorithms for getting estimates of the tags that either left or entered the system in a certain time frame are proposed in~\cite{xi2020NCEInAndOut, xiao2016NCEInAndOut}. Apart from RFID systems, security surveillance frameworks are used in constructing AI digital twins~\cite{kim2022eco} and providing discriminative public and private services~\cite{kim2022intelligent}. In contrast to the problems considered in the above papers, in this paper, the problem of node cardinality estimation by an MBS in an HWN deployed over a large region is considered.}

Protocols for tag cardinality estimation in homogeneous RFID systems have been proposed in~\cite{qian2011cardinality, arjona2017scalable, hou2015PLACE, liu2019NCENetworked, lin2019NCETash, zhou2018NCETimeVarying, ng2020NCEBeacon, deng2015NCEAlohaFrameCollisonRate, shahzad2014NCE(ART), shahzad2012NCE(ART), ferreira2019NCEAlohaFrame,  han2010NCEAnonymity, xi2020NCEInAndOut, nguyen2019NCEAlohaExpectationMaximization, hasan2018NCEGaussianDistribution, zheng2014NCE(ZOE), zheng2013NCE(ZOE), li2010NCEEnergyMin, zheng2011NCE(PET), wang2022NCENoisyChannel, xiao2016NCEInAndOut, kodialam2007NCE(EZB), lodialam2006NCE}. Node cardinality estimation schemes for homogeneous M2M networks have been proposed in~\cite{bui2017novel, lugo2021NCEM2MLoadEstimationMPRExpectationMaximization}. The above schemes assume that the reader or BS is fixed and that all the tags or nodes are present inside the coverage region of the reader or BS. Another estimation problem has been addressed for homogeneous RFID systems deployed over a large region in which, when the entire region is not in the coverage range of a reader, it sequentially moves to different locations in order to estimate the total number of active tags present in the entire region~\cite{zhou2014understanding,  han2010counting, xiao2019estimating}. However, these schemes cannot be used to efficiently estimate the node cardinalities of different types of nodes in a \emph{heterogeneous} network with $T$ types of nodes, where $T \ge 2$ is an integer.

Node cardinality estimation schemes for heterogeneous RFID systems and M2M networks have been proposed in~\cite{lee2019NCEHeteroIDFree, xiao2019NCEHeteroSnapshot, zhang2020NCEHeteroSnapshotAnonymity,liu2016NCEHeteroFrameBinarySearch, kadam2017fast, kadam2020fast, sesha2019rapid, kadam2019rapid}. In these works, the authors have considered the case where a single static reader or BS can cover the entire region. In contrast, in this paper, we address the problem of rapidly estimating the node cardinality of each node type in an HWN with $T$ types of nodes, which are distributed over a large region; in this case, a single static BS is not sufficient. {\color{black}Hence, an MBS is considered, which moves around and makes multiple stops to cover the region, so that the union of the coverage ranges of the MBS at the set of all stops covers all the nodes.} To the best of our knowledge, our paper is the first to propose node cardinality estimation schemes for an HWN with nodes distributed over a large region.

\section{Network Model, Problem Formulation, and Background}
\label{SC:nw:model:prb:form:background}

\subsection{Network Model and Problem Formulation}\label{nwmodel_SRCM}

\begin{figure}[tbp]
  \centering
    \includegraphics[width=0.48\textwidth]
    {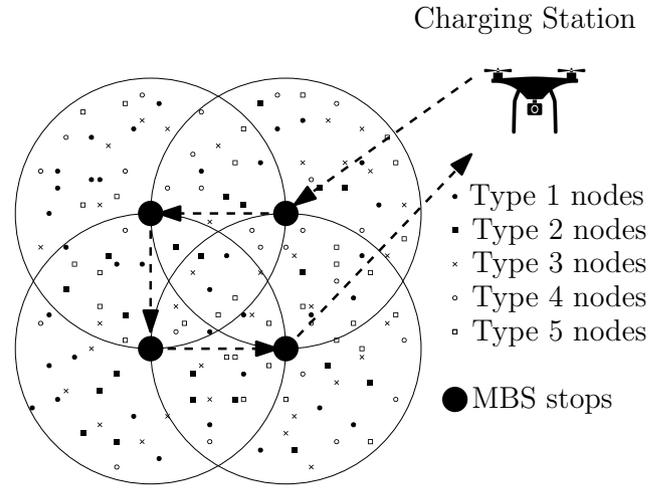}
    \caption{The figure shows $M = 4$ locations (stops) of a mobile base station (MBS) and $T = 5$ types of nodes in a region. The coverage range of the MBS at a stop is the area  inside the circle with that stop as the centre.}
    \label{Network_Model_SRCM} 
\end{figure}

Consider an HWN consisting of an MBS, which moves around in a given region and stops at $M$ different locations, and $T$ different types, say Type 1 ($\mathscr{T}_1$), \ldots, Type $T$ ($\mathscr{T}_T$), of wireless devices (nodes) in the region,  where $T \ge 2$ is an arbitrary integer. {\color{black}It is assumed that all the nodes lie within the union of the coverage ranges of the MBS at the $M$ stops.} Also, the different types of nodes may, e.g., be nodes that send emergency, periodic, normal type data, etc. Fig.~\ref{Network_Model_SRCM} illustrates such a network for the case $T=5$ and $M=4$. {\color{black}The sets of nodes of $\mathscr{T}_1$, \ldots, $\mathscr{T}_T$ are denoted by $\mathscr{N}_1$, \ldots, $\mathscr{N}_T$, respectively;} let $|\mathscr{N}_b| = N_b$, $b \in \{1, \ldots, T\}$.\footnote{$|A|$ denotes the cardinality of set $A$.} 
Only a subset of the set of all nodes, i.e., $\bigcup_{b=1}^{T} \mathscr{N}_b$, are \emph{active}, i.e., have data to send. 
Let $\mathscr{A}^{(m)}_{b}$ be the set of active nodes of $\mathscr{T}_b$, $b \in \{1,\ldots, T\}$, within the coverage range of the MBS when it is at stop $m \in \{1, \ldots, M\}$, $n^{(m)}_{b} = |\bigcup_{i=1}^{m} \mathscr{A}^{(i)}_{b}|$, and $\bar{n}^{(m)}_b=|\mathscr{A}_b^{(m)}|$. 
So the total number of active nodes of $\mathscr{T}_b$ in the entire region is $n_{b} =  n^{(M)}_{b} = |\bigcup_{m=1}^{M} \mathscr{A}^{(m)}_{b}|$.  Our objective is to rapidly estimate the values of $n_b$, $b \in \{1,\ldots, T\}$. In particular, let $\hat{n} ^{(m)}_{b}$ (respectively, $\hat{n}_{b}$)  be the estimated value of $n^{(m)}_{b}$ (respectively, $n_b$). 
Let $\epsilon$, the desired relative error bound, and $\delta$, the desired error probability, be the user specified accuracy requirements, i.e., the accuracy parameters with which the estimate $\hat{n}_b$ needs to be obtained.  Our objective is to rapidly find estimates $\hat{n}_b$ for $n_b$, $b \in \{1,\ldots, T\}$, such that $P(| \hat{n}_b - n_b|$ $\leq$ $\epsilon n_b)$ $\geq 1 - \delta$, $\forall b \in \{1,\ldots, T\}$. 

Note that a node may lie in the coverage region(s) of the MBS at one or more stopping locations. We cannot simply use an estimation protocol designed for a network with  a static base station~\cite{kadam2017fast, kadam2020fast, sesha2019rapid, kadam2019rapid} to separately estimate the node cardinality of a given type at each stop $m$ and add up the estimates since a node that is in the range of the MBS at multiple stops would appear multiple times in the estimate.

{\color{black}A list of mathematical symbols used in this paper along with their meanings is provided in Table~\ref{Tab:MathSymbols}.}
\begin{table}[!ht]
    \centering
    \caption{List of mathematical symbols}
    \begin{tabular}{|c|l|}
        \hline
         Symbol & Meaning \\ \hline
         $T$ & Number of types of nodes \\ \hline
         $M$ & Number of stops of MBS
         \\ \hline
         $\mathscr{T}_i$ & Type $i$ \\ \hline
         $\mathscr{N}_i$ & Set of nodes of $\mathscr{T}_i$ \\ \hline
         $\mathscr{A}^{(m)}_{b}$ & Set of active nodes of $\mathscr{T}_b$ at stop $m$ \\ \hline
         $n_{b}$ & Total number of active nodes of $\mathscr{T}_b$ in the \\&entire region \\ \hline
         $\tilde{n}_{b}$ & Estimated value of $n_b$ after phase 1 \\ \hline
         $\hat{n}_{b}$ & Final estimated value of $n_b$ \\ \hline
         $\epsilon$ & Desired relative error bound \\ \hline
         $\delta$& Desired error probability \\ \hline
         $W$ & Number of independent trials executed in \\&phase 1 of the SRC$_M$ protocol  \\ \hline
         $t$ & Number of slots used in phase 1 of the \\&SRC$_M$ protocol for every trial \\ \hline
         $\ell$ & Number of slots used in phase 2 of the \\&SRC$_M$ protocol  \\ \hline
         $\alpha$, $\beta$ & Symbols used for transmission \\ \hline
         $B^{(m)}_h$, &$h^{th}$ Block , $h \in \{1, \ldots, \ell \}$ \\ \hline
         $E$ and $C$ & Empty and Collision \\ \hline
         $\Delta^{(m)}_i$ & Number of slots required in phase $i$ of \\&HSRC-M1, at stop $m \in \{1, \ldots, M\}$ \\ \hline
         $Z^{(m)}_{BP}$ & Number of slots required by the \\&broadcast packets BP$^{(m)}_1$ and BP$^{(m)}_2$ \\ \hline
         $S_W$ & Slot width in bits \\ \hline
         $\xi_b^{(m)}$ &  Probability with which a given node is\\& active when the MBS is at stop $m$ \\ \hline
         $\bar{\gamma}_b$ & Energy spent per slot by a node of $\mathscr{T}_b$,\\&  in the idle state \\ \hline $\hat{\gamma}_b$ & Energy spent per slot by a node of $\mathscr{T}_b$, \\& in the reception state \\ \hline
         $\gamma^{\alpha}_1$ & Energy spent per slot by a node of $\mathscr{T}_1$ \\&for transmitting the symbol $\alpha$ \\ \hline         $\gamma^{\beta}_b$ & Energy spent per slot by a node of $\mathscr{T}_b$ \\&for transmitting the symbol $\beta$ \\ \hline
         $\gamma'$ & Energy spent by a node  for transmitting \\&a signal in the transmission state in phase 1 \\ \hline
         $\Gamma^{(m)}_{b,\tau}(\Phi_i)$ & Expected energy consumption of a \\&$\mathscr{T}_b$ node, in phase $i$, in the transmission state \\ \hline
         $\Gamma^{(m)}_{b,\iota}(\Phi_i)$ & Expected energy consumption of a \\&$\mathscr{T}_b$ node, in phase $i$, in the idle state \\ \hline
         $\Gamma^{(m)}_{b,\rho}(\Phi_i)$ & Expected energy consumption of a \\&$\mathscr{T}_b$ node, in phase $i$, in the reception state \\ \hline
         $\mathcal{G}$ & A connected directed graph \\ \hline
         $\mathcal{M}$ & Set of possible stops for the MBS \\ \hline
         $\mathcal{E}$ & Set of links (edges) between the stops \\ \hline
         $\mathcal{N}$ & Set of nodes in the region \\ \hline
         $c_{u,v}$ &  Cost of travelling from stop $u$ to stop $v$\\& via the link from $u$ to $v$  \\ \hline
         $\eta_{k,m}$ & Energy spent by node $k$ during the \\&estimation process when it is in the\\& range of the MBS at stop $m$ \\ \hline
         $\overline{\eta}$ & Upper limit on the total energy \\&consumed by the nodes \\ \hline
    \end{tabular}    
    \label{Tab:MathSymbols}
\end{table}

\subsection{Review of the Multiple-set Simple RFID Counting (\texorpdfstring{SRC$_M$}{}) Protocol}\label{SubSec_SRCM}
The SRC$_M$ protocol was proposed in~\cite{zhou2014understanding} for node cardinality estimation in a homogeneous multiple-set network; our proposed schemes extend it for node cardinality estimation in an HWN with $T$ types of nodes using an MBS.\footnote{Note that in~\cite{zhou2014understanding}, the SRC$_M$ protocol is designed for an RFID system, which consists of a reader and several tags in a region. {\color{black}In this paper, the SRC$_M$ protocol is used in the context of a wireless network; so while reviewing the SRC$_M$ protocol, the terms ``MBS'' and ``node'' are used in place of ``reader'' and ``tag'', respectively}.} The SRC$_M$ protocol estimates the number of active nodes in a homogeneous network within the given accuracy requirements of $\epsilon$ and $\delta$. Let us consider the network model and objectives described in Section~\ref{nwmodel_SRCM} with $T=1$. In the SRC$_M$ protocol, the MBS moves around and makes $M$ stops, and at each stop $m \in \{1, \ldots, M\}$, it estimates the number of active nodes, say ${n} ^{(m)}_{1}$,  present in the union of its coverage regions at the stops up to stop $m$.   The SRC$_M$ protocol consists of two phases, and it executes both phases at each stop $m \in \{1, \ldots, M\}$ (see Fig.~\ref{SRCM_window}). At stop $m$, at the end of phase 1 (respectively, phase 2), it finds a rough estimate $\tilde{n} ^{(m)}_{1}$ (respectively, the final estimate $\hat{n} ^{(m)}_{1}$) of ${n} ^{(m)}_{1}$~\cite{zhou2014understanding}. Phase 1 (respectively, phase 2) of the protocol consists  of a sequence of trials (respectively, a single trial), and each trial consists of multiple  slots (see Fig.~\ref{SRCM_window}).  The number of slots in a trial is called the length of the trial. After a trial, a slot can be in one of the following three states: (i) \emph{Empty}: No node  transmitted in that slot, (ii) \emph{Success}: Exactly one node transmitted in that slot, (iii) \emph{Collision}: More than one node transmitted in that slot.
{\color{black}A brief review of phase 1 (respectively, phase 2) of the SRC$_M$ protocol is provided in Section~\ref{SubSec_SRCM_P1} (respectively, Section~\ref{SubSec_SRCM_P2}).}

\subsubsection{Review of Phase 1 of the \texorpdfstring{SRC$_M$}{} Protocol}\label{SubSec_SRCM_P1}
Let $T=1$, $n_{all}$ be the total number of nodes manufactured and $t = \ceil{\log_{2}{n_{all}}}$.\footnote{$\ceil{x}$ denotes the smallest integer greater than or equal to $x$. Also, $\floor{x}$ denotes the largest integer smaller than or equal to $x$.} At each stop $m \in \{1, \ldots, M\}$, phase 1 of the SRC$_M$ protocol executes $W$ independent trials, each consisting of $t$ time slots, and $W$ is determined based on the desired error probability $\delta$. 
{\color{black}For a given $\delta \in (0, 1)$, the smallest $W$, which satisfies the following equation, is chosen~\cite{zhou2014understanding}:} 
\begin{equation}
    \sum_{i=\floor{\frac{W+1}{2}}}^{W} {\binom{W}{i}}(1-\delta)^{i}\delta^{W-i} \ge 1 - \delta. 
\end{equation}
For example, for $\delta=0.2$, $W = 30$ is used. Let 
\begin{equation}
\label{EQ:Phase1:bar:pi}
{\bar{p}}^{(m)}(i) = \left\{ 
\begin{array}{ll}
2^{-i}, & \mbox{for } i \in \{1, \ldots, t-1\}, \\
2^{-(t-1)}, & \mbox{for } i = t. \\
\end{array}
\right. 
\end{equation}

Suppose in trial $w \in \{1, \ldots, W\}$, each active node in the coverage range of the MBS at stop $m$ independently transmits in a slot $i$, $i \in \{1, \ldots, t\}$, with probability ${\bar{p}}^{(m)}(i)$ (see Fig.~\ref{SRCM_window}). Let $s^{(m)}_w(i)$, $i \in \{1, \ldots, t\}$, be a bit vector of length $t$ at stop $m \in \{1, \ldots, M\}$, and trial $w \in \{1, \ldots, W\}$, whose $i^{th}$ bit is 0 if no active node transmits in the $i^{th}$ slot, else it is 1. {\color{black}Another bit vector is found using the following equation:} ${Y}^{(m)}_w(i) = {Y}^{(m-1)}_w(i) \lor s^{(m)}_w(i)$,\footnote{$\lor$ denotes the bitwise OR operator.} where ${Y}^{(0)}_w(i)$ is a bit vector, all of whose elements are zeros. 
Now, phase 1 searches ${Y}^{(m)}_w(i)$ at bit positions $i = 2^{j-1}$, $j \in \{1, \ldots, 1+ \log_2(t) \}$, until it encounters a  $0$ bit at $j = j'$ (say) for the first time.\footnote{If no bit is $0$, then we take $v_w^m = t$.} Then it uses the binary search algorithm over the set $\{2^{j'-2}, \ldots, 2^{j'-1}-1\}$, to find the maximum integer (slot number) $v_w^m$, $w \in \{1, \ldots, W\}$, such that the bit ${Y}^{(m)}_w(v_w^m)$ is $1$.
At stop $m$, at the end of all $W$ trials, the estimate of $n^{(m)}_{1}$ is computed as~\cite{zhou2014understanding}:  
\begin{equation}
\label{EQ:Phase1:estimate:nm_tilde}
\tilde{n} ^{(m)}_{1}= 0.794 \times 2^{\left(\Sigma_{w=1}^{W} v_w^m\right)/W}.
\end{equation}

\subsubsection{Review of Phase 2 of the \texorpdfstring{SRC$_M$}{} Protocol}\label{SubSec_SRCM_P2}
\begin{figure}
\begin{center}
	\includegraphics[scale = 0.45]{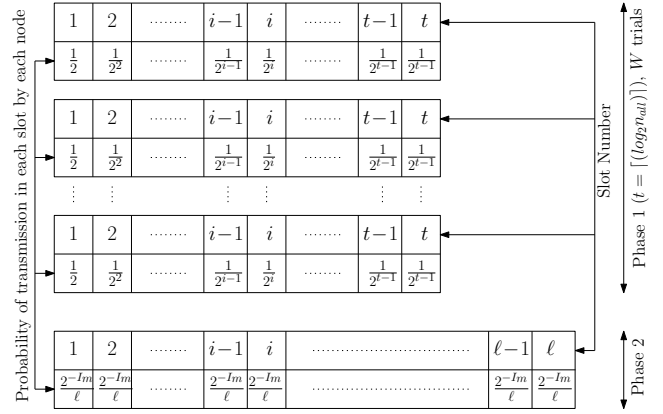}
	\caption{The figure shows the frame structure used in the SRC$_M$ protocol when the MBS is at stop $m$.}
	\label{SRCM_window}
\end{center}
\end{figure}

At each stop $m \in \{1, \ldots, M\}$, phase 2 of the SRC$_M$ protocol uses the ``balls-and-bins'' ($BB$)  method~\cite{zhou2014understanding}. In this method, each active node independently chooses a slot out of a fixed number, say $\ell$, of slots uniformly at random, transmits in that slot with a fixed probability assigned to it, and otherwise does not transmit (see Fig.~\ref{SRCM_window}). The probability of transmission is $2^{-I_m}$, where for each $m \in \{1, \ldots, M \}$, the value of $I_m$ is computed using the following equations:
\begin{align}
\label{EQ:SRCM:p_m}
p_m &= \min{(1, 1.6 \ell/\tilde{n}^{(m)}_{1})}, \\ 
\label{EQ:SRCM:I_m}
I_m &=  \argmin_{j \in \{1, 2, 3, \ldots \}} |2^{-j} - p_m |.
\end{align}
Note that the probability of transmission by a given active node in a given slot out of the $\ell$ slots is $\frac{2^{-I_m}}{\ell}$ (see Fig.~\ref{SRCM_window}). 
The parameter $\ell$ is a function of the desired relative error $\epsilon$ and it is found from a numerical lookup table, which is constructed by executing the SRC$_M$ protocol for different values of $n^{(m)}_{1}$, and finding the value of $\ell$ required to achieve a given value of $\epsilon$~\cite{zhou2014understanding}.

Let $z$ be the number of slots out of the $\ell$ slots that are empty in the phase 2 executions at all stops $m \in \{1, \ldots, M \}$.  
The protocol counts the number of empty slots, $z$, and computes the final estimate as follows~\cite{zhou2014understanding}:\footnote{If $z = 0$, then $\hat{n}_1$ is set to an arbitrary integer.}
\begin{equation}
\label{EQ:SRCM:hatn}
\hat{n}_1 = \hat{n}^{(M)}_{1}= \frac{\ln({z/\ell})}{\ln{(1-p_M/\ell})}.
\end{equation} 

 \section{Proposed Node Cardinality Estimation Schemes}\label{EstScheme_SRCM}
We now describe the proposed schemes, which are extensions of the SRC$_M$ protocol for estimating the number of active nodes of each type in the model with an MBS with $M$ stops and $T$ different types of nodes in the union of its coverage ranges described in Section~\ref{nwmodel_SRCM}. The proposed schemes are Heterogeneous SRC$_M$-1  (HSRC-M1) and   Heterogeneous SRC$_M$-2 (HSRC-M2), and both consist of two phases-- they correspond to the two phases of the SRC$_M$ protocol (see Section~\ref{SubSec_SRCM}). Each of HSRC-M1 and HSRC-M2 executes both its phases at each stop $m \in \{1, \ldots, M\}$. Phase 2 of HSRC-M1 consists of $3$ steps and that of HSRC-M2  consists of $2$ steps (except for $T=2$ and $T=3$). So henceforth, we refer to them as ``The 3-Step Protocol'' (3-SP) (see Section~\ref{Sec:HSRC-M1}) and ``The 2-Step Protocol'' (2-SP) (see Section~\ref{Sec:HSRC-M2}), respectively. 

At each stop $m \in \{1, \ldots, M\}$, let the MBS have an array of length $\ell$ (number of slots used in phase 2 of the SRC$_M$ protocol) assigned to nodes of each $\mathscr{T}_b$, $b \in \{1, \ldots, T\}$; also, each element of the array takes value 0 or 1. {\color{black}This array is referred to as the bit pattern of $\mathscr{T}_b$, $b \in \{1, \ldots, T\}$, and denoted by $X^{(m)}(b,i)$, $i \in \{1, \ldots, \ell\}$.}
At each stop $m \in \{1, \ldots, M\}$, during phase 1 of HSRC-M1 or HSRC-M2, phase 1 of SRC$_M$ is separately executed $T$ times to compute rough estimates, $\tilde{n}^{(m)}_{1},\ldots, \tilde{n}^{(m)}_{T}$, of the active node cardinalities of the $T$ node types. 
Note that this requires the execution of $W$ independent trials for each node type, i.e., a total of $WT$ independent trials at each stop $m$, in phase 1.\footnote{We use this simple scheme of $T$ separate executions of phase 1 of SRC$_M$ in phase 1 of HSRC-M1 and HSRC-M2 since the number of slots used in phase 1 of SRC$_M$ is negligible compared to the number of slots used in phase 2. The performances of the proposed schemes can be slightly improved by instead using schemes similar to those described in Sections~\ref{Sec:HSRC-M1} and~\ref{Sec:HSRC-M2} in phase 1.}

In phase 2 of both HSRC-M1 and HSRC-M2, one possible approach is to execute $T$ trials of phase 2 of SRC$_M$ separately, one for each of the $T$ node types;  note that in the trial for  $\mathscr{T}_b$ nodes,  the probability $\frac{2^{-I_{b,m}}}{\ell}$, where $I_{b,m}$ is given by \eqref{EQ:SRCM:I_bm}, is used as the probability with which a given active node of $\mathscr{T}_b$ transmits in a given slot. This approach requires a total of $MT\ell$ time slots to execute. {\color{black}This approach is referred to as ``$T$-$Rep$''. In Section~\ref{Sec:HSRC-M1} and Section~\ref{Sec:HSRC-M2}, alternative schemes to execute phase 2 of HSRC-M1 and HSRC-M2, respectively, are described, which require fewer time slots to execute than ``$T$-$Rep$''.}

\subsection{Phase 2 of the Heterogeneous \texorpdfstring{SRC$_M$-1}{} (HSRC-M1) Scheme}\label{Sec:HSRC-M1}
\begin{figure}
	\begin{center}
		\includegraphics[scale = 0.43]{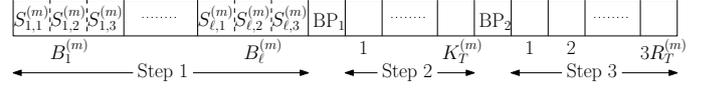}
		\caption{The figure shows the frame structure used in the 3-Step Protocol (3-SP) for the case $T=4$.} 
		\label{Est_Window}
	\end{center}
\end{figure}
{\color{black}Recall that in phase 2 of the HSRC-M1 scheme, 3-SP is used, which is described now.}  
At each stop $m \in \{1, \ldots, M\}$, step 1 of 3-SP consists of $\ell$ blocks, say $B^{(m)}_h$, $h \in \{1, \ldots, \ell \}$ (see Fig.~\ref{Est_Window}).  
Each block, $B^{(m)}_h$, is divided into ($T-1$) slots $S^{(m)}_{h,1}$, \ldots, $S^{(m)}_{h,T-1}$. Each active node of each of the $T$ types independently chooses a block $B^{(m)}_h$ out of the $\ell$ blocks uniformly at random and transmits its corresponding symbol combination (see Fig.~\ref{Sym_Combo3}) in that block with  probability $2^{-I_{b,m}}/\ell$, where $I_{b,m}$ is obtained from the following equations:\footnote{Note that the values of $\tilde{n}^{(m)}_{b}$, $m \in \{1, \ldots, M\}$, $b \in \{1, \ldots, T\}$, are available from phase 1 of the HSRC-M1 protocol.} 
\begin{align}
\label{EQ:SRCM:p_bm}
p_{b,m} & = \min{(1, 1.6 \ell/\tilde{n}^{(m)}_{b})},\\
\label{EQ:SRCM:I_bm}
I_{b,m} & =  \argmin_{j \in \{1, 2, 3, .....\}} |2^{-j} - p_{b,m} |. 
\end{align}
Fig.~\ref{Sym_Combo3} shows the symbol combination used by each node type. $\mathscr{T}_1$ active nodes whose chosen block is $B^{(m)}_h$ transmit symbol $\alpha$ in all ($T-1$) slots, i.e., $S^{(m)}_{h,1}$, \ldots, $S^{(m)}_{h,T-1}$, of block $B^{(m)}_h$. $\mathscr{T}_2$ (respectively, $\mathscr{T}_3, \ldots, \mathscr{T}_T$) active nodes whose chosen block is $B^{(m)}_h$ transmit symbol $\beta$ in slot $S^{(m)}_{h,1}$ (respectively, $S^{(m)}_{h,2}$, \ldots, $S_{h,T-1}$) and do not transmit in the other slots of block $B^{(m)}_h$. For example, for $T=4$, $\mathscr{T}_1, \mathscr{T}_2, \mathscr{T}_3,$ and $\mathscr{T}_4$ active nodes whose chosen block is $B^{(m)}_h$ transmit symbols ($\alpha, \alpha, \alpha$), ($\beta, 0, 0$),  ($0, \beta, 0$), and ($0, 0, \beta$), respectively,  in the $(T-1) = 3$ slots of $B^{(m)}_h$. Step 1 concludes with this. The outcome in each slot can be any of the following: (i) Empty ($E$), (ii) Success ($\alpha$ or $\beta$), (iii) Collision ($C$). 

\begin{figure}
\centering
\resizebox{0.8\columnwidth}{!}{\includegraphics{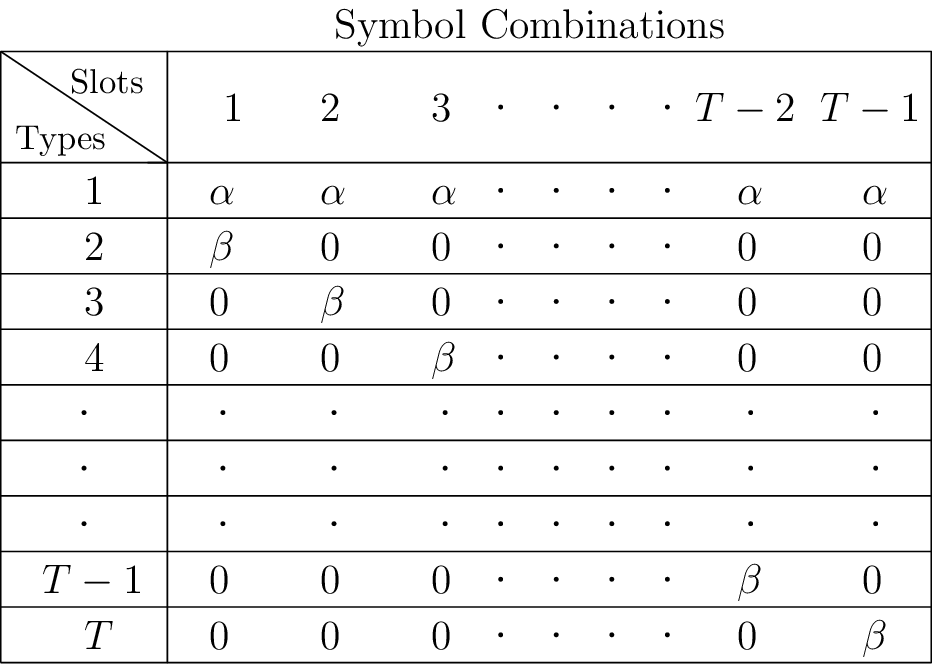}}
\caption{ The figure shows the symbol combination used by each node type in HSRC-M1. The symbol $0$ indicates ``no transmission''.}
 \label{Sym_Combo3} 
\end{figure}

The MBS considers the result of each block $B^{(m)}_h$, $h \in \{1, \ldots, \ell \}$, in step 1 to determine which of the active nodes should participate in step 2. If the MBS unambiguously determines the types of nodes that transmitted in a given block $B^{(m)}_h$, then active nodes of the $T$ types, which chose block $B^{(m)}_h$,  do not need to participate in step 2; otherwise, $\mathscr{T}_1$ active nodes that participated in $B^{(m)}_h$ must participate again in step 2 (this occurs when the outcome is ($C, \ldots, C$), in which collisions occur in all $T-1$ slots of $B^{(m)}_h$). The outcome ($C, \ldots, C$)  could be the result of transmissions by at least two nodes from $\mathcal{N}_1$, or by one node from $\mathcal{N}_1$ and at least one node each from $\mathcal{N}_2, \ldots, \mathcal{N}_T$, or by at least two nodes each from $\mathcal{N}_2, \ldots, \mathcal{N}_T$ and none from $\mathcal{N}_1$. To resolve the ambiguity, after the end of step 1, the MBS transmits a broadcast packet (BP), say BP$_1$ (see Fig.~\ref{Est_Window}), in which the list of the numbers of all blocks in step 1 in which collisions occurred in all ($T-1$) slots is encoded. 

In step 2, there are $K^{(m)}_T$ slots (see Fig.~\ref{Est_Window}), where $K^{(m)}_T$ is the number of blocks in step 1 in which collisions occurred in all ($T-1$) slots. For $i \in \{1, \ldots, K^{(m)}_T$\}, in the $i^{th}$ slot of step 2, $\mathscr{T}_1$ nodes that transmitted in the $i^{th}$ block of step 1 in which collisions occurred in all ($T-1$) slots, transmit symbol $\alpha$. $\mathscr{T}_2, \ldots, \mathscr{T}_T$ nodes do not transmit in step 2. Now, it is easy to see that at the end of step 2, the MBS  unambiguously knows the set of block numbers from step 1 in which $\mathscr{T}_1$ nodes transmitted. However,  if in step 2, there are collisions in some of the slots, ambiguity remains with the MBS  on whether $\mathscr{T}_2, \ldots, \mathscr{T}_T$ nodes transmitted in the corresponding blocks of step 1. To resolve this ambiguity, after the end of step 2, the MBS  transmits a BP, say BP$_2$ (see Fig.~\ref{Est_Window}), in which it encodes the list of block numbers of step 1 for which collisions occurred in the corresponding slots of step 2. Suppose there are $R^{(m)}_T$ blocks in this list. 

In step 3, $(T-1)R^{(m)}_T$ slots are used (see Fig.~\ref{Est_Window}). For $i \in \{1,\ldots,R^{(m)}_T\}$,  $\mathscr{T}_2$ (respectively, $\mathscr{T}_3, \ldots, \mathscr{T}_T$) active nodes corresponding to the $i^{th}$ block in the above list transmit symbol $\beta$ in the $((i-1)(T-1)+1)^{th}$ (respectively, $((i-1)(T-1)+2)^{th}$, \ldots, $(i(T-1))^{th}$) slot of step 3. It is easy to see that for each $b \in \{1, \ldots, T\}$, at the end of step 3, the MBS  unambiguously knows the set of block numbers of step 1 in which $\mathscr{T}_b$  nodes transmitted. For each $b \in \{1, \ldots, T\}$ and $i \in \{1, \ldots, \ell\}$, $X^{(m)}(b,i)$ (defined in the second paragraph of \Sectref{EstScheme_SRCM}) is set to $1$ if at least one node of $\mathscr{T}_b$ transmitted in block $i$ of step 1, and to $0$ otherwise. 

\subsection{Phase 2 of the Heterogeneous \texorpdfstring{SRC$_M$-2}{} (HSRC-M2) Scheme}\label{Sec:HSRC-M2}
\begin{figure}
	\begin{center}
		\includegraphics[scale = 0.48]{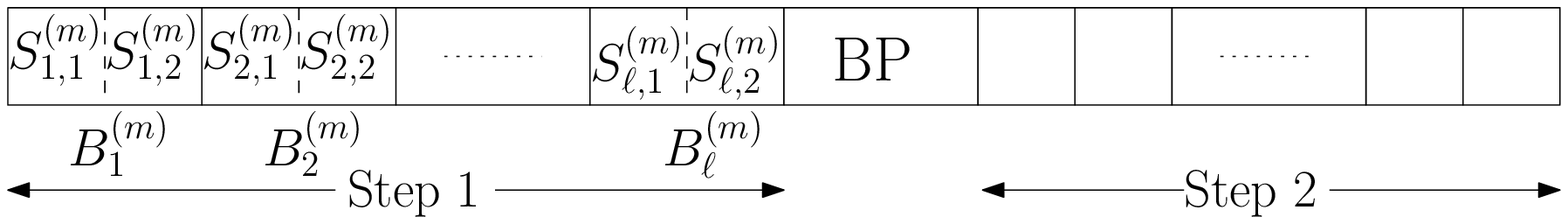}
		\caption{The figure shows the frame structure used in the 2-Step Protocol (2-SP) for the case $T=4$.} 
		\label{Est_Window2}
	\end{center}
\end{figure}

{\color{black}Recall that in phase 2 of HSRC-M2, 2-SP is used, which is described now.}
For $T=2$ and $T=3$, 2-SP is identical to 3-SP. For $T \ge 4$,  2-SP is a more sophisticated scheme than 3-SP and has only two steps.

\begin{figure}
\centering
\resizebox{0.96\columnwidth}{!}{\includegraphics{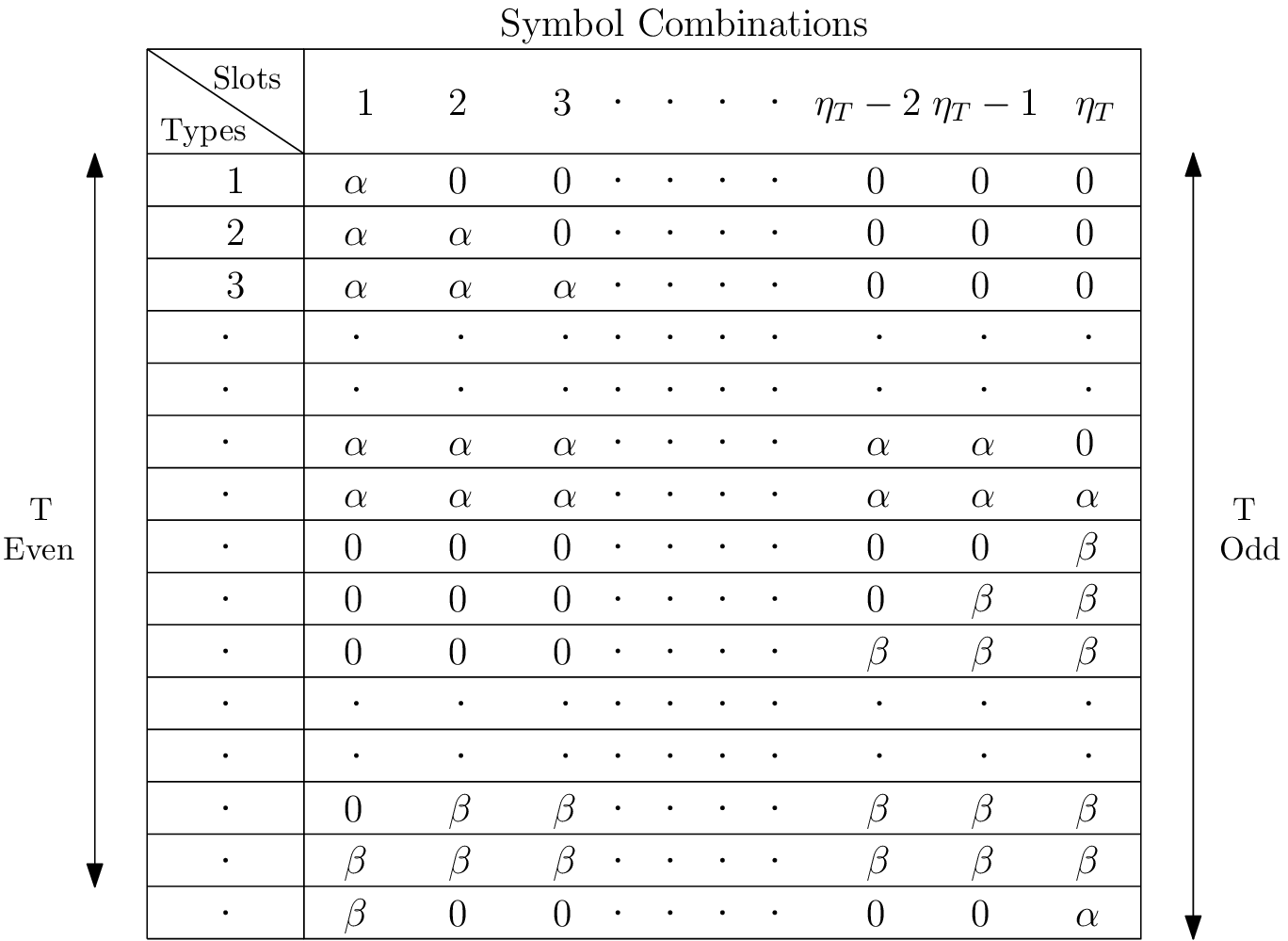}}
\caption{{ The figure shows the symbol combination used by each node type in HSRC-M2. $\eta_T = T/2$ if $T$ is even and $\eta_T = (T-1)/2$ if $T$ is odd. The symbol $0$ indicates ``no transmission''.}}
\label{Sym_Combo} 
\end{figure}

Fig.~\ref{Est_Window2} shows the 2-SP frame structure, which consists of two steps, for $T=4$.  In step 1, $\ell$ blocks, with $\eta_T$ slots each, where $\eta_T = T/2$ if $T$ is even and $\eta_T = (T-1)/2$ if $T$ is odd, are used. Hence, step 1 consists of $\tau_f(T) = (T \ell)/2$ slots if $T$ is even and $\tau_f(T) = ((T-1)\ell)/2$ slots if $T$ is odd. The number of slots in step 2, $\tau_s(T)$, is determined by the slot results in step 1. Each active node of each of the $T$ types independently chooses a block $B_i^{(m)}, \, i \in \{1, \ldots , \ell\}$, out of the $\ell$ blocks uniformly at random and transmits its corresponding symbol combination (see Fig.~\ref{Sym_Combo}) in that block with probability $2^{-I_{b,m}/\ell}$, where $I_{b,m}$ is obtained using~\eqref{EQ:SRCM:p_bm} and~\eqref{EQ:SRCM:I_bm}. Fig.~\ref{Sym_Combo} shows the symbol combination used by each node type.\footnote{For example, each active node of $\mathscr{T}_2$  and selected block $B_i^{(m)}$ transmits symbol $\alpha$ in slots 1 and 2 and does not transmit in slots $3, \ldots, \eta_T$ of block $B_i^{(m)}$. Each active node of $\mathscr{T}_3$  and selected block $B_i^{(m)}$ transmits symbol $\alpha$ in slots 1, 2, and 3, and does not transmit in slots $4, \ldots, \eta_T$ of block $B_i^{(m)}$.  For $T$ even, each active node of $\mathscr{T}_T$  and selected block $B_i^{(m)}$ transmits symbol $\beta$ in slots $1, \ldots, \eta_T$ of block $B_i^{(m)}$. For $T$ odd, each active node of $\mathscr{T}_T$  and selected block $B_i^{(m)}$ transmits symbol $\beta$ in slot 1, symbol $\alpha$ in slot $\eta_T$, and does not transmit in slots $2, \ldots, \eta_T -1$ of block $B_i^{(m)}$.} There are four possible outcomes in each slot: 0 (no transmission), $\alpha$, $\beta$, and $C$ (collision). The MBS keeps track of the results for each slot in each block. Let us define the bit patterns $X^{(m)}(b,i)$, $b \in \{1, \ldots T\}$, $i \in \{1, \ldots , \ell\}$, and $m \in \{1, \ldots, M\}$ as in the second paragraph of \Sectref{EstScheme_SRCM}.

\textcolor{black}{
Under HSRC-M2, the bit patterns $X^{(m)}(b,i)$ for the majority of $b$ and $i$ values are found in step 1; ambiguity about the rest remains, which is resolved in step 2. It is easy to verify that if no collision occurs in any of the slots of a block $B^{(m)}_i$, then the bit patterns $X^{(m)}(b,i)$, $b \in \{ 1, \ldots, T\}$, can be determined unambiguously.\footnote{For example: Let us consider the case $T=6$. From Fig.~\ref{Sym_Combo}, we get that the symbol combinations used are $(\alpha, 0, 0)$, $(\alpha, \alpha, 0)$, $(\alpha, \alpha, \alpha)$, $(0, 0, \beta)$, $(0, \beta, \beta)$, and $(\beta, \beta, \beta)$. If the slot results in step 1 of block $B^{(m)}_i$ are ($\alpha, 0,\beta$), it implies that exactly one node each from $\mathscr{T}_1$ (hence, $X^{(m)}(1,i)=1$) and $\mathscr{T}_4$  (hence, $X^{(m)}(4,i)=1$), and no nodes from $\mathscr{T}_2$  (hence, $X^{(m)}(2,i)=0$), $\mathscr{T}_3$  (hence, $X^{(m)}(3,i)=0$), $\mathscr{T}_5$  (hence, $X^{(m)}(5,i)=0$), and $\mathscr{T}_6$ (hence, $X^{(m)}(6,i)=0$) have transmitted.} If a collision occurs in one or more slots of a block $B^{(m)}_i$, then ambiguity about $X^{(m)}(b,i)$ for some values of $b$ remains, which is resolved in step 2 using transmissions of appropriate symbol combinations. At the end of step 1, the MBS sends a BP informing the active nodes of the block numbers $i$ and type numbers $b$ whose corresponding nodes must participate in step 2. The active nodes corresponding to those blocks $B^{(m)}_i$ and types $b$ for which there is still ambiguity about $X^{(m)}(b,i)$  after step 1 participate in step 2, whereas the others do not. If the  result  ($C, \ldots, C$) happens in a block $B^{(m)}_i$ of step 1, there is ambiguity about whether nodes with selected block number $i$ of \emph{all} $T$ types are active or inactive at the end of step 1. 
These node types are divided into two groups, with ambiguity addressed in the following manner: If $T$ is an even number, the node types are divided into two groups, each with a size of $T/2$. To address the ambiguity, each of the groups $\big(\{$Type 1, \ldots, Type $T/2\}$ and \{Type $T/2 + 1$, \ldots, Type $T\} \big)$ separately transmits symbol combinations in step 2 using the scheme of $T/2$ node types (see Fig.~\ref{Sym_Combo}) twice. If $T$ is an odd number, the node types are divided into two groups of sizes $(T+1)/2$ and $(T-1)/2$. To resolve the ambiguity, in step 2, each of the groups $\big(\{$Type 1, \ldots, Type $(T+1)/2\}$ and \{Type $(T+1)/2 + 1$, \ldots, Type $T\} \big)$ separately transmits symbol combinations using the schemes of $(T+1)/2$ and $(T-1)/2$ node types (see Fig.~\ref{Sym_Combo}), respectively.
Note that HSRC-M2 is \emph{recursive}-- if the result ($C, \ldots, C$) occurs in a block of step 1 while executing the scheme for $T$ node types, then the scheme for $T/2$, $(T+1)/2$ and/ or $(T-1)/2$ node types is employed to resolve the ambiguity in step 2.   
}
{\color{black}As examples, the estimation schemes used in 2-SP for $T=7$ and $T=8$ are provided in detail in Appendices~\ref{Apdx_2SP_T7} and~\ref{Apdx_2SP_T8}, respectively.}\footnote{For the scheme for $T=3$, refer to~\cite{kadam2017fast} and for the schemes for $T=4,~5,$ and $6$, refer to~\cite{kadam2020fast}.} 

A BP is sent by the MBS  after step 1 (see Fig.~\ref{Est_Window2}), which contains instructions that the active nodes should follow to resolve the remaining ambiguity, if any, in step 2. 
It can be checked that for each $b \in \{1, \ldots, T\}$, at the end of step 2, the MBS  unambiguously knows the set of block numbers of step 1 in which $\mathscr{T}_b$  nodes transmitted. For each $b \in \{1, \ldots, T\}$ and $i \in \{1, \ldots, \ell\}$, $X^{(m)}(b,i)$ is set to 1 if at least one node of $\mathscr{T}_b$ transmitted in block $i$ of step 1, and to 0 otherwise.

\subsection{Computation of Node Cardinality Estimates}\label{Subsec:ComputeEst}
After the end of 3-SP or 2-SP at stop $m \in \{1, \ldots, M\}$, the MBS knows the bit patterns $X^{(m)}(b,i)$, $\forall b \in \{1, \ldots, T\}$, $\forall i \in \{1, \ldots, \ell\}$, of all the types of  nodes. Let $z_b$ be the number of zeros in $X(b,i) = \lor_{m=1}^{M} X^{(m)}(b,i)$, $i \in \{1, \ldots, \ell\}$, at the last stop $M$ of the MBS. Then for each $b \in \{1, \ldots, T\}$, the final estimate generated by the protocol is (see~\eqref{EQ:SRCM:hatn})\footnote{If $z_b = 0$, then $\hat{n}_b$ is set to be an arbitrary integer.}:
\begin{equation}
\label{EQ:SRCM:hatn2}
\hat{n}_b = \frac{\ln({z_b/\ell})}{\ln{(1-{p_{b,M}/\ell})}}.
\end{equation}

\begin{theorem}\label{Thm_Accuracy}
The final node cardinality estimate, $\hat{n}_{b}$, of each type $b \in \{1, \ldots, T\}$, obtained using any one of the proposed schemes, viz., HSRC-M1 and HSRC-M2, \emph{equals, and hence is as accurate as,} the estimate that would have been obtained if phases 1 and 2 of the SRC$_M$  protocol were separately executed $T$ times at each stop $m \in \{1, \ldots, M\}$, to estimate the number of active nodes of each type $b \in \{1, \ldots, T\}$.
\end{theorem} 
\begin{proof}
Theorem~\ref{Thm_Accuracy} follows from the fact that, in both the proposed schemes, viz., HSRC-M1 and HSRC-M2, we  compute the same bit pattern $X^{(m)}(b,i)$, $b \in \{1, \ldots, T\}$, $i \in \{1, \ldots, \ell\}$, $m \in \{1, \ldots, M\}$, but using fewer time slots, as computed by the SRC$_M$ protocol if it is separately executed $T$ times (see Sections~\ref{Sec:HSRC-M1} and~\ref{Sec:HSRC-M2}).
\end{proof}

\section{Performance Analysis}\label{Analysis}  
In the following two subsections, the expected number of time slots required by HSRC-M1  to execute and the expected energy consumption of a node under HSRC-M1 are mathematically analyzed.

\subsection{Expected Number of Slots Required in  HSRC-M1}\label{p2_3stage}
Recall from Section~\ref{EstScheme_SRCM} that the number of slots required in phase 1 of HSRC-M1, at stop $m \in \{1, \ldots, M\}$, is:
\begin{equation} \label{eq:slots_phase1}
  \Delta^{(m)}_1 = tWT.  
\end{equation}
Now we compute the expected number of slots required in phase 2. 

At a given stop $m \in \{1, \ldots, M\}$, the number of slots required in step 1 is $(T-1)\ell$ (see Section~\ref{Sec:HSRC-M1}). Let $K^{(m)}_T$ (respectively, $(T-1)R^{(m)}_T$) be the number of slots required in step 2 (respectively, step 3). Let ${S}_{h,1}^{(m)}$ (respectively, ${S}_{h,2}^{(m)}, \ldots, {S}_{h,T-1}^{(m)}$), $h \in \{1, \ldots, \ell\}$, represent the result (collision, success, or empty slot) of the first (respectively, second, \ldots, $(T-1)^{th}$) slot of block $B_h^{(m)}$ of step 1. Also, let {I}$_{\upsilon}$ denote the indicator random variable corresponding to event $\upsilon$, i.e., {I}$_{\upsilon}$ is 1 if $\upsilon$ occurs, else it is 0.
 
From Section~\ref{Sec:HSRC-M1}, it is easy to see that  $K^{(m)}_T = \Sigma^{\ell}_{{h} = 1} {I}_{\{{S}_{h,1}^{(m)} = C, \ldots, {S}_{h,T-1}^{(m)} = C\} }$, where $C$ denotes collision. So:
\begin{equation}
E(K^{(m)}_T) = \sum^{\ell}_{{h} = 1} {P}({S}_{h,1}^{(m)} = C, \ldots, {S}_{h,(T-1)}^{(m)} = C).
\label{eq1}
\end{equation}

{\color{black}Under the following conditions, collisions occur in all $(T-1)$ slots of block $B_h^{(m)}$, $h \in \{1, \ldots, \ell\}$, $m \in \{1, \ldots, M\}$: 
\begin{itemize}[label={}]
    \item $E1$: At least two nodes of $\mathscr{T}_1$ transmit in block $B_h^{(m)}$.
    \item $E2$: Exactly one node of $\mathscr{T}_1$ and at least one node each of $\mathscr{T}_2, \ldots, \mathscr{T}_T$ transmit in block $B_h^{(m)}$.
    \item $E3$: At least two nodes each of $\mathscr{T}_2, \ldots, \mathscr{T}_T$ and none of $\mathscr{T}_1$ transmit in block $B_h^{(m)}$.
    \end{itemize}
    
    Let $Q^{(m)}_1(h)$, $Q^{(m)}_2(h)$, and $Q^{(m)}_3(h)$ denote the probabilities of the events $E1$, $E2$, and $E3$, respectively.} Since the probability of selecting  a block $B_h^{(m)}$ by a node of a given $\mathscr{T}_b$ is the same for all the blocks $B_h^{(m)}$ irrespective of $h$, we can write: $Q^{(m)}_{j}(h) = Q^{(m)}_{j}$, $j \in\{1,2, 3\}, \ h \in \{1, \ldots, \ell\}$.  
Hence:
\begin{equation}
{P}({S}_{h,1}^{(m)} = C, \ldots, {S}_{h,(T-1)}^{(m)} = C) = Q^{(m)}_{1}+Q^{(m)}_{2}+Q^{(m)}_{3}.
\label{equ2}
\end{equation}
Also\footnote{{Recall from Section~\ref{nwmodel_SRCM} that $\bar{n}^{(m)}_b$ is the number of active nodes present in the coverage range of the MBS at stop $m$.}}:
\begin{subequations}
\begin{align} 
Q^{(m)}_1 & = 1 - u\left(\bar{n}^{(m)}_1\right) - v\left(\bar{n}^{(m)}_1\right),
\label{eqQ1}\\
Q^{(m)}_2  & = v\left(\bar{n}^{(m)}_1\right) \prod_{b=2}^{T}\left(1 -  u\left(\bar{n}^{(m)}_b\right)\right),
\label{eqQ2}\\
Q^{(m)}_3  & = u\left(\bar{n}^{(m)}_1\right) \prod_{b=2}^{T}\left(1 - u\left(\bar{n}^{(m)}_b\right) - v\left(\bar{n}^{(m)}_b\right)\right),
\label{eqQ3}
\end{align}
\end{subequations}
where $u(\bar{n}^{(m)}_b), b \in \{1, \ldots, T\}$, is the probability that none of the nodes out of the $\bar{n}^{(m)}_b$ nodes of $\mathscr{T}_b$ select a given block and $v{(\bar{n}^{(m)}_b)}$ is the probability that exactly one node out of the $\bar{n}^{(m)}_b$ nodes of $\mathscr{T}_b$ selects a given block. 
So: 
\begin{subequations}
\begin{align} 
u\left(\bar{n}^{(m)}_b\right) & =  \left(1 - \frac{2^{-I_{b,m}}}\ell \right)^{\bar{n}^{(m)}_b},
\label{equnb}\\
v\left(\bar{n}^{(m)}_b\right) & = {\bar{n}^{(m)}_b} \left(\frac{2^{-I_{b,m}}}\ell\right) \left( 1 - \frac{2^{-I_{b,m}}}\ell \right)^{{\bar{n}^{(m)}_b} -1},
\label{eqvnb}
\end{align} 
\end{subequations}
where $I_{b,m}  =  \argmin_{j \in \{1, 2, 3, .....\}} |2^{-j} - p_{b,m} |$ (see~\eqref{EQ:SRCM:I_bm}) and $p_{b,m}  = \min{(1, 1.6 \ell/\tilde{n}^{(m)}_{b})}$ (see~\eqref{EQ:SRCM:p_bm}). 
By~\eqref{eq1} and~\eqref{equ2}:
\begin{equation}
\label{EQ:EK}
E(K^{(m)}_T) = \ell (Q^{(m)}_1 + Q^{(m)}_2+ Q^{(m)}_3). 
\end{equation}
Also: 
\begin{equation}
\label{EQ:ER}
E(R^{(m)}_T) = \ell Q^{(m)}_1,  
\end{equation}
since in step 3, only those nodes of $\mathscr{T}_2, \ldots, \mathscr{T}_T$ transmit, for which collisions occurred  in all the slots of the corresponding blocks of step 1 due to two or more $\mathscr{T}_1$ nodes transmitting (see Section~\ref{Sec:HSRC-M1}). So the expected  number of slots required in phase 2 of HSRC-M1 at stop $m$ is:
\begin{equation}
\label{EQ:3-SP}
    \Delta^{(m)}_2 = (T-1) \ell + E(K^{(m)}_T) + (T-1) E(R^{(m)}_T)+E(Z^{(m)}_{BP}),
\end{equation}
where $Z^{(m)}_{BP}$ is the number of slots required by the broadcast packets BP$^{(m)}_1$ and BP$^{(m)}_2$, $m \in \{1, \ldots, M\}$ (see Fig.~\ref{Est_Window}). It is easy to see that:
\begin{equation}\label{EQ:EZ}
    E(Z^{(m)}_{BP}) = \left\lceil{\frac{\ell}{S_W}}\right\rceil + E \left( \left\lceil{\frac{K^{(m)}_T}{S_W}}\right\rceil \right),
\end{equation}
where $S_W$ denotes the slot width in bits. 
By \eqref{eq:slots_phase1} and \eqref{EQ:3-SP}, the total expected number of slots required under HSRC-M1 is:
\begin{subequations}
\begin{align}
\label{eq:slots_total}
    \Delta &= \sum_{m=1}^M \Delta^{(m)} = \left(\sum_{m=1}^M \left(\Delta^{(m)}_1 + \Delta^{(m)}_2\right)\right)\\
    &= \sum_{m=1}^M  \left( tWT + (T-1) (\ell + E(R^{(m)}_T)) \right.\nonumber\\
    &\left.\quad \quad \quad+ E(K^{(m)}_T) + E(Z^{(m)}_{BP})\right), 
\end{align}
\end{subequations}
where $\Delta^{(m)} = \Delta_1^{(m)} + \Delta_2^{(m)}$ is the number of slots required at stop $m$ and $E(K^{(m)}_T)$, $E(R^{(m)}_T)$, and $E(Z^{(m)}_{BP})$ are given by~\eqref{EQ:EK},~\eqref{EQ:ER}, and~\eqref{EQ:EZ}, respectively. 

\subsection{Expected Energy Consumption of a Node under  HSRC-M1}\label{subsec_energy}
Let $\omega_b$ be a node of $\mathscr{T}_b, ~b\in \{1, \ldots, T\}$. Let us assume that it is active with probability $\xi_b^{(m)} \in [0,1]$ and inactive with probability $1-\xi_b^{(m)}$ when the MBS is at stop $m, ~m\in \{1, \ldots, M\}$. Let $\psi_b^{(m)}$ be $1$ when the node $\omega_b$ is in the range of the MBS when it is at stop $m$ and $0$ otherwise.  Let $\bar{\gamma}_b$ and $\hat{\gamma}_b$ be the energy spent per slot by a node of $\mathscr{T}_b, \, b\in \{1, \ldots, T\}$,  in the idle state and  reception state, respectively. Also, let $\gamma^{\alpha}_1$ (respectively, $\gamma^{\beta}_b$) be the energy spent per slot by a node of $\mathscr{T}_1$ (respectively, $\mathscr{T}_b, b\in \{2, \ldots, T\}$) for transmitting the symbol $\alpha$ (respectively, symbol $\beta$) and $\gamma '$ be the energy spent by a node of any type for transmitting a signal in the transmission state in a slot in phase 1. Let $\Gamma^{(m)}_{b,\tau}(\Phi_i),\Gamma^{(m)}_{b,\iota}(\Phi_i),$ and $\Gamma^{(m)}_{b,\rho}(\Phi_i)$ be the expected energy consumption of a $\mathscr{T}_b, b\in \{1, \ldots, T\}$, node, in phase $i, i\in \{1,2\}$, in the transmission state, idle state, and reception state, respectively, at a given stop $m \in \{1, \ldots, M\}$.

\subsubsection{Expected Energy Consumption of a Node in Phase 1 at a Given Stop}\label{SubsubsecP1}
Recall from Section~\ref{EstScheme_SRCM} and Section~\ref{SubSec_SRCM_P1} that during phase 1, due to the separate execution of $W$ trials of phase 1 of the SRC$_M$ protocol, all types of nodes spend equal amounts of expected energy in the transmission state, idle state, and reception state at a given stop $m$; also, for any node of $\mathscr{T}_b, \, b\in \{1, \ldots, T\}$:
\begin{subequations}
\begin{align}
    \Gamma^{(m)}_{b,\tau}(\Phi_1) &= \xi_b^{(m)} \psi_b^{(m)} W \gamma ',\label{eq:T1P1Tra}\\
    \Gamma^{(m)}_{b,\iota}(\Phi_1) &= \left( \!(t-\!1) \xi_b^{(m)} \psi_b^{(m)} + t (1-\xi_b^{(m)} \psi_b^{(m)})\!\right)\! W \bar{\gamma}_b,\label{eq:T1P1Idl}\\
    \Gamma^{(m)}_{b,\rho}(\Phi_1) &= 0.\label{eq:T1P1Rec}
\end{align}
\end{subequations}

\subsubsection{Expected Energy Consumption of a Node in Phase 2 at a Given Stop} \label{SubsubsecP2}
Recall from Section~\ref{Sec:HSRC-M1} that a $\mathscr{T}_b, \, b \in \{1, \ldots, T\}$, node chooses a block out of $\ell$ blocks uniformly at random and transmits in that block with probability $2^{-I_{b,m}}/\ell$. Let 
\begin{equation}
    \pi_b^{(m)} = \xi_b^{(m)} \psi_b^{(m)} 2^{-I_{b,m}}, \, b \in \{1, \ldots, T\}, \, m \in \{1, \ldots, M\}.
\end{equation}
\paragraph{\texorpdfstring{$\mathscr{T}_1$}{} Node:}
Recall from Section~\ref{Sec:HSRC-M1} that a $\mathscr{T}_1$ node, $\omega_1$, participates in step 1 and step 2. In step 1, $\omega_1$ transmits the symbol $\alpha$ in all $(T-1)$ slots of the chosen block $h, \, h \in \{1, \ldots, \ell\}$. So the expected energy consumption in the transmission state at step 1 is: $\pi_1^{(m)}  (T-1) \gamma^{\alpha}_1.$ 
In step 2, node $\omega_1$ participates only if any of the following two events have occurred in step 1: (a) at least one node of $\mathscr{T}_1$, other than $\omega_1$, has transmitted in block $h$, (b) no node of $\mathscr{T}_1$ and at least one node each of $\mathscr{T}_b, \, b \in \{2, \ldots, T\}$, have transmitted. Let $\bar{Q}_1^{(m)}$ and $\bar{Q}_2^{(m)}$ be the probabilities of the events (a) and (b), respectively. The values of $\bar{Q}_1^{(m)}$ and $\bar{Q}_2^{(m)}$ are:
\begin{subequations}
\begin{align}
    \bar{Q}_1^{(m)} &= 1 - u\left(\bar{n}^{(m)}_1-1\right),\\
    \bar{Q}_2^{(m)} &= u\left(\bar{n}^{(m)}_1-1\right) \prod_{b=2}^T \left(1 - u\left(\bar{n}^{(m)}_b\right)\right).
\end{align}
\end{subequations}
So the expected energy consumption in the transmission state at step 2 is: $\pi_1^{(m)} (\bar{Q}_1^{(m)} + \bar{Q}_2^{(m)}) \gamma^{\alpha}_1.$ 
Hence, the expected energy consumption of node $\omega_1$ in the transmission state is:
\begin{subequations}
\begin{align}
  \Gamma^{(m)}_{1,\tau} (\Phi_2)  &= \pi_1^{(m)} (T-1) \gamma^{\alpha}_1 +  \pi_1^{(m)} (\bar{Q}_1^{(m)} + \bar{Q}_2^{(m)}) \gamma^{\alpha}_1,\\
    &= \left(T-1+ \bar{Q}_1^{(m)} + \bar{Q}_2^{(m)}\right) \pi_1^{(m)} \gamma^{\alpha}_1. \label{eq:T1P2Tra}
\end{align}
\end{subequations}
In the reception state, recall from Section~\ref{Sec:HSRC-M1} that node $\omega_1$ reads all the slots of BP$_1$ in phase 2. So the expected energy consumption in the reception state is:
\begin{equation}
    \Gamma^{(m)}_{1,\rho}(\Phi_2) = \pi_1^{(m)} \left\lceil{\frac{\ell}{S_W}}\right\rceil \hat{\gamma}_1. \label{eq:T1P2Rec}
\end{equation}
In the rest of the slots of phase 2, node $\omega_1$ is in the idle state. The expected energy consumption in this state is:
\begin{align}
    \!\!\!\!\!\!\!\!\!\!\!\!\!\Gamma^{(m)}_{1,\iota}(\Phi_2) &=  \bar{\gamma}_1\left[\left(1-\pi_1^{(m)}\right)\Delta^{(m)}_2 + \pi_1^{(m)} \times \right. \nonumber \\
    &\left.  
    \left\{\Delta^{(m)}_2 - \left(T-1 + \left\lceil{\frac{\ell}{S_W}}\right\rceil + \bar{Q}_1^{(m)} + \bar{Q}_2^{(m)}\right)\right\} \right]. \label{eq:T1P2Idl}
\end{align}
The total expected energy consumption of node $\omega_1$ is obtained by adding the quantities in~\eqref{eq:T1P1Tra}, \eqref{eq:T1P1Idl}, \eqref{eq:T1P1Rec}, \eqref{eq:T1P2Tra}, \eqref{eq:T1P2Rec}, and \eqref{eq:T1P2Idl}.

\paragraph{\texorpdfstring{$\mathscr{T}_b, \, b\in \{2, \ldots, T\}$}{}$,$ Node:}
Recall from Section~\ref{Sec:HSRC-M1} that a $\mathscr{T}_b$ node, $\omega_b$, participates in step 1 and step 3. In step 1, $\omega_b$ transmits the symbol $\beta$ in one slot of the chosen block $h, \, h \in \{1, \ldots, \ell\}$. So the expected energy consumption in the transmission state at step 1 is: $\pi_b^{(m)} \gamma^{\beta}_b$.
In step 3, node $\omega_b$ participates only if the following event has occurred in step 1: (c) at least two nodes of $\mathscr{T}_1$ have transmitted in block $h$.  The probability of the event (c) is $Q^{(m)}_1$ in~\eqref{eqQ1}. So the expected energy consumption in the transmission state at step 3 is: $Q^{(m)}_1 \pi_b^{(m)} \gamma^{\beta}_b$. Hence, the total expected energy consumption in the transmission state is:
\begin{equation}
  \Gamma^{(m)}_{b,\tau}(\Phi_2) =  \left(1+Q^{(m)}_1\right)\pi_b^{(m)} \gamma^{\beta}_b. \label{eq:TbP2Tra}
\end{equation}
Suppose $\delta_h \in \{0, 1\}, \, h \in \{1, \ldots, \ell\}$, denotes the value of the $h^{th}$ bit in BP$_1$ (see Fig.~\ref{Est_Window}). Recall from Section~\ref{Sec:HSRC-M1} that node $\omega_b$ reads BP$_1$ and reads its corresponding slot of BP$_2$ only if $\delta_h$ is 1. The events that lead to $\delta_h=1$ are: (c), (d) exactly one node from $\mathscr{T}_1$ transmits in block $h$ and at least one node from each of $\mathscr{T}_2,\ldots, \mathscr{T}_{b-1}, \mathscr{T}_{b+1},..., \mathscr{T}_T$ transmits in the same block $h$, and (e) at least two nodes each from $\mathscr{T}_2,\ldots, \mathscr{T}_{b-1}, \mathscr{T}_{b+1},..., \mathscr{T}_T$ transmit in block $h$, at least one node from $\mathscr{T}_b \setminus \omega_b$ transmits in block $h$ and no node from $\mathscr{T}_1$ transmits in block $h$. Let $\hat{Q}_1^{(m)}$ and $\hat{Q}_2^{(m)}$ be the probabilities of the events (d) and (e), respectively. The values of $\hat{Q}_1^{(m)}$ and $\hat{Q}_2^{(m)}$ are:
\begin{subequations}
\begin{align}
    \hat{Q}_1^{(m)} &=  v\left(\bar{n}^{(m)}_1\right) \prod_{\substack{{i=2}\\i \neq b}}^T \left(1 - u\left(\bar{n}^{(m)}_i\right)\right),\\
    \hat{Q}_2^{(m)} &= u\left(\bar{n}^{(m)}_1\right)\left(1-u\left(\bar{n}^{(m)}_b-1\right)\right) \nonumber \\
    &~~\prod_{\substack{{i=2}\\i \neq b}}^T \left(1 - u\left(\bar{n}^{(m)}_i\right) - v\left(\bar{n}^{(m)}_i \right)\right).
\end{align}
\end{subequations}
Hence, the expected energy consumption in the reception state is:
\begin{equation}
  \Gamma^{(m)}_{b,\rho}(\Phi_2) =  \left(\left\lceil{\frac{\ell}{S_W}}\right\rceil+Q^{(m)}_1+\hat{Q}_1^{(m)}+\hat{Q}_2^{(m)}\right)\pi_b^{(m)} \hat{\gamma}_b. \label{eq:TbP2Rec}
\end{equation}
In the rest of the slots of phase 2, node $\omega_b$ is in the idle state. The expected energy consumption in this state is:
\begin{align}
    \!\!\!\!\!\!\!\!\!\!\!\!\!\!\Gamma^{(m)}_{b,\iota}(\Phi_2) &=  \bar{\gamma}_b\left[\left(1-\pi_b^{(m)}\right)\Delta^{(m)}_2 + \pi_b^{(m)} \times \right. \nonumber \\
    &\left.  
    \left\{\Delta^{(m)}_2 \!-\! \left(\!1\!+ \!\!\left\lceil{\frac{\ell}{S_W}}\right\rceil\!\! + 2Q^{(m)}_1+\hat{Q}_1^{(m)}+\hat{Q}_2^{(m)}\!\right)\!\!\right\} \right] . \label{eq:TbP2Idl}
\end{align}
The total expected energy consumption of node $\omega_b$, $b\in \{2, \ldots, T\}$, is obtained by adding the quantities in~\eqref{eq:T1P1Tra},~\eqref{eq:T1P1Idl},~\eqref{eq:T1P1Rec},~\eqref{eq:TbP2Tra},~\eqref{eq:TbP2Rec}, and~\eqref{eq:TbP2Idl}.

\subsubsection{Total Expected Energy Consumption}
Now we add the expected energy consumption by a node of $\mathscr{T}_b, \, b \in \{1, \ldots, T\}$, computed in Sections~\ref{SubsubsecP1} and~\ref{SubsubsecP2} to obtain the total expected energy consumption by a node at stop $m$:
\begin{equation}
    \Gamma^{(m)}_{b} = \sum_{i=1}^2 \left( \Gamma^{(m)}_{b,\tau}(\Phi_i) + \Gamma^{(m)}_{b,\iota}(\Phi_i) + \Gamma^{(m)}_{b,\rho}(\Phi_i) \right).
\end{equation}

\section{Optimal MBS Tour (OMT) Problem}\label{OMT}
Let $\mathcal{G} = (\mathcal{M},\mathcal{E})$ be a connected directed graph where $\mathcal{M}$ denotes the given set of possible stops for the MBS in a given region, $\mathcal{E}$ denotes the set of links (edges) between these stops, and $|\mathcal{M}| = \tilde{M}+1, |\mathcal{E}| = E$. Also, let $\mathcal{N}$ be the set of nodes in the region, and $|\mathcal{N}| = N$. Let $c_{u,v}$ and $\eta_{k,m}$ denote the cost of travelling from stop $u$ to stop $v$ via the link from $u$ to $v$ and the energy spent by node $k$ during the estimation process when it is in the range of the MBS at stop $m$, respectively. Note that $c_{u,v}$ is defined to be $\infty$ if there is no link from stop $u$ to stop $v$.  Let $X_{k,m}\in~\{0,1\}$ be $1$ if node $k$ is in the coverage range of the MBS when it is at stop $m$, and zero else. That is, $X_{k,m}$, $\forall k \in \mathcal{N}, \forall m \in \mathcal{M}$, can be defined as:
\begin{align}
    X_{k,m} := 
    \begin{cases}
    1, &\mbox{if node $k$ is in the coverage range}\\ &\mbox{of the MBS when it is at stop $m$,} \\
    0, &\mbox{otherwise}.
    \end{cases}
\end{align}
Let $\mathcal{R}$ denote the set of real numbers.

The MBS starts from a charging station, viz.,  stop $0 \in \mathcal{M}$\footnote{We assume that there are no nodes in the coverage range of the MBS when it is at the charging station. So $X_{k,0} = 0, \forall k \in \mathcal{N}$.}, and visits multiple stops for the node cardinality estimation process. Once it is completed, the MBS travels back to stop $0$. {\color{black}The MBS is assumed to visit a given stop (other than $0$) at most once during its tour.} 
Let $\hat{\mathcal{M}} \subseteq \mathcal{M}$ be the set of stops that belong to a given tour of the MBS and $|\hat{\mathcal{M}}| = \hat{M} + 1 \le \tilde{M} +1$. For $i \in \{0,1, \ldots, \hat{M}+1\}$, let $m_i \in \hat{\mathcal{M}}$ denote the $i^{th}$ stop that the MBS visits during its tour, where we define $m_0=m_{\hat{M}+1}=0$. A tour can be represented as $(m_i,m_{i+1})$, $i \in \{0,1, \ldots, \hat{M}\}$. 

{\color{black}Our optimization problem, which is called the Optimal MBS Tour (OMT) problem,} is to find a tour for the MBS, such that every node in $\mathcal{N}$ is in the union of the coverage ranges of the MBS when it is at the stops in $\hat{\mathcal{M}}$, with the minimum travel cost and the limit $\overline{\eta}$ on the total energy consumed by the nodes. Mathematically, it can be written as follows:
\begin{subequations}
\begin{align} 
\min_{\substack{\hat{M},~m_i, \\ i \in \{1,\ldots, \hat{M}\}}} & \sum_{i=0}^{\hat{M}} c_{m_i,m_{i+1}}   \label{eq:obj_fun}\\
\textrm{s.t.} \quad & \sum_{k \in \mathcal{N}} \sum_{i=1}^{\hat{M}}\eta_{k,m_i} X_{k,m_i} \le \overline{\eta}, \label{eq:energy_constraint}\\ 
\quad & \sum_{i=1}^{\hat{M}}X_{k,m_i} \ge 1, \quad \forall k \in \mathcal{N}, \label{eq:coverage_constraint} \\
\quad & m_0 = m_{\hat{M}+1} = 0, \label{eq:BeginEndStops_constraint} \\
\quad & m_i \neq m_j, \,  \forall i,j \in \{1, \ldots, \hat{M}\}, \, i \neq j, \label{eq:UniqueStop_constraint}\\
\quad & \hat{M} \in \{1, \ldots, \tilde{M}\}, \label{eq:Mstar_constraint}\\
\quad & X_{k,m} \in\{0, 1\},~ \overline{\eta}, \eta_{k,m}, c_{u,v} \in \mathcal{R},  \nonumber\\
\quad & \quad \forall m, u, v \in \{0,1, \ldots, \tilde{M}\}, \forall k \in \mathcal{N}. \nonumber
\end{align}
\end{subequations}
The objective of the OMT problem, which is stated in~\eqref{eq:obj_fun}, is to minimize the total travel cost of the MBS. The constraint in~\eqref{eq:energy_constraint} says that the total energy spent by the nodes during the estimation process should be upper bounded by $\overline{\eta}$. The constraint in~\eqref{eq:coverage_constraint} ensures that the MBS covers all nodes at least once during its tour. 
The constraints in~\eqref{eq:BeginEndStops_constraint} and~\eqref{eq:UniqueStop_constraint} ensure that the MBS's first and last stops are the charging station, and that the MBS visits each other stop in $\hat{\mathcal{M}}$ only once, respectively.  
The set of possible values that $\hat{M}$ can take is specified in~\eqref{eq:Mstar_constraint}. 

\begin{theorem} \label{Thm_OMT}
The OMT problem is NP-complete.
\end{theorem}
\begin{proof}
The proof is provided in Appendix~\ref{Apdx_Thm_OMT}.
\end{proof}

{\color{black}Theorem~\ref{Thm_OMT} shows that the OMT problem is NP-complete, and hence, a greedy heuristic algorithm is proposed to solve it. This algorithm is now described.}

We assume that in the graph $\mathcal{G}$, there is a link from every stop to every other stop. Let $\mathcal{Z}_{m_i}$ (respectively, $z_m$) be the set of nodes covered by the MBS until stop $m_i, \, i \in \{0,1, \ldots, \hat{M}\}$ (respectively, at stop $m$). Let $\overline{\eta}_{m_i}, \, i \in \{1, \ldots, \hat{M}\},$ denote the sum of $\overline{\eta}_{m_{i-1}}$ and the energy spent by the nodes of the set $z_{m_i}$ with $\overline{\eta}_{m_0} = 0$. Also, let $\mathcal{U}_i$ denote the set of stops that have not been visited by the MBS until stop $m_i, \, i \in \{0,1, \ldots, \hat{M}\}$. Since the primary objective of the OMT problem is to generate a low travel cost tour for the MBS, at each stop, the greedy algorithm attempts to choose the next stop with the lowest travel cost from the current stop. The MBS begins its tour at the charging station $m_0$.
First, it selects the stop $m$ from the set of stops $\mathcal{U}_0$ with the lowest value of $c_{m_0,m}$, travels there, and estimates the number of nodes. Following the estimation process, the MBS checks the number of nodes it has covered so far. If it has covered all of the nodes in the region, it returns to the charging station, and the tour found by the greedy algorithm will be the sequence of stops it has visited so far. If not, it checks the amount of energy spent by the nodes until that stop. If this exceeds the threshold $\overline{\eta}$, the MBS returns to the charging station. The tour found by the greedy algorithm in this scenario will be the sequence of stops visited so far, even though the MBS would not have covered all the nodes. If not, the MBS will repeat the above process starting from the current stop $m$ instead of stop $m_0$ and so on. The detailed algorithm is provided in Algorithm~\ref{alg:greedy}.

Note that the greedy algorithm would perform well when the bound $\overline{\eta}$ on the total energy consumed by the nodes is sufficiently large, since in that case it will be unlikely that the amount of energy spent by the nodes exceeds $\overline{\eta}$ even though the MBS has not covered all the nodes.

\begin{algorithm}
\caption{Greedy Algorithm}\label{alg:greedy}
\begin{algorithmic}
\State \textbf{Input:} $\mathcal{M}$, $\mathcal{N}$, $\overline{\eta}$, $c_{u,v},$ $u,v \in \mathcal{M}$, $u \neq v$, $\eta_{k,m}$, $k \in \mathcal{N}$, $m \in \mathcal{M}$
\State Initialize $i=0$, $m_0 = 0$, $\hat{\mathcal{M}} \!\gets \{0\}$, $\mathcal{Z}_{m_0} \!\!\gets \emptyset$, $\overline{\eta}_{m_0} \!\!\gets 0$, $\mathcal{U}_0\!\! \gets \mathcal{M} \backslash \{0\}$
\While {$\mathcal{Z}_{m_i} \neq \mathcal{N}$ AND $\overline{\eta}_{m_i} \le \overline{\eta}$}
\State Choose stop $m \in \mathcal{U}_i$ such that $c_{m_i,m}$ is smallest. 
\State $\mathcal{Z}_{m_{i+1}} \gets \mathcal{Z}_{m_i} \bigcup z_{m}$
\State $\overline{\eta}_{m_{i+1}} \gets \overline{\eta}_{m_i} + \sum_{j \in z_{m}}\eta_{j,m}$
\State $\mathcal{U}_{i+1} \gets \mathcal{U}_i \setminus \{m\}$
\State $m_{i+1} \gets m$; $\hat{\mathcal{M}} \gets \hat{\mathcal{M}} \bigcup \{m\}$
\State $i \gets i+1$
\EndWhile
\State $m_{i+1} \gets 0$
\State $\hat{M} \gets |\hat{\mathcal{M}}| - 1$
\State \textbf{Output:} $\hat{\mathcal{M}}$, $(m_i, m_{i+1}), i=0, \ldots, \hat{M}$.
\end{algorithmic}
\end{algorithm}

\section{Simulations}\label{Sec:Simu}
{\color{black}Simulations for the amounts of time required for estimation by the proposed protocols, viz., HSRC-M1 and HSRC-M2, as well as the $T$ repetitions of SRC$_M$ protocol, were performed using MATLAB version R2022a, with the Statistics and Machine Learning Toolbox as an add-on. The host machine had an Intel $i5$ $1.6$ GHz CPU with $12$ GB of DDR4 memory clocking at $2400$ MHz. All the source codes and plots are made available online in a  GitHub repository~\cite{gitlink}, so that they can be used for future work by the research community at large.}
\begin{figure}[tbp]
    \centering
\includegraphics[width=0.48\textwidth]{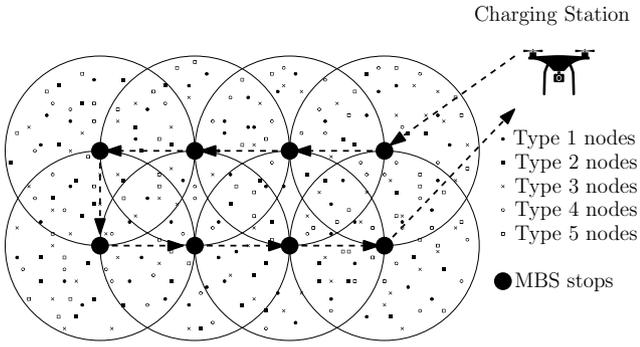}
    \caption{The figure shows $M = 8$ locations (stops) of a MBS and $T = 5$ types of nodes in a region. The coverage range of the MBS at a stop is the area  inside the circle with that stop as the centre.}
    \label{Network_Model_SRCM_rect}
\end{figure}

{\color{black}Simulations  of three scenarios are performed using the two network models shown in Fig.~\ref{Network_Model_SRCM} (Section~\ref{nwmodel_SRCM}) and Fig.~\ref{Network_Model_SRCM_rect}.} Network Model I (see Fig.~\ref{Network_Model_SRCM}) uses $M = 4$ and a square topology of dimensions $1 \times 1$ and coordinates of vertices $(0,0)$, $(1,0)$, $(1,1)$, and $(0,1)$, with the MBS making stops at the following locations: ($0.75, 0.75$), ($0.25, 0.75$), ($0.25, 0.25$), and ($0.75, 0.25$). Let $q_b$, $b \in \{1, \ldots, T\}$,  be the probability with which a given node of $\mathscr{T}_b$ is active.  In Scenario I, we work with Network Model I and assume that $|\mathscr{N}_1| = \ldots = |\mathscr{N}_T| = D$ (say) and $q_1 = \ldots = q_T =q$ (say).
For Scenario II, we again work with Network Model I and assume that for each $b \in \{1,\ldots, T\}$, $|\mathcal{N}_b|$ (respectively, $q_b$) is chosen randomly 
from a distribution with a mean $\overline{D}$ (respectively, $\overline{q}$). 
For Scenario III, we work with Network Model II, with $|\mathscr{N}_b|$ and $q_b$ chosen as in Scenario II. 
Network Model II (see Fig.~\ref{Network_Model_SRCM_rect}) uses $M=8$ and a rectangular topology of dimensions $2 \times 1$ and coordinates of vertices $(0,0)$, $(2,0)$, $(2,1)$, and $(0,1)$, with the MBS making stops at the following locations: ($1.75, 0.75$), ($1.25, 0.75$), ($0.75, 0.75$), ($0.25, 0.75$), ($0.25, 0.25$), ($0.75, 0.25$), ($1.25, 0.25$), and ($1.75, 0.25$). Nodes of all $T$ types are placed at locations that are chosen uniformly at random inside the square of Network Model I (Scenarios I and II) or rectangle of Network Model II (Scenario III). In all three scenarios, the coverage range of the MBS is $R = \pi /4$ units, and the locations of the MBS's stops are chosen in such a way that all the nodes of all $T$ types inside the square or rectangular topology are covered by the MBS at least once.


\begin{figure}[ht]
\begin{subfigure}[b]{0.5\linewidth}
\centering
\resizebox{1.0\columnwidth}{!}{\includegraphics{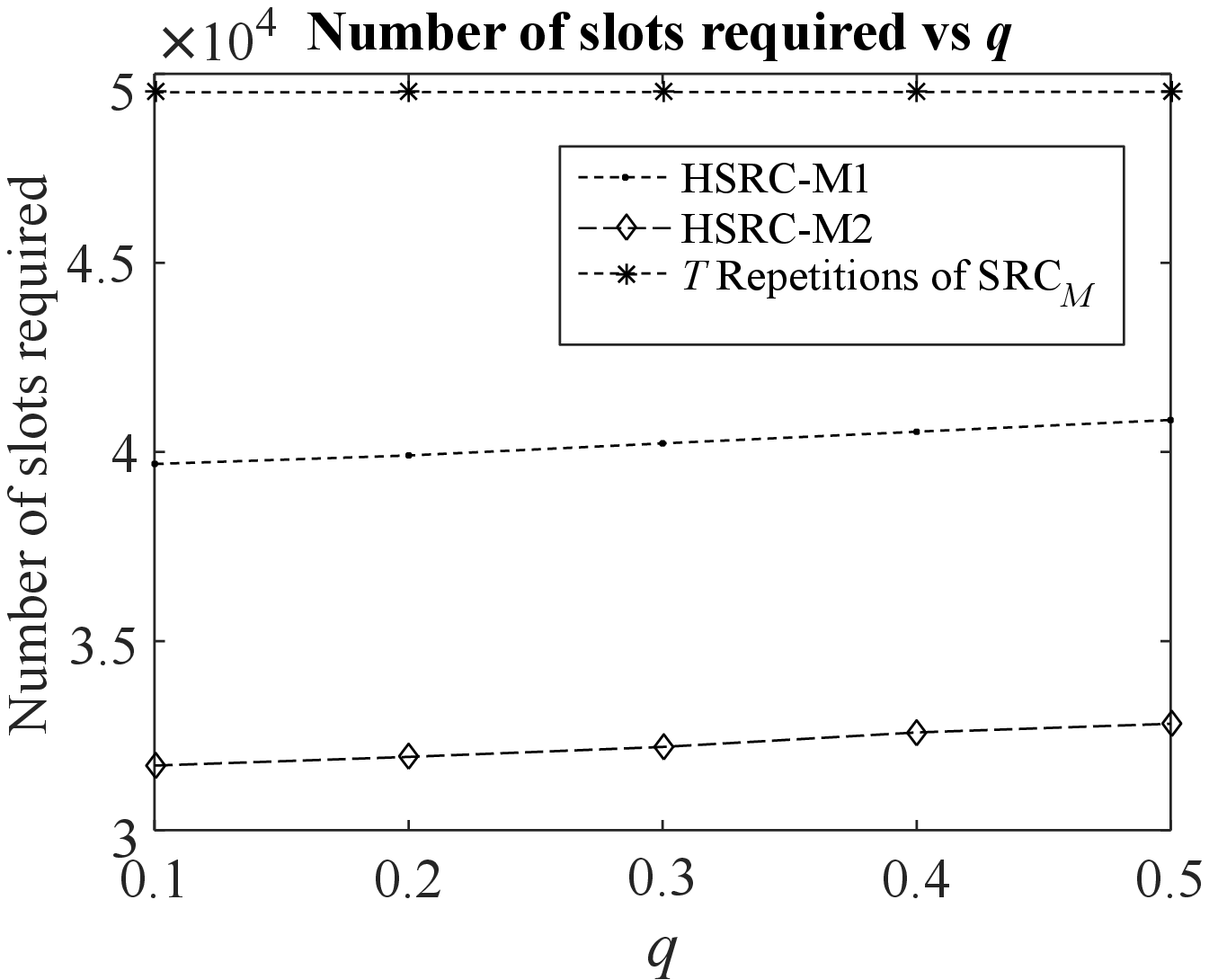}}
\caption{$D = 300$}
\label{Sim_q}
\end{subfigure}%
\begin{subfigure}[b]{0.5\linewidth}
\centering
\resizebox{1.0\columnwidth}{!}{\includegraphics{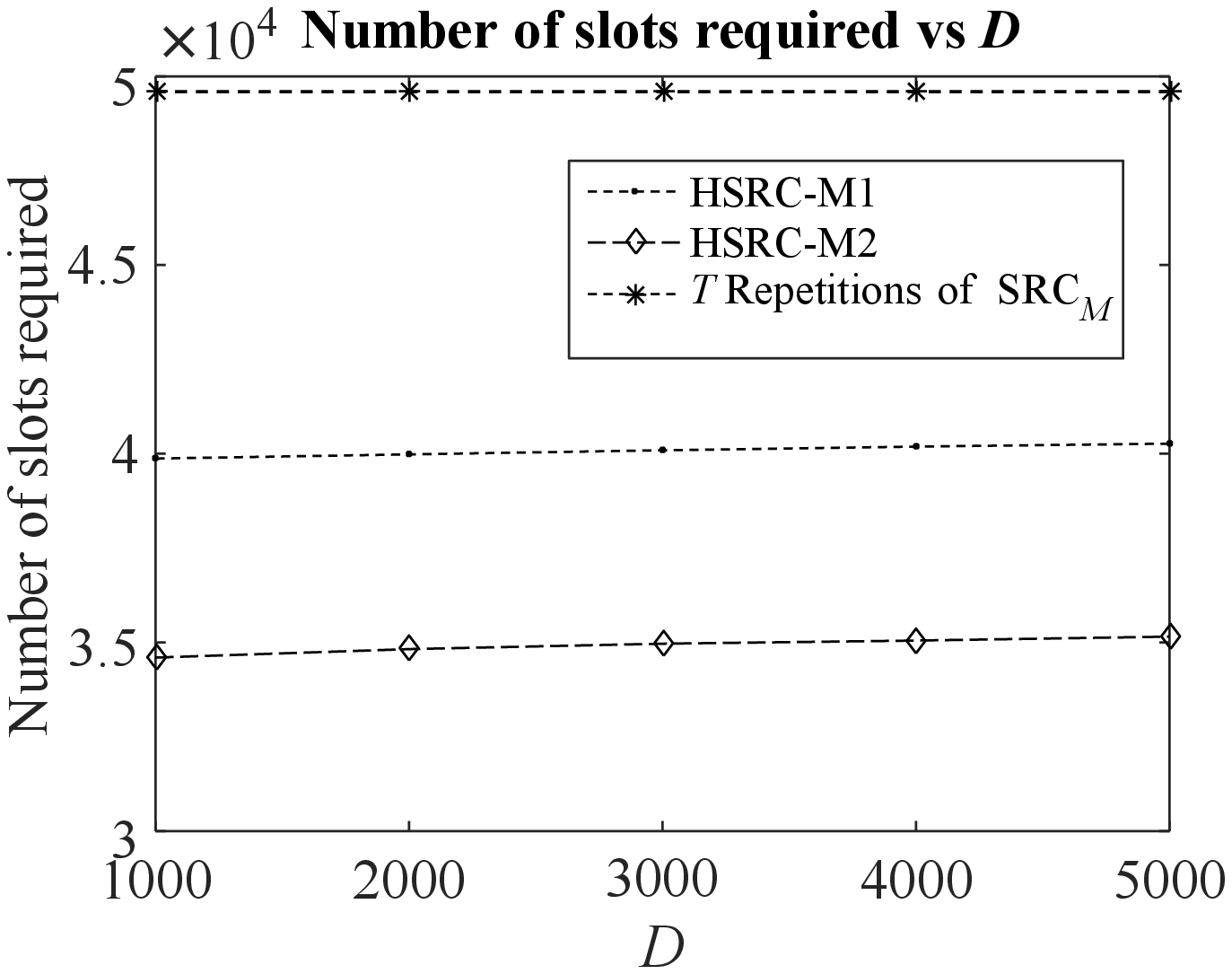}}
\caption{$q = 0.6$}
\label{Sim_D}
\end{subfigure}
\caption{These plots show the average numbers of slots required by various estimation schemes versus $q$ and $D$, respectively, for Scenario I. The following common parameter values are used in these plots: $T = 4$ and $\epsilon = 0.03$.}
\label{Sim6}
\end{figure}

\begin{figure}[ht]
\begin{subfigure}[b]{0.5\linewidth}
\centering
\resizebox{1.0\columnwidth}{!}{\includegraphics{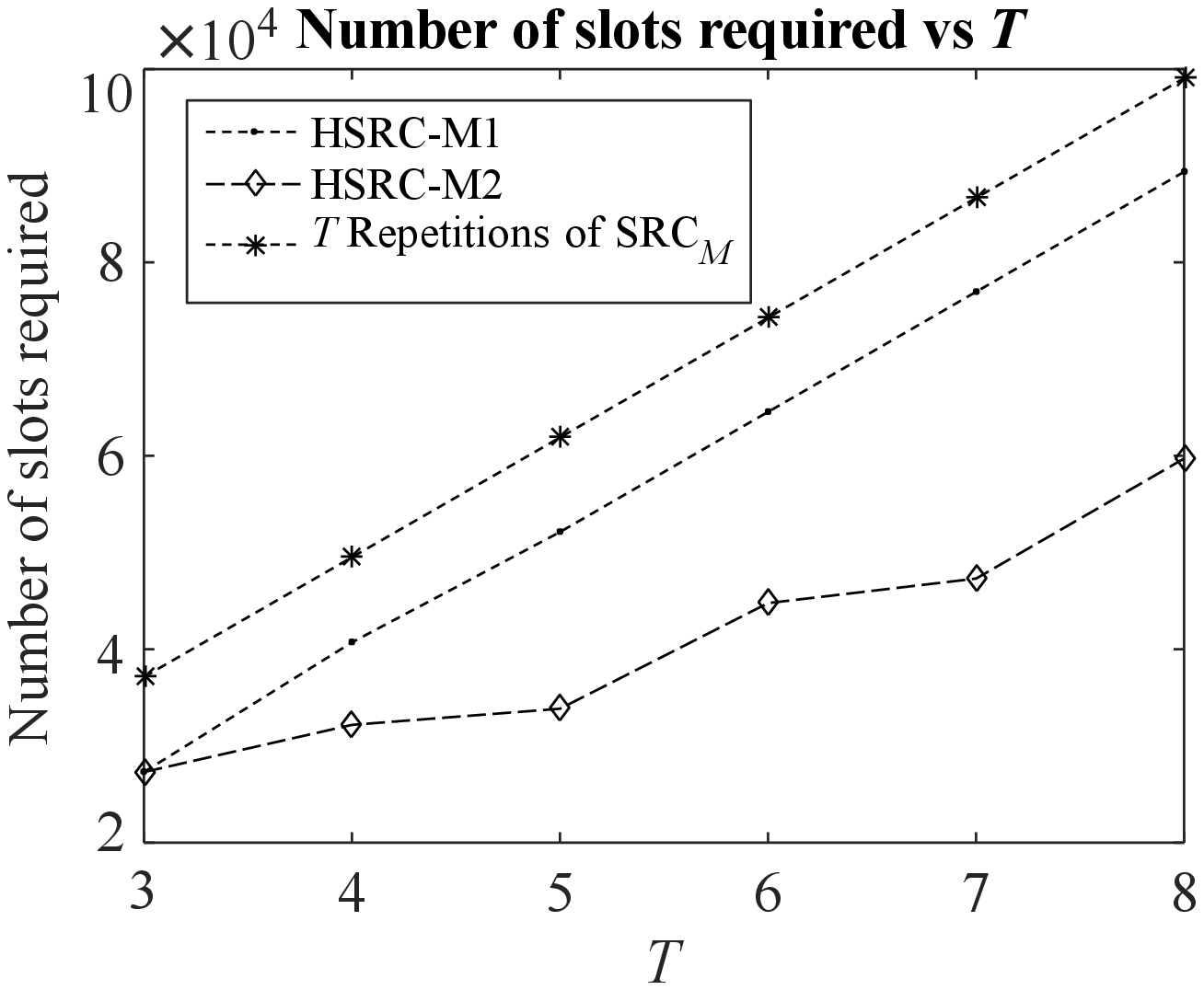}}
\caption{$\epsilon = 0.03$}
\label{Sim_T}
\end{subfigure}%
\begin{subfigure}[b]{0.5\linewidth}
\centering
\resizebox{1.0\columnwidth}{!}{\includegraphics{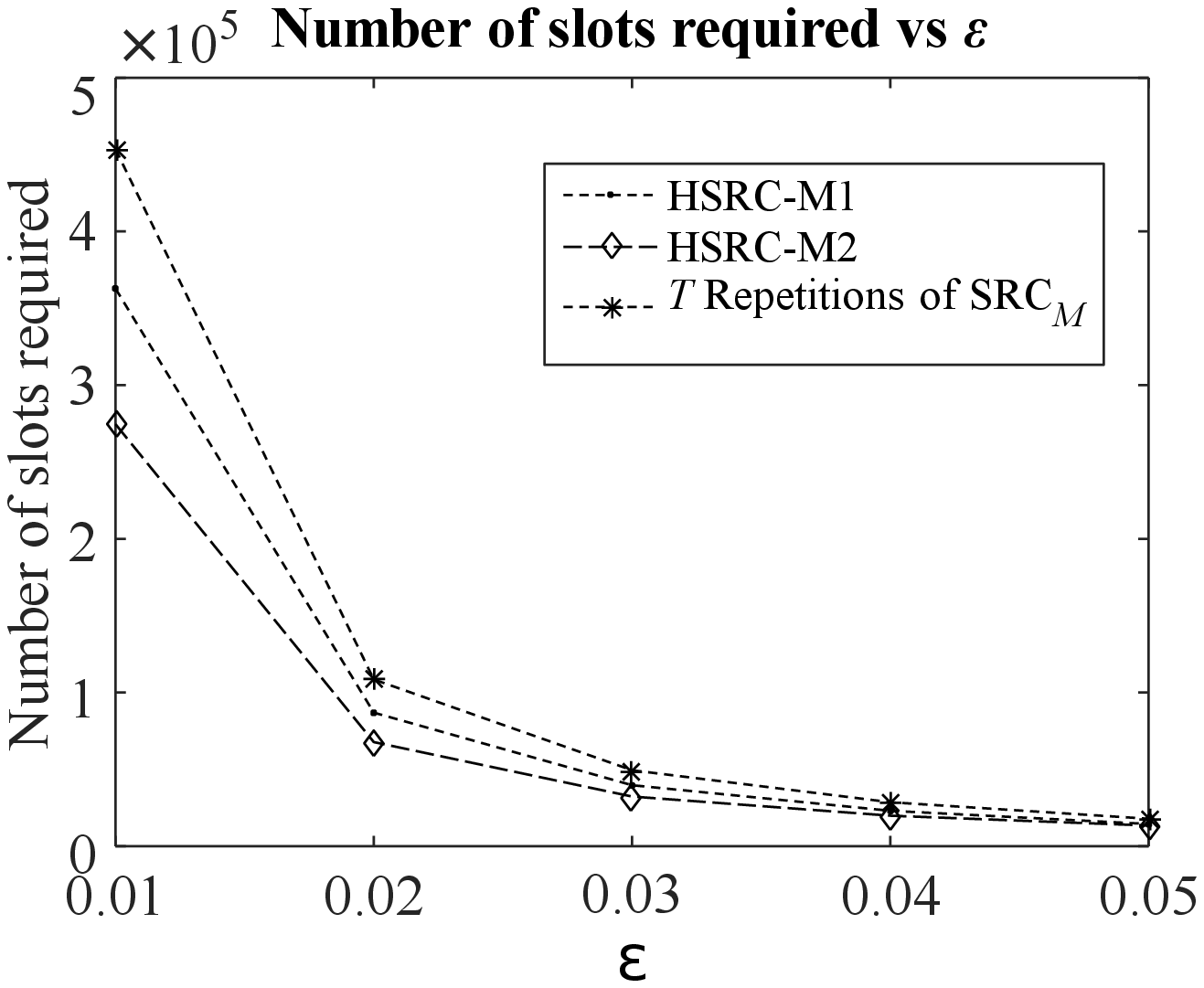}}
\caption{$T = 4$}
\label{Sim_eps}
\end{subfigure}
\caption{These plots show the average numbers of slots required by various estimation schemes versus $T$ and $\epsilon$, respectively, for Scenario I.  The following common parameter values are used in these plots: $D = 300$ and $q = 0.3$.}
\label{Sim7}
\end{figure} 

{\color{black}For all three scenarios, the performances of the proposed schemes, viz., HSRC-M1 and HSRC-M2, are compared with that of the scheme in which the  SRC$_M$ protocol proposed } in~\cite{zhou2014understanding} is separately executed $T$ times to estimate the active node cardinality of each node type. For a fair comparison, all the schemes are executed as many times as required to achieve the same accuracy level $\delta = 0.2$.



Figs.~\ref{Sim_q},~\ref{Sim_D},~\ref{Sim_T}, and~\ref{Sim_eps} show the number of slots required by various estimation schemes versus $q$, $D$, $T$, and $\epsilon$, respectively, for Scenario I.  Figs.~\ref{c2_q},~\ref{c2_D},~\ref{c2_T}, and~\ref{c2_eps} show the number of time slots required by various estimation schemes versus $\overline{q}$, $\overline{D}$, $T$, and $\epsilon$, respectively, for Scenario II.  Figs.~\ref{c4_q},~\ref{c4_D},~\ref{c4_T}, and~\ref{c4_eps} show the number of time slots required by various estimation schemes versus $\overline{q}$, $\overline{D}$, $T$, and $\epsilon$, respectively, for Scenario III. We observe that \emph{in all the plots, both the proposed schemes, viz., HSRC-M1 and HSRC-M2, significantly outperform the scheme in which SRC$_M$ is executed $T$ times}. Also, among the proposed schemes, HSRC-M2 outperforms HSRC-M1. Table~\ref{tab:results_table} shows the percentage improvement in the average number of time slots required by HSRC-M1 and HSRC-M2 relative to the scheme in which SRC$_M$ is executed $T$ times. For example, the first row in the table shows that in Fig.~\ref{Sim_q}, HSRC-M1 (respectively, HSRC-M2) outperforms the $T$ repetitions of SRC$_M$ protocol by 18.74\%  (respectively, 34.88\%) on average.

\begin{table}[H]
\centering
\caption{The table shows the percentage improvement in the average number of time slots required by HSRC-M1 and HSRC-M2 relative to the scheme in which SRC$_M$ is executed $T$ times for all three scenarios and all the considered parameters.}
\begin{tabular}{|c|c|c|c|}
\hline
& \multicolumn{1}{c|}{} & \multicolumn{2}{c|}{Improvement (in $\%$)} \\
\hline
& \multicolumn{1}{c|}{Evaluation} & \multicolumn{1}{c|}{HSRC-M1} & \multicolumn{1}{c|}{HSRC-M2} \\
\hline
\parbox[t]{2mm}{\multirow{4}{*}{\rotatebox[origin=c]{90}{\scriptsize  Scenario I}}} & No. of slots vs. $q$ (Fig.~\ref{Sim_q}) & 18.74 & 34.88\\
& No. of slots vs. $D$ (Fig.~\ref{Sim_D}) & 19.16 & 29.56\\
& No. of slots vs. $T$ (Fig.~\ref{Sim_T}) & 15.68  & 38.59\\
& No. of slots vs. $\epsilon$ (Fig.~\ref{Sim_eps}) & 19.74 & 33.47\\
\hline

\parbox[t]{2mm}{\multirow{4}{*}{\rotatebox[origin=c]{90}{\scriptsize  Scenario II}}} & No. of slots vs. $\overline{q}$ (Fig.~\ref{c2_q}) & 17.44 & 49.28\\
& No. of slots vs. $\overline{D}$ (Fig.~\ref{c2_D}) & 17.22 & 50.49\\
& No. of slots vs. $T$ (Fig.~\ref{c2_T}) & 17.71 & 42.91\\
& No. of slots vs. $\epsilon$ (Fig.~\ref{c2_eps}) & 17.21 & 50.86\\
\hline

\parbox[t]{2mm}{\multirow{4}{*}{\rotatebox[origin=c]{90}{\scriptsize Scenario III}}} & No. of slots vs. $\overline{q}$ (Fig.~\ref{c4_q}) & 17.43 & 50.55\\
& No. of slots vs. $\overline{D}$ (Fig.~\ref{c4_D}) & 17.23 & 50.68\\
& No. of slots vs. $T$ (Fig.~\ref{c4_T})  & 17.72 & 43.13\\
& No. of slots vs. $\epsilon$ (Fig.~\ref{c4_eps}) & 17.21 & 51.29\\
\hline
\end{tabular}
\label{tab:results_table}
\end{table}

%


\begin{figure}[ht]
\begin{subfigure}[b]{0.5\linewidth}
\centering
\resizebox{1.0\columnwidth}{!}
{\includegraphics{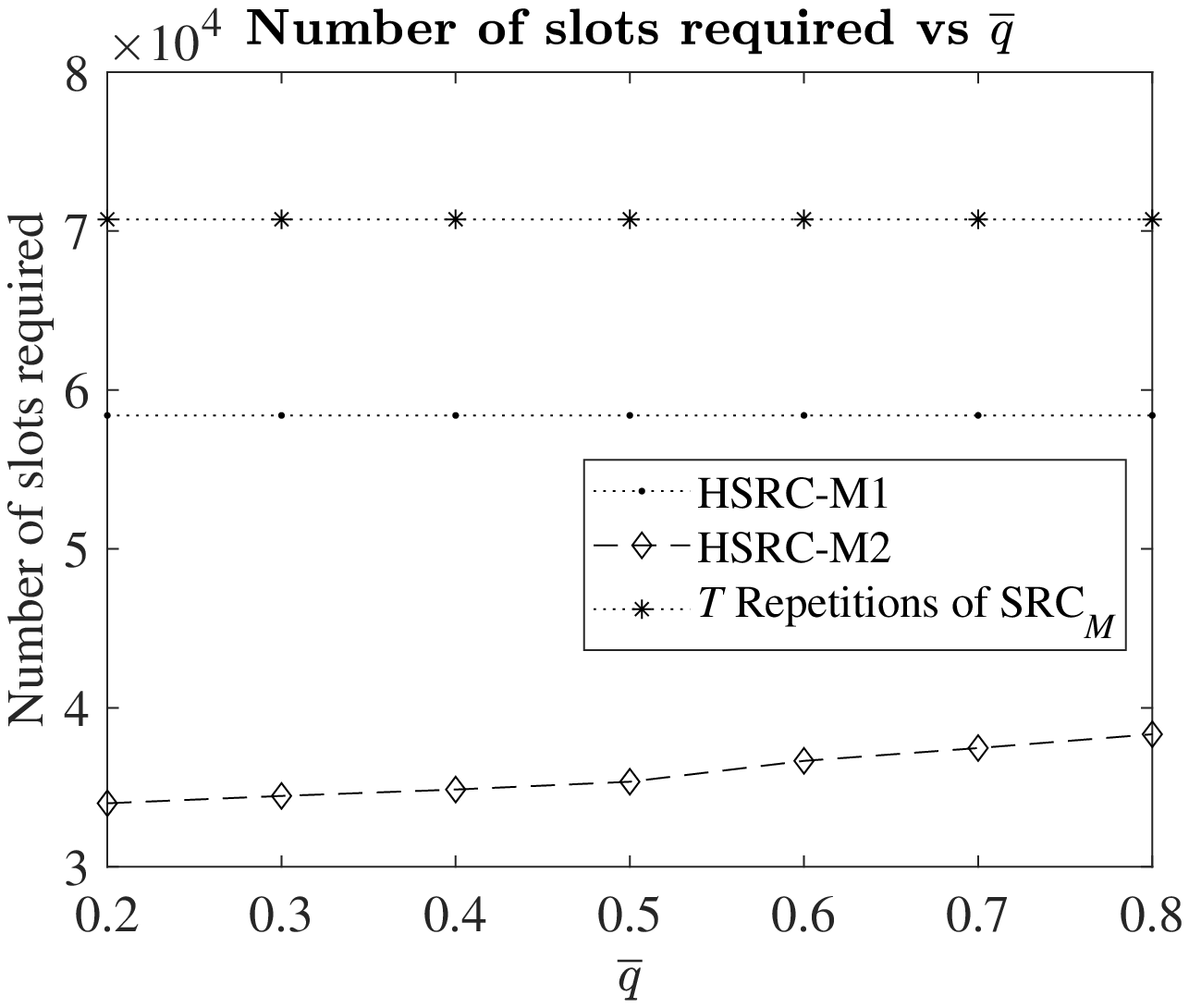}}
\caption{$\overline{D} = 1000$}
\label{c2_q}
\end{subfigure}%
\begin{subfigure}[b]{0.5\linewidth}
\centering
\resizebox{1.0\columnwidth}{!}{\includegraphics{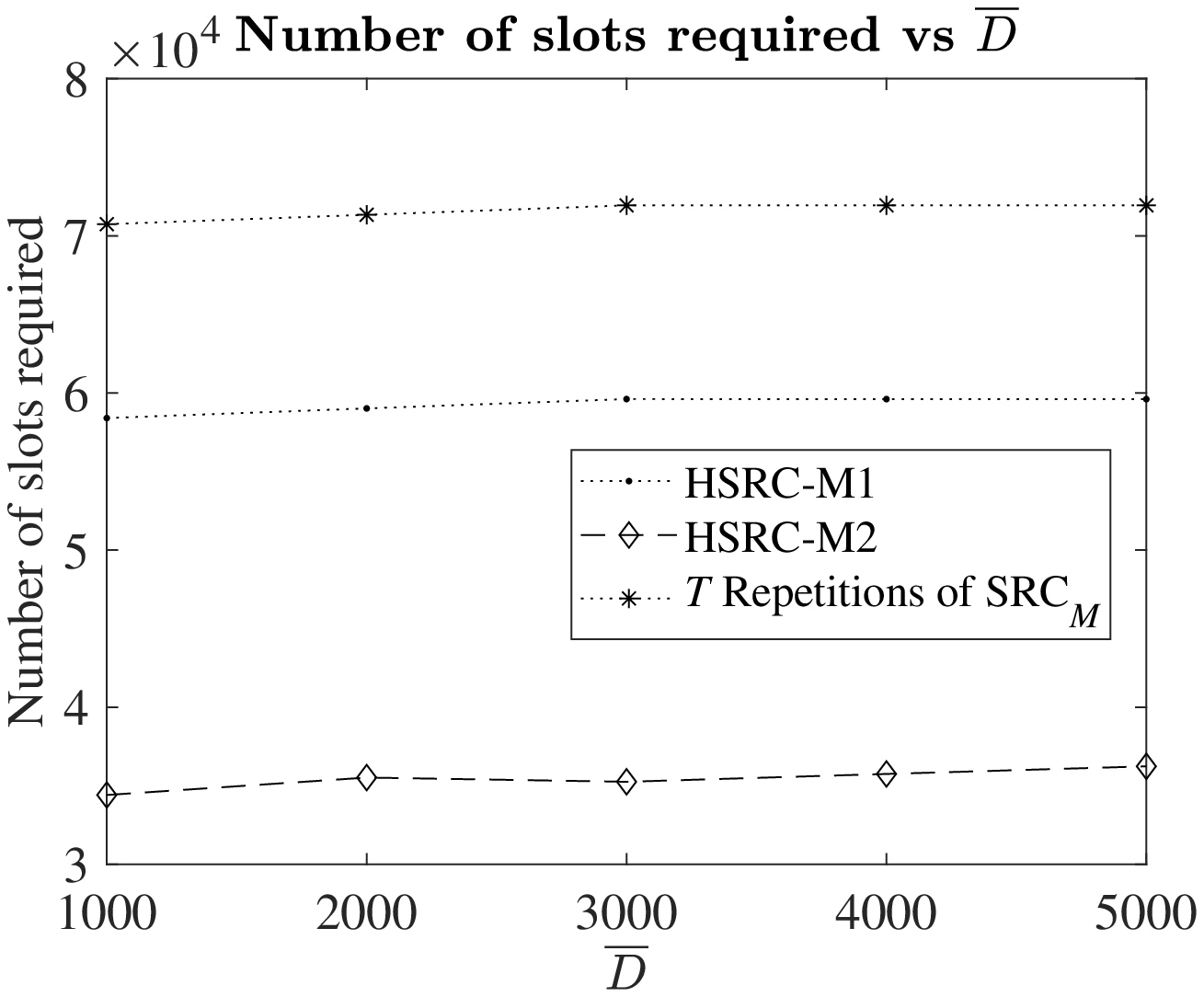}}
\caption{$\overline{q} = 0.3$}
\label{c2_D}
\end{subfigure}
\caption{These plots show the average numbers of slots required by various estimation schemes versus $\overline{q}$ and $\overline{D}$, respectively, for Scenario II. The following common parameter values are used in these plots: $T = 5$ and $\epsilon = 0.03$.}
\label{Sim8}
\end{figure}

\begin{figure}[ht]
\begin{subfigure}[b]{0.5\linewidth}
\centering
\resizebox{1.0\columnwidth}{!}{\includegraphics{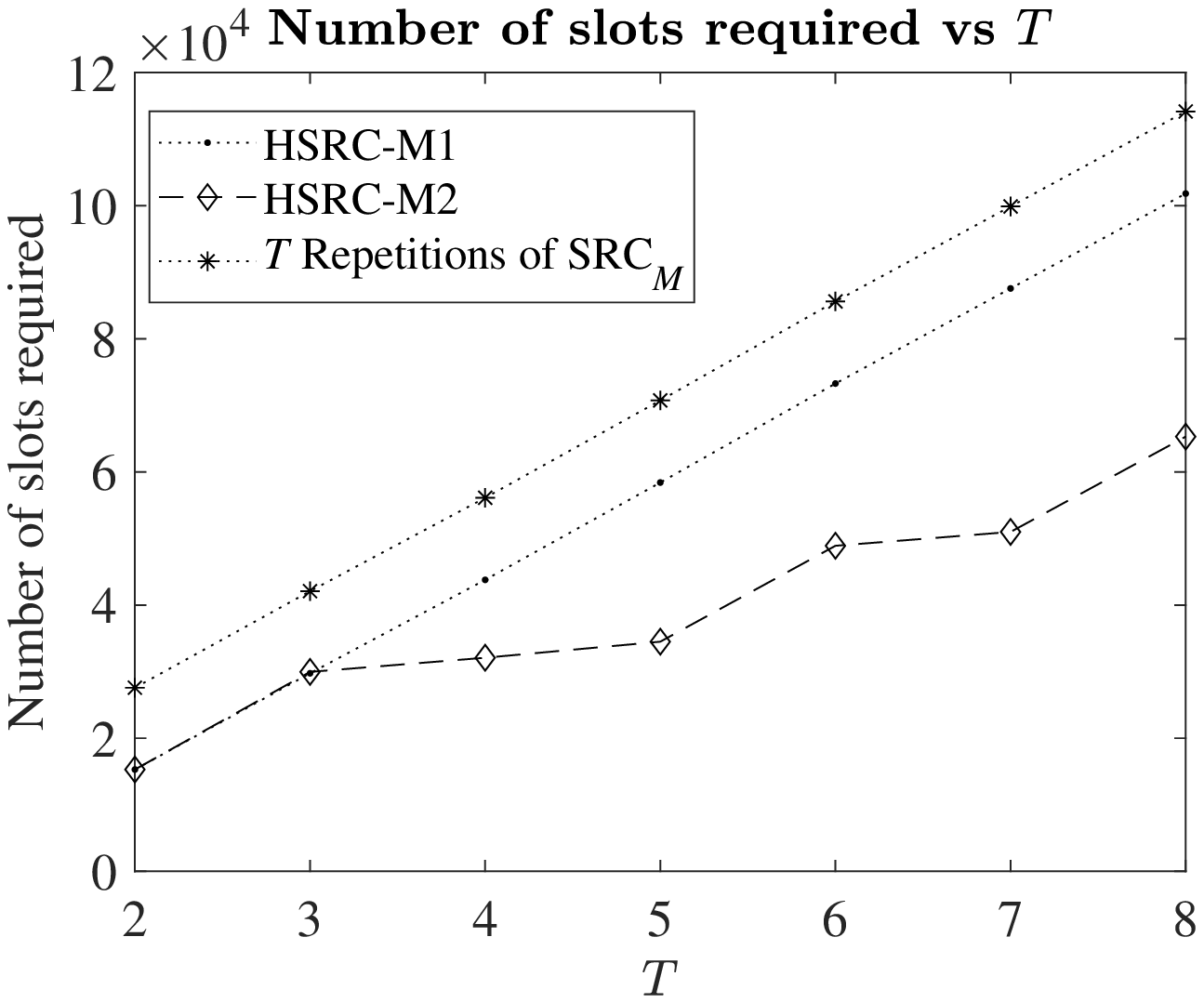}}
\caption{$\epsilon = 0.03$}
\label{c2_T}
\end{subfigure}%
\begin{subfigure}[b]{0.5\linewidth}
\centering
\resizebox{1.0\columnwidth}{!}{\includegraphics{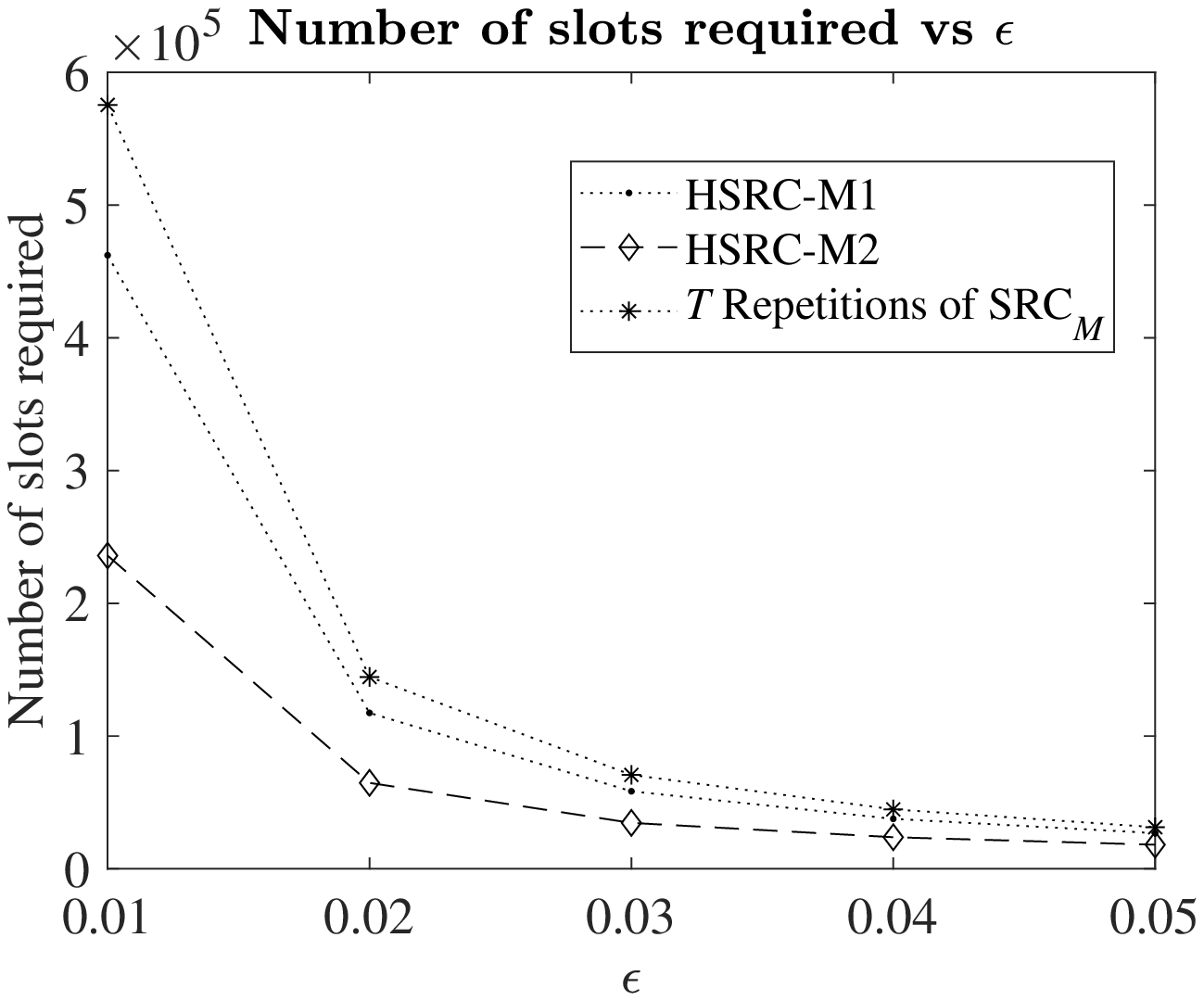}}
\caption{$T = 5$}
\label{c2_eps}
\end{subfigure}
\caption{These plots show the average numbers of slots required by various estimation schemes versus $T$ and $\epsilon$, respectively, for Scenario II. The following common parameter values are used in these plots: $\overline{D} = 1000$ and $\overline{q} = 0.3$.}
\label{Sim9}
\end{figure}




\begin{figure}[ht]
\begin{subfigure}[b]{0.5\linewidth}
\centering
\resizebox{1.0\columnwidth}{!}
{\includegraphics{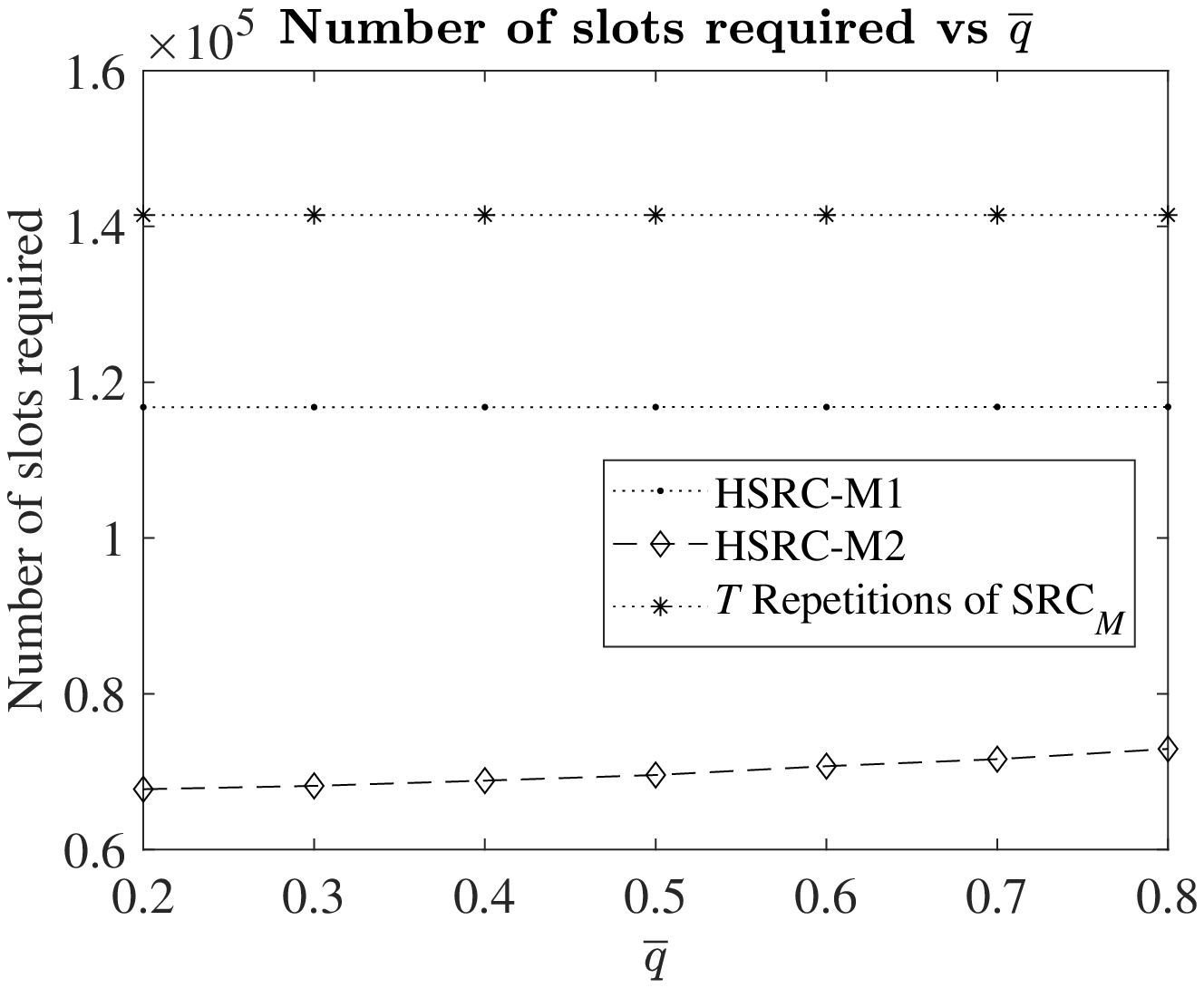}}
\caption{$\overline{D} = 1000$}
\label{c4_q}
\end{subfigure}%
\begin{subfigure}[b]{0.5\linewidth}
\centering
\resizebox{1.0\columnwidth}{!}{\includegraphics{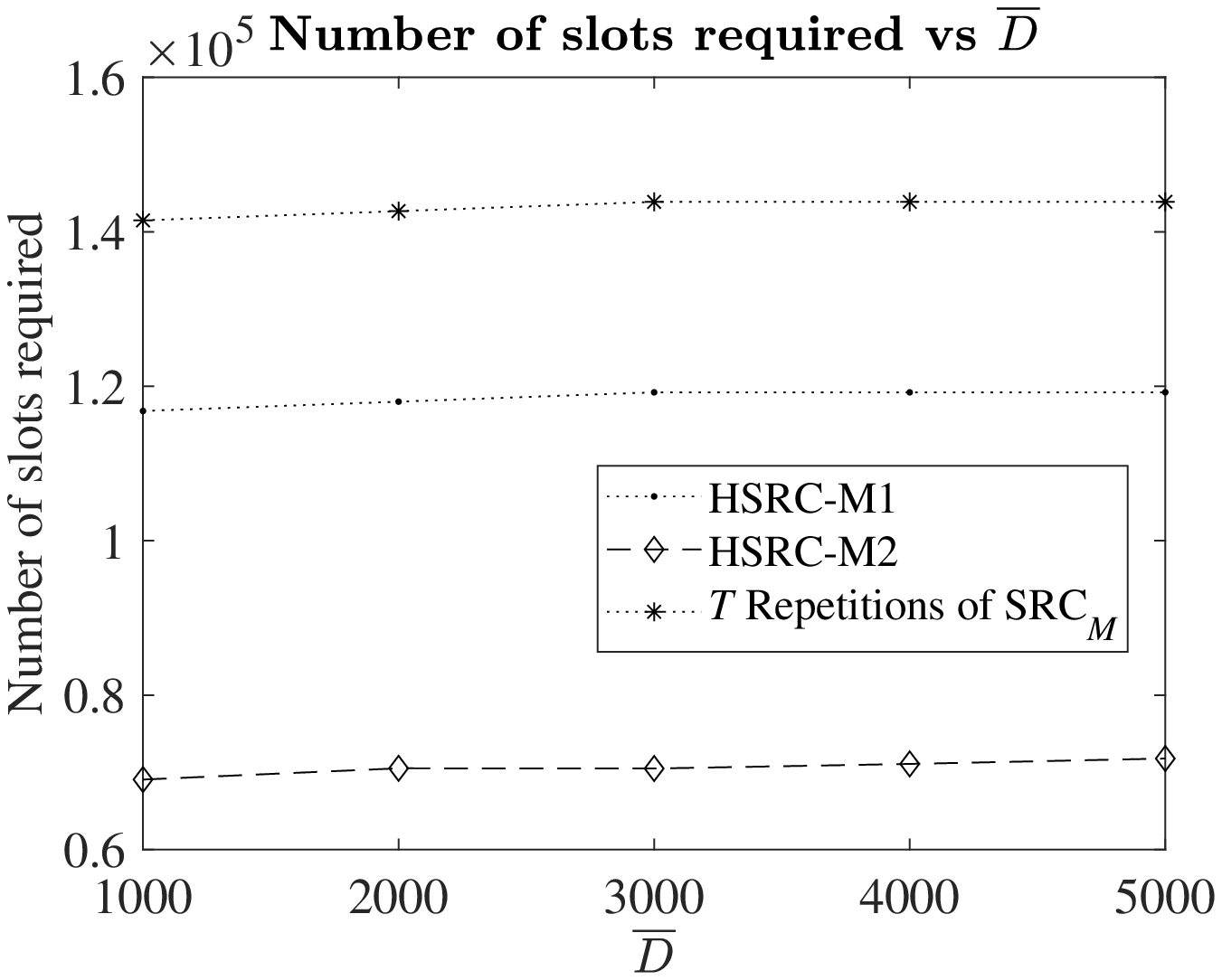}}
\caption{$\overline{q} = 0.3$}
\label{c4_D}
\end{subfigure}
\caption{These plots show the average numbers of slots required by various estimation schemes versus $\overline{q}$ and $\overline{D}$, respectively, for Scenario III. The following common parameter values are used in these plots: $T = 5$ and $\epsilon = 0.03$.}
\label{Sim10}
\end{figure}

\begin{figure}[ht]
\begin{subfigure}[b]{0.5\linewidth}
\centering
\resizebox{1.0\columnwidth}{!}{\includegraphics{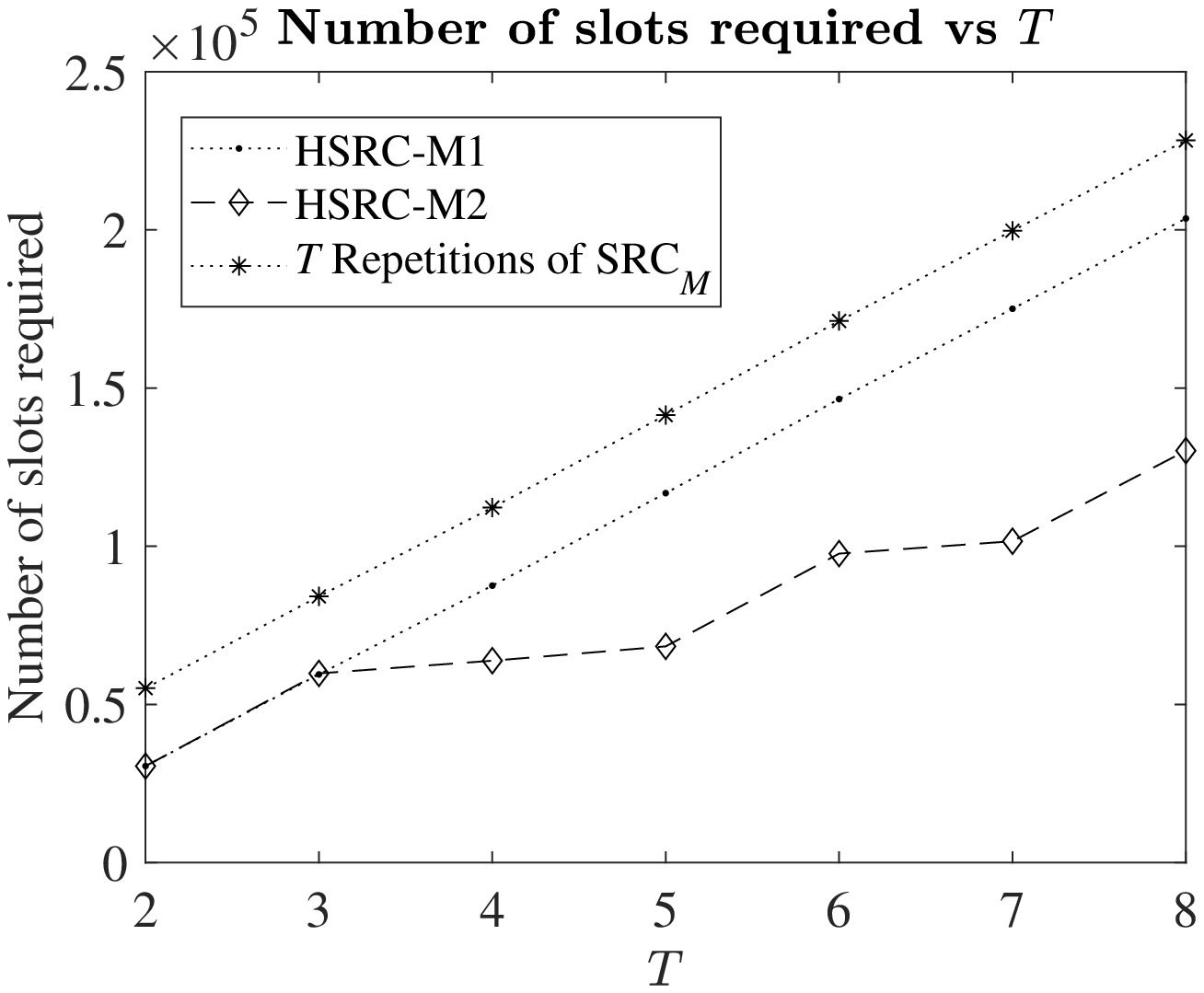}}
\caption{$\epsilon = 0.03$}
\label{c4_T}
\end{subfigure}%
\begin{subfigure}[b]{0.5\linewidth}
\centering
\resizebox{1.0\columnwidth}{!}{\includegraphics{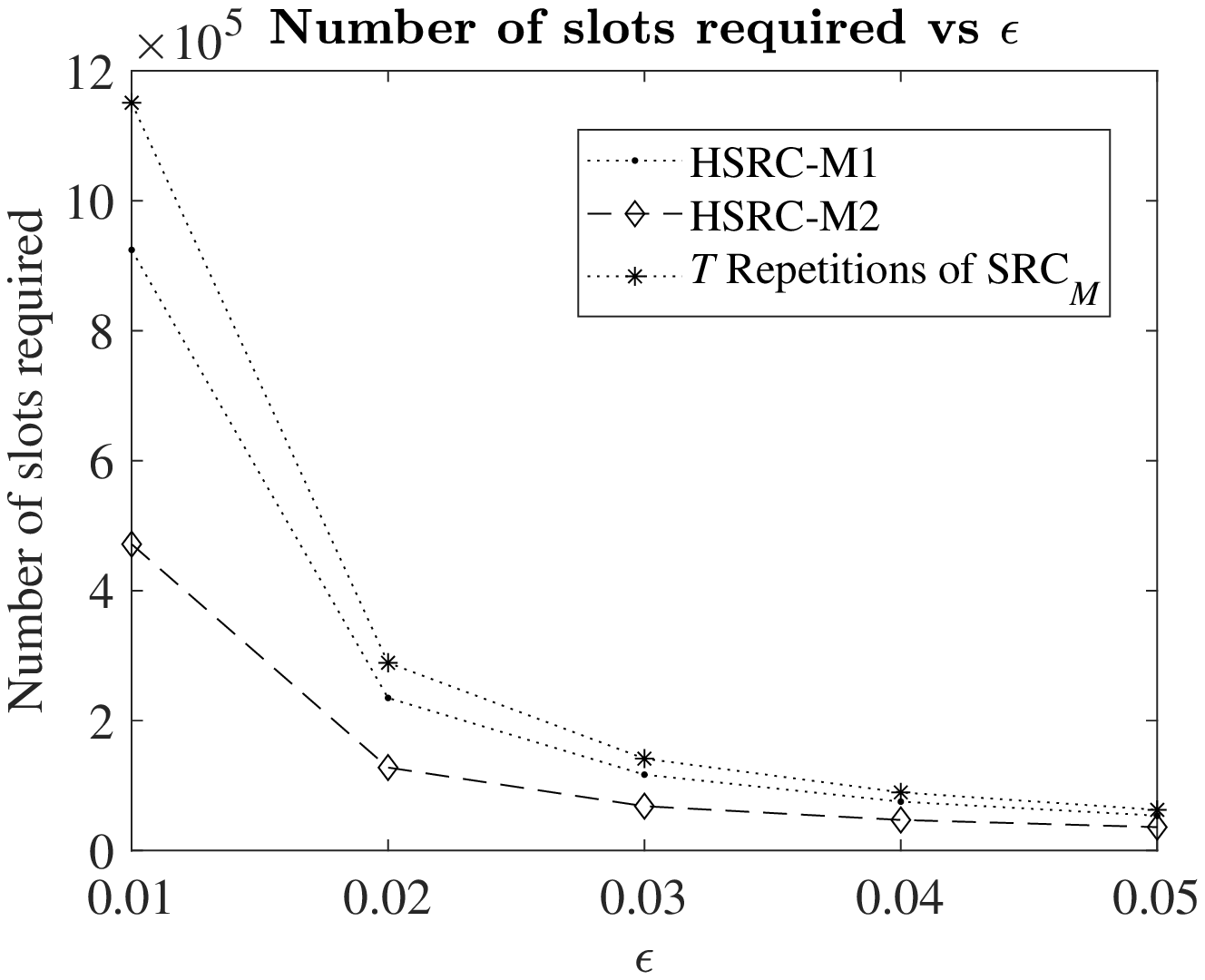}}
\caption{$T = 5$}
\label{c4_eps}
\end{subfigure}
\caption{These plots show the average numbers of slots required by various estimation schemes versus $T$ and $\epsilon$, respectively, for Scenario III. The following common parameter values are used in these plots: $\overline{D} = 1000$ and $\overline{q}= 0.3$.}
\label{Sim11}
\end{figure}




\begin{figure}[ht]
\begin{subfigure}[b]{0.5\linewidth}
\centering
\resizebox{1.0\columnwidth}{!}{\includegraphics{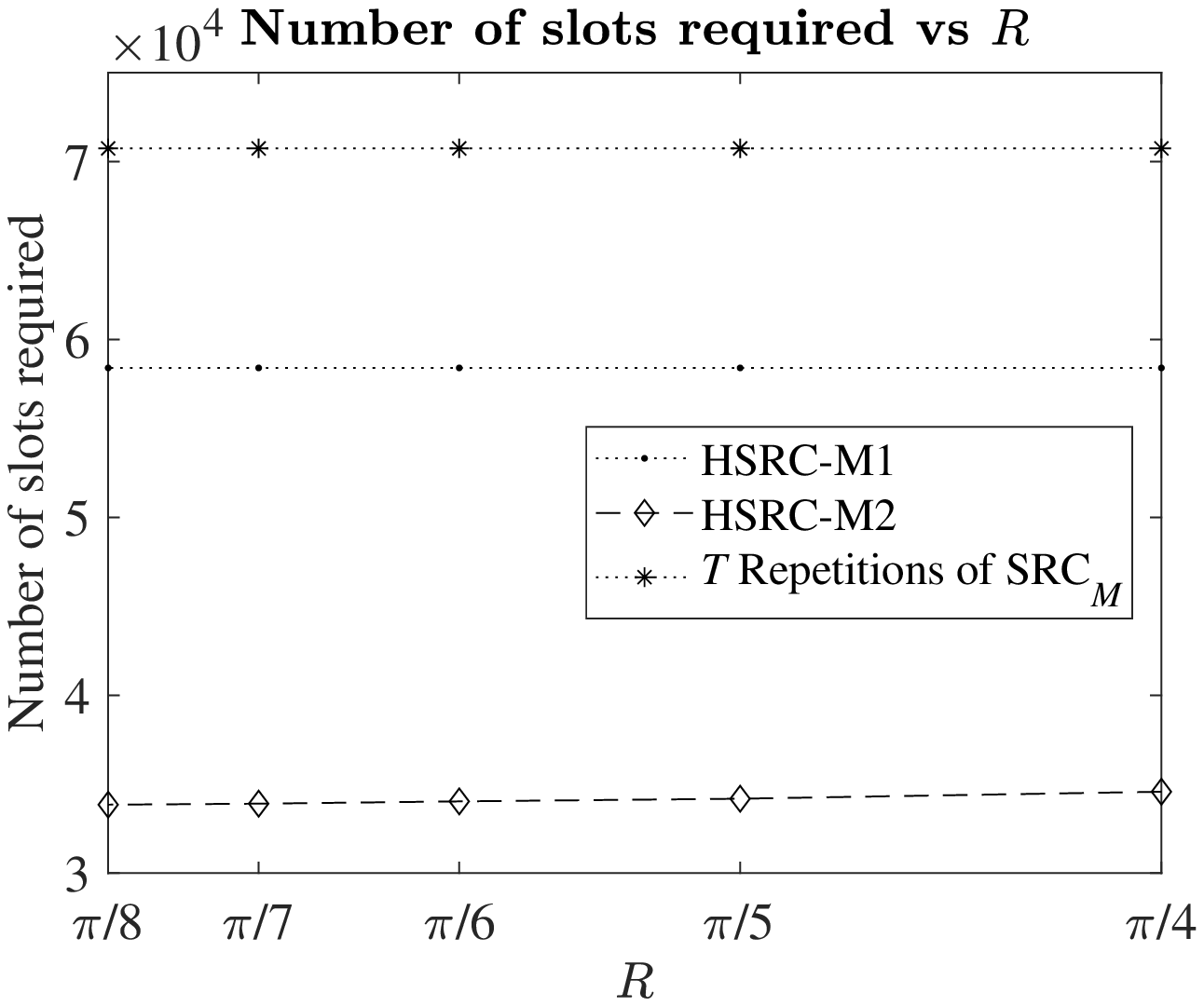}}
\caption{Scenario II}
\label{c2_R}
\end{subfigure}%
\begin{subfigure}[b]{0.5\linewidth}
\centering
\resizebox{1.0\columnwidth}{!}{\includegraphics{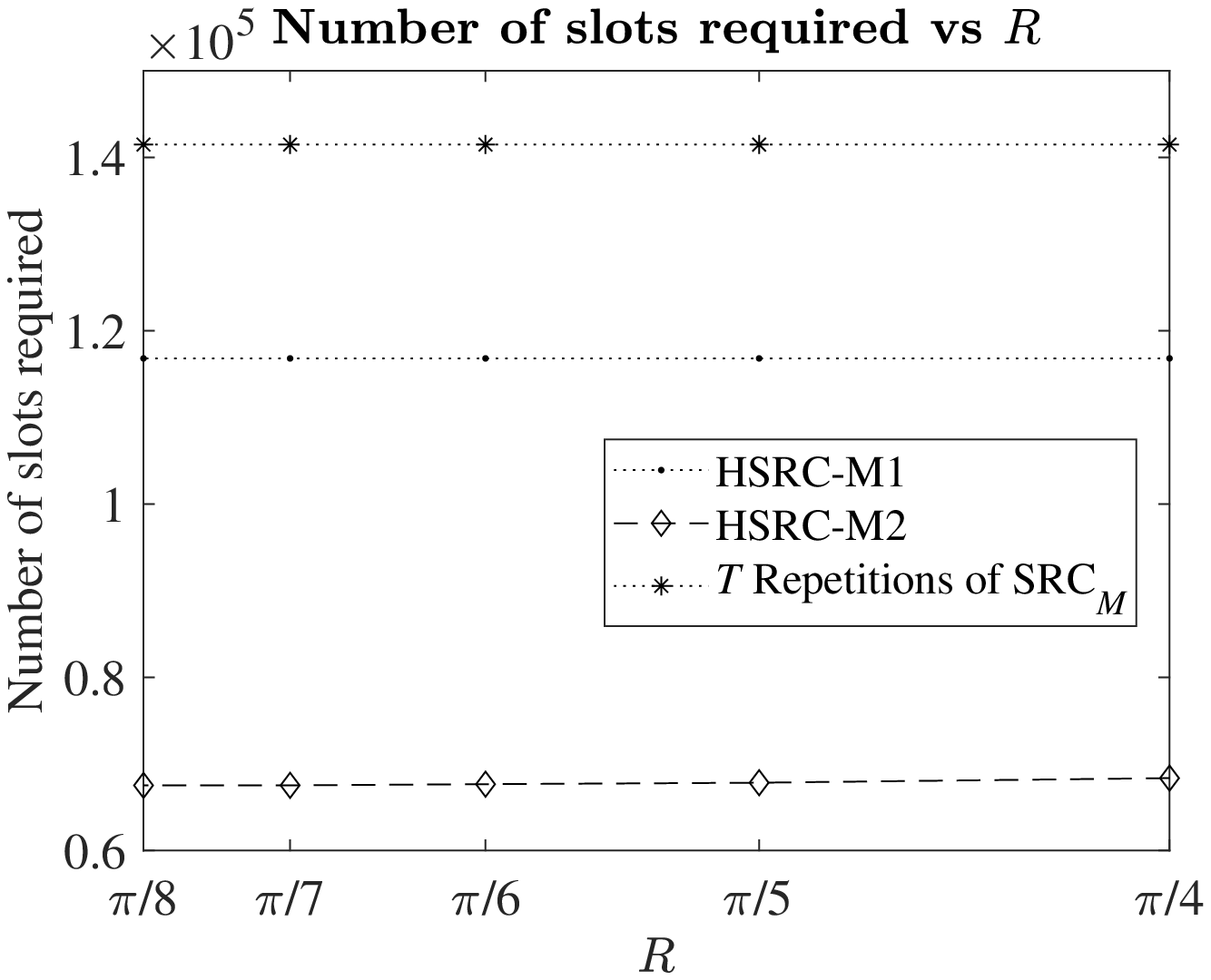}}
\caption{Scenario III}
\label{c4_R}
\end{subfigure}
\caption{These plots show the average numbers of slots required by various estimation schemes versus $R$ (range of MBS) in Scenarios II and III, respectively. The following common parameter values are used in these plots: $\overline{D} = 1000$ and $\overline{q}= 0.3$.}
\label{Sim11a}
\end{figure}

Figs.~\ref{c2_R} and~\ref{c4_R} show the number of time slots required by various estimation schemes versus the range ($R$) of the MBS for Scenario II and Scenario III, respectively. {\color{black} For Scenario II, from Fig.~\ref{c2_R}, we observe that HSRC-M2 outperforms the $T$ repetitions of SRC$_M$ protocol by $51.78\%$, whereas HSRC-M1 outperforms the $T$ repetitions of SRC$_M$ protocol by $17.44\%$. Similarly, for Scenario III, from Fig.~\ref{c4_R}, we observe that HSRC-M2 outperforms the $T$ repetitions of SRC$_M$ protocol by $52.09\%$, whereas HSRC-M1 outperforms the $T$ repetitions of SRC$_M$ protocol by $14.14\%$}. 


\section{Conclusions and Future Work}\label{conclusion}
{\color{black}Our work shows that using the proposed schemes, viz., HSRC-M1 and HSRC-M2, it is possible to find node cardinality estimates for a heterogeneous wireless network deployed in a large region using a significantly smaller number of time slots as compared with $T$ repetitions of the benchmark protocol, SRC$_M$, but with the same accuracy. Our mathematical analysis of the expected number of
time slots required for HSRC-M1 to execute and the expected energy consumption of a node under HSRC-M1 provides several insights into the operation of the scheme. Also, our work has established that the OMT problem is computationally hard.
}

A direction for future research is to address the information-theoretic question of finding a lower bound on the number of time slots that a protocol requires for finding separate estimates of the numbers of active nodes of each type in a heterogeneous network deployed over a large region. {\color{black}Another challenge for future research is to design node cardinality estimation schemes that are provably optimal, i.e., require the minimum possible number of time slots to find estimates of the numbers of active nodes of each type in a heterogeneous network deployed over a large region, or order optimal. Finally, another open problem is to design an approximation algorithm with a provable approximation ratio for the OMT problem formulated in Section 6.}

\appendix
\section*{Appendix}
\section{Proof of Theorem~\ref{Thm_OMT}} \label{Apdx_Thm_OMT}
The decision version of the OMT problem is as follows: ``Given a number $L$, do there exist a number $\hat{M} \le \tilde{M}$ and a tour $(m_i,m_{i+1}), m_0=m_{\hat{M}+1}=0, i \in \{0,1, \ldots, \hat{M}\}$, that satisfy the constraints~\eqref{eq:energy_constraint}-\eqref{eq:Mstar_constraint} such that $\sum_{i=0}^{\hat{M}} c_{m_i,m_{i+1}} \le L$''? Given a number $\hat{M} \le \tilde{M}$ and a tour $(m_i,m_{i+1}), m_0=m_{\hat{M}+1}=0, i \in \{0,1, \ldots, \hat{M}\}$, we can check in polynomial time whether they satisfy ~\eqref{eq:energy_constraint}-\eqref{eq:Mstar_constraint} and whether $\sum_{i=0}^{\hat{M}} c_{m_i,m_{i+1}} \le L$.
Thus, the OMT problem is in class NP~\cite{kleinberg2006algorithm}. {\color{black}The OMT problem is now shown to be NP-complete by reduction of the travelling salesman problem (TSP), which has been shown to be NP-complete~\cite{kleinberg2006algorithm}, to it. }

Consider the following instance of the TSP. There is a network represented by the graph $\mathcal{G}' = (\mathcal{M}', \mathcal{E}')$, $|\mathcal{M}'| = M'+1$, where $\mathcal{M}'$ is a set of cities and $\mathcal{E}'$ is a set of links between cities,  and the set of distances $\{d_{i,j} >0:  (i,j) \in \mathcal{E}', \forall i,j \in \mathcal{M}' \mbox{ and } i \neq j \}$. A tour in the TSP is defined as $(m'_i,m'_{i+1}), m'_0=m'_{M'+1}=\overline{m}, i \in \{0,1, \ldots, M'\}$, where $\overline{m} \in \mathcal{M}'$, and every city (except $\overline{m}$) is visited exactly once. 
The decision version of the TSP is as follows~\cite{kleinberg2006algorithm}:
``Given a set of distances between $\mathcal{M}'$ cities, and a bound $L^{\prime}$, is there a tour of length at most $L^{\prime}$?''

{\color{black}Now, the TSP is shown to be polynomial-time reducible to the OMT problem, i.e., TSP $<_p$ OMT problem.} Consider the instance of the TSP stated in the previous paragraph. {\color{black}From this instance, the following instance of the OMT problem is constructed:} Let $\mathcal{M} = \mathcal{M}'$, $\mathcal{E} = \mathcal{E}'$, $0 = \overline{m}$, and $c_{u,v} = d_{u,v}$  $\forall u,v \in \mathcal{E}$, $u \neq v$. Select the node set $\mathcal{N}$ such that for every $m \in \mathcal{M}$, $\exists k \in \mathcal{N}$ such that $X_{k,m} = 1$ and $X_{k,m^{\prime}} = 0, ~\forall m^{\prime} \in \mathcal{M} \setminus \{m\}$, i.e.,
\begin{align}
    X_{k,u} = 
    \begin{cases}
    1, & u=m, \\
    0, & \forall u \in \mathcal{M} \setminus \{m\}.
    \end{cases}
\end{align}
This ensures that the constraint~\eqref{eq:coverage_constraint} is satisfied only when  $\hat{M} = \tilde{M}$, which satisfies the constraint~\eqref{eq:Mstar_constraint}. Let $\eta_{k,m} = 1,$ $\forall k \in \mathcal{N}, m \in \mathcal{M}$. {\color{black}Also, no energy constraint is considered in this instance, i.e., let $\overline{\eta} = \infty$; then~\eqref{eq:energy_constraint} is satisfied for every tour.} Finally, let $L = L^{\prime}$.

We claim that there exists a tour of length at most $L^{\prime}$ in the TSP iff there is a tour that satisfies the constraints~\eqref{eq:energy_constraint}-\eqref{eq:Mstar_constraint} such that $\sum_{i=0}^{\hat{M}} c_{m_i,m_{i+1}} \le L$ in the OMT problem. 
To show necessity, suppose there exists a tour $(m'_i,m'_{i+1})$, $m'_0=m'_{M'+1}=\overline{m}$, $i \in \{0,1, \ldots, M'\}$ of length at most $L^{\prime}$ in the TSP. 
Let $\hat{M} = M'$ and $m_i = m'_i$ $\forall i \in \{0,1, \ldots, \hat{M} + 1\}$. Then by the above construction of the OMT problem, it follows that the tour $(m_i,m_{i+1})$, $i \in \{0,1, \ldots, \hat{M}\}$ satisfies the constraints~\eqref{eq:energy_constraint}-\eqref{eq:Mstar_constraint} and its travel cost $\sum_{i=0}^{\hat{M}} c_{m_i,m_{i+1}}$ equals the length of the tour $(m'_i,m'_{i+1})$, $m'_0=m'_{M'+1}=\overline{m}$, $i \in \{0,1, \ldots, M'\}$ in the TSP and hence is $\le L$. This proves necessity.


To prove sufficiency, let us consider a tour $(m_i,m_{i+1})$, $m_0=m_{\hat{M}+1}=0$, $i \in \{0,1, \ldots, \hat{M}\}$, that satisfies the  constraints~\eqref{eq:energy_constraint}-\eqref{eq:Mstar_constraint} and such that $\sum_{i=0}^{\hat{M}} c_{m_i,m_{i+1}} \le L$ in the OMT problem. Let $m'_i = m_i$, $\forall i \in \{0,1, \ldots, M' + 1\}$. Then it is easy to see that the length of the tour $(m'_i,m'_{i+1}), m'_0=m'_{M'+1}=\overline{m}, i \in \{0,1, \ldots, M'\}$ in the TSP is equal to the travel cost $\sum_{i=0}^{\hat{M}} c_{m_i,m_{i+1}}$ in the OMT problem and hence is $\le L^{\prime}$. 
This proves sufficiency, and the result follows.
\section{2-SP in the special cases
\texorpdfstring{$T = 7$}{} and \texorpdfstring{$T = 8$}{}}\label{Apdx_2SP}
\subsection{Scheme for \texorpdfstring{$T = 7$}{}} \label{Apdx_2SP_T7}
Step 1 consists of $\ell$ blocks, each consisting of 3 slots. From Fig.~\ref{Sym_Combo}, it follows that the symbol combinations given in the following table are used. 
\begin{center}
  \begin{tabular}{| c | c |}
    \hline
    Type & Symbol Combination \\ \hline
    1 & $\alpha$ $0$ $0$  \\ \hline
    2 & $\alpha$ $\alpha$ $0$ \\ \hline
    3 & $\alpha$ $\alpha$ $\alpha$ \\ \hline
    4 & $0$ $0$ $\beta$ \\ \hline
    5 & $0$ $\beta$ $\beta$ \\ \hline
    6 & $\beta$ $\beta$ $\beta$ \\ \hline
    7 & $\beta$ $0$ $\alpha$ \\ \hline
  \end{tabular}
\end{center}
\subsubsection{Step 1}
In each slot, there are four possible outcomes: $E$ (empty or no transmission), $\alpha$, $\beta$, and $C$ (collision). So the total number of potential outcomes in the three slots of a block $B_i$ is $4^3=64$. It is clear that the MBS can find the bit patterns unambiguously if there is no $C$ in any of the slots of block $B_i$.\footnote{For example, if Slot 1 results in $\alpha$ and Slots 2 and 3 result in $\beta$, the MBS deduces that in block $B_i$, only one node each of $\mathcal{T}_1$ and $\mathcal{T}_5$ transmitted and no nodes of $\mathcal{T}_2,\mathcal{T}_3,\mathcal{T}_4,\mathcal{T}_6,\mathcal{T}_7$ transmitted. Similarly, if Slot 1 results in $\alpha$, Slot 2 results in $E$, and Slot 3 results in $\beta$, the MBS deduces that in block $B_i$, only one node each of $\mathcal{T}_1$ and $\mathcal{T}_4$ transmitted and no nodes of $\mathcal{T}_2,\mathcal{T}_3,\mathcal{T}_5,\mathcal{T}_6,\mathcal{T}_7$ transmitted.} {\color{black}So  the cases where at least one $C$ happens in each block $B_i$ are considered.} The number of cases in which a collision happens in precisely one slot (respectively, two slots) in a block is $\binom{3}{1} \times 3^2 = 27$ (respectively, $\binom{3}{2}\times 3^1 = 9$). Table~\ref{Tab_OneC1_T7} only shows the instances in which there is uncertainty about the activity or inactivity of some of the types of nodes and additional slots in step 2 are required to address the ambiguity. The types of nodes mentioned in the `Active' (respectively, `Inactive') column are detected as certainly active (respectively, inactive) by the MBS at the conclusion of step 1; those listed in the `Not Sure' column have uncertainty, which must be resolved in step 2.\footnote{For example, if Slots 1 and 2 result in $\alpha$ and Slot 3 results in $C$, the MBS deduces that in block $B_i$, at least 1 node of $\mathcal{T}_4$ is active, no nodes of $\mathcal{T}_1,\mathcal{T}_5,\mathcal{T}_6,\mathcal{T}_7$ are active  and exactly one node of $\mathcal{T}_2$ or $\mathcal{T}_3$ is active.  Similarly, if Slot 1 results in $E$ and Slots 2 and 3 result in $C$, the MBS deduces that in block $B_i$, at least $2$ nodes of $\mathcal{T}_5$ are active, no nodes of $\mathcal{T}_1,\mathcal{T}_2,\mathcal{T}_3,\mathcal{T}_6,\mathcal{T}_7$ are active  and the presence/ absence of nodes of $\mathcal{T}_4$ cannot be inferred.} Also, if the result $CCC$ happens in a block of step 1, there is uncertainty concerning the activity or inactivity of all seven node types, which must be addressed in step 2. 
\begin {table}[h] 
\centering
\caption{This table shows the deductions of the MBS about the activity or inactivity of various types of nodes for each of the outcomes in which one or two collisions occur in a block and there is uncertainty about the activity or inactivity of some of the nodes.}
\label{Tab_OneC1_T7}
\resizebox{\columnwidth}{!}{%
\begin{tabular}{|c|c|c|c|c|c|} 
\hline
\multicolumn{3} { | c |} {Outcome in Block $B_i$}  & \multicolumn{3} { | c |} {Types}\\ 
\hline 
Slot 1 & Slot 2 & Slot 3 & Active & Inactive & Not Sure \\
\hline
$C$ & $\beta$ & $\beta$ & $1$ & $2, 3,4, 7$ & One of $\{5, 6\}$\\ \hline
 $\alpha$ & $\alpha$ & $C$ & $4$ & $1,5,6, 7$ & One of $\{2, 3\}$\\ \hline
 $\beta$ & $\beta$ & $C$ & - & $1,2, 3$ & One of $\{\{4, 6\},\{4,5,7\}\}$\\ \hline
 $C$ & $\alpha$ & $\alpha$ &  - & $4,5,6$ & One of $\{\{1, 3\},\{1,2,7\}\}$\\ \hline
 $C$ & $C$ & $E$ & $2$ & $3,4,5,6,7$ & $1$ \\ \hline
 $C$ & $C$ & $\alpha$ &  $2$ & $4,5,6$ & $1$, One of $\{3,7\}$ \\ \hline
$C$ & $C$ & $\beta$ & $2$ & $3,7$ & $1$, One of $\{4,5,6\}$ \\ \hline
$E$ & $C$ & $C$ &  $5$ & $1,2,3,6,7$ & $4$ \\ \hline
$\alpha$ & $C$ & $C$ &  $5$ & $6,7$ & $4$, One of $\{1,2,3\}$ \\ \hline
$\beta$ & $C$ & $C$ &  $5$ & $1,2,3$ & $4$, One of $\{6,7\}$ \\ \hline
$C$ & $E$ & $C$ &  - & $2,3,5,6$ & $1,4,7$ \\ \hline
$C$ & $\alpha$ & $C$ &  - & $5,6$ & $1, 4, 7$, One of $\{2,3\}$ \\ \hline
$C$ & $\beta$ & $C$ &  - & $2,3$ & $1, 4, 7$, One of $\{5,6\}$ \\ \hline

\end{tabular}%
}
\end{table}

\subsubsection{Broadcast packet (BP)} \label{sec:BPT7}
Just after step 1, the MBS informs all nodes of the list of type numbers for which ambiguity regarding activity or inactivity must be addressed in each block through a BP consisting of a bit stream (see Fig.~\ref{Est_Window2}). Let $\chi(T)$ be the set of all possible subsets of types of nodes for which the MBS may not be certain of their activity or inactivity after step 1 for a given $T$; the case where the activity (or inactivity) of every type of node is known to the MBS with certainty is included in the set $\chi(T)$ as an empty set ($\emptyset$). We can observe from Table~\ref{Tab_OneC1_T7} that $\chi(7) = \{ \{5, 6\}, \{2,3\}, \{4,6\}, \{4,5,7\}, \{1,3\}, \{1,2,7\},  \{1\}, \{3,7\}, \\ \{4,5,6\}, \{4\}, \{1,2,3\}, \{6,7\}, \{7\}, \{1,2,3,4,5,6, 7\}, \emptyset\}$. Assume that the bit strings $0000, 0001, \ldots, 1110$ are used to represent the members, respectively, of the set $\chi(7)$. The MBS concatenates the bit strings corresponding to all the blocks of step 1 and transmits the concatenated bit string in the BP; this concatenated list is received by all active nodes, who act accordingly in step 2.

Suppose $b'(T)$ denotes the number of bits utilized to represent each member of the set $\chi(T)$. The length of the BP (in terms of slot count) for a given $T$ is then given by $\bar{\ell}(T) = \ceil{b'(T) \times \ell/S_W}$. 
For $T=7$, $b'(7) = \ceil{\log_2|\chi(7)|} = 4$ and $\bar{\ell}(7) = \ceil{b'(7) \times \ell/S_W} = \ceil{4 \ell/S_W}$.
\subsubsection{Step 2} \label{sec:T7Step2}
Only active nodes of those types and blocks that have ambiguity after step 1 participate in this step. To resolve the ambiguity, e.g., the step 1 block result $C\beta\beta$ (see Table~\ref{Tab_OneC1_T7}) requires one more slot in step 2. For this outcome, active $\mathcal{T}_5$  nodes corresponding to that block reply with symbol $\beta$, and active $\mathcal{T}_6$ nodes corresponding to that block do not respond in their corresponding slot of step 2. If the slot outcome is symbol $\beta$, it means that a node of $\mathcal{T}_5$ is active; otherwise, the slot outcome is $E$, which means that a node of $\mathcal{T}_6$ is active. As a result, the ambiguity in that block is resolved at the end of step 2. If the result $CCC$ occurs in a block of step 1, then there is uncertainty regarding the activity or inactivity of nodes of all seven types at the end of step 1; then, in step 2, the estimation techniques for $T=4$ and $T=3$ are employed to resolve this ambiguity for each of the groups $\{\mathcal{T}_1,\mathcal{T}_2,\mathcal{T}_3,\mathcal{T}_4\}$ and $\{\mathcal{T}_5,\mathcal{T}_6,\mathcal{T}_7\}$. Step 2 is \textit{recursive}-- if the result $CCC$ occurs in a block of step 1  while performing the scheme for $T = 7$, the strategies of $T = 4$ and $T = 3$  are utilized to resolve the ambiguity in step 2. 

\subsection{Scheme for \texorpdfstring{$T = 8$}{}} \label{Apdx_2SP_T8}
Step 1 consists of $\ell$ blocks, each consisting of 4 slots. From Fig.~\ref{Sym_Combo}, it follows that the symbol combinations
given in the following table are used. 
\begin{center}
  \begin{tabular}{| c | c |}
    \hline
    Type & Symbol Combination \\ \hline
    1 & $\alpha$ $0$ $0$ $0$ \\ \hline
    2 & $\alpha$ $\alpha$ $0$ $0$ \\ \hline
    3 & $\alpha$ $\alpha$ $\alpha$ $0$ \\ \hline
    4 & $\alpha$ $\alpha$ $\alpha$ $\alpha$ \\ \hline
    5 & $0$ $0$ $0$ $\beta$ \\ \hline
    6 & $0$ $0$ $\beta$ $\beta$ \\ \hline
    7 & $0$ $\beta$ $\beta$ $\beta$ \\ \hline
    8 & $\beta$ $\beta$ $\beta$ $\beta$ \\ \hline 
  
  \end{tabular}
\end{center}

\subsubsection{Step 1}
The number of cases in which a collision occurs in exactly one slot (respectively, two and three slots) of a block is $\binom{4}{1} \times 3^3 = 108$ (respectively, $\binom{4}{2}\times 3^2 = 54$ and $\binom{4}{3} \times 3 = 12$). Out of these, in Table~\ref{Tab_OneC1}, we only show the cases where ambiguity about the activity or inactivity of some of the types of nodes exists and additional slots in step 2 are needed to resolve the ambiguity. Also, it is easy to check that if the result $CCCC$ occurs in a block of step 1, then ambiguity remains about the activity or inactivity of all eight node types, and this ambiguity needs to be resolved
in step 2.

\subsubsection{Broadcast Packet (BP)}
Following step 1, the MBS sends a BP identical to that described in Section~\ref{sec:BPT7}. Using notations similar to those in Section~\ref{sec:BPT7}, $\chi(8) = \{ \{7, 8\}, \{3, 4\}, \{1\}, \{1,6,7,8\}, \{5\}, \\~~\{2,3,4,5\}, \{1,2\},  \{1,2,5,6,7,8\}, \{1,3,4\}, \{5,7,8\}, \{5,6\},\\ ~~ \{1,2,3,4,5,6\},  \{1,2,3,4,5,6, 7,8 \}, \emptyset\}$. So $b'(8)$ \\= $\ceil{\log_2|\chi(8)|} = 4$ and $\bar{\ell}(8) = \ceil{b'(8) \times \ell/S_W} = \ceil{4 \ell/S_W}$. 

\subsubsection{Step 2}
In step 2, only nodes of those types and blocks participate, for which ambiguity still exists after step 1. To resolve the ambiguity in case of the block results listed in Table~\ref{Tab_OneC1}, an approach similar to that described in Section~\ref{sec:T7Step2} is used. For the block results $CCC \beta$ and $\alpha CCC$, four slots are required in step 2; for the block results $CC \beta \beta$, $\alpha \alpha CC$, $CCCE$, $CCC \alpha$, $CC \beta C$, $C \alpha CC$, $C \beta CC$, $CC \alpha C$, $\beta CCC$, and $ECCC$, two slots are required in step 2;  for the rest of the  block results, one slot is sufficient in step 2 to resolve the ambiguity. For example, in case the block result $CCC \beta$ occurs, in the first (respectively, second) of the four required slots, each active $\mathcal{T}_1$ (respectively, $\mathcal{T}_2$) node transmits symbol $\alpha$. If the slot result is empty, it implies that all $\mathcal{T}_1$ (respectively, $\mathcal{T}_2$) nodes are inactive, else at least one $\mathcal{T}_1$ (respectively, $\mathcal{T}_2$) node is active. In the third of the four required slots, each active $\mathcal{T}_5$ (respectively, $\mathcal{T}_6$) node transmits symbol $\alpha$ (respectively, $\beta$). Only three outcomes are possible: if the outcome is $\alpha$ (respectively, $\beta$), it implies that one $\mathcal{T}_5$ (respectively, $\mathcal{T}_6$) node is active and none of the $\mathcal{T}_6$ (respectively, $\mathcal{T}_5$) nodes are active. If the outcome is $E$, it implies that all the nodes of $\mathcal{T}_5$ and $\mathcal{T}_6$ are inactive. Similarly, in the last of the four required slots, each active $\mathcal{T}_7$ (respectively, $\mathcal{T}_8$) node transmits symbol $\alpha$ (respectively, $\beta$). Only three outcomes are possible: if the outcome is $\alpha$ (respectively, $\beta$), it implies that one $\mathcal{T}_7$ (respectively, $\mathcal{T}_8$) node is active and none of the $\mathcal{T}_8$ (respectively, $\mathcal{T}_7$) nodes are active. If the outcome is $E$, it implies that all the nodes of $\mathcal{T}_7$ and $\mathcal{T}_8$ are inactive. 
Finally, in case the result $CCCC$ occurs in a block of step 1, then the set of all node types is divided into two groups of size $4$ each:
$\{\mathcal{T}_1,\mathcal{T}_2,\mathcal{T}_3,\mathcal{T}_4\}$ and  $\{\mathcal{T}_5,\mathcal{T}_6,\mathcal{T}_7,\mathcal{T}_8\}$, and the scheme for $T = 4$ is used twice in step 2 to resolve the ambiguity. 

\begin {table} 
\centering
\caption{This table shows the deductions of the MBS about the activity or inactivity of various types of nodes for each of the outcomes in which one, two, or three collisions occur in a block and there is uncertainty about the activity or inactivity of some of the nodes.}
\label{Tab_OneC1}
\resizebox{\columnwidth}{!}{%
\begin{tabular}{|c|c|c|c|c|c|c|} 
\hline
\multicolumn{4} { | c | } {Outcome in Block $B_i$}  & \multicolumn{3} { | c |} {Types}\\ 
\hline 
Slot 1 & Slot 2 & Slot 3 & Slot 4 & Active & Inactive & Not Sure \\
\hline
$C$ & $\beta$ & $\beta$ & $\beta$ & $1$ & $2, 3,4,5,6$ & One of $\{7, 8\}$\\ \hline
$\alpha$ & $\alpha$ & $\alpha$ & $C$ & $5$ & $1,2, 6,7,8$ & One of $\{3,4\}$\\ \hline
$C$ & $C$ & $E$ & $E$ & $2$ & $3,4,5, 6,7,8$ & $1$\\
\hline
$C$ & $C$ & $E$ & $\beta$ & $2, 5$ & $3,4, 6,7,8$ & $1$ \\
\hline
$C$ & $C$ & $\alpha$ & $E$ & $2, 3$ & $4, 5,6,7,8$ & $1$\\
\hline
$C$ & $C$ & $\alpha$ & $\alpha$ & $2, 4$ & $3,5, 6,7,8$ & $1$\\
\hline
$C$ & $C$ & $\alpha$ & $\beta$ & $2, 3,5$ & $4, 6,7,8$ & $1$\\
\hline
$C$ & $C$ & $\beta$ & $\beta$ & $2$ & $3,4,5$ & $1$,   $\text{One of } \{ 6, 7, 8\}$\\
\hline
$C$ & $\alpha$ & $\alpha$ & $C$ & $1, 5$ & $2, 6,7,8$ & One of $\{3 ,4\}$\\
\hline
$C$ & $\beta$ & $\beta$ & $C$ & $1, 5$ & $2, 3, 4, 6$ & One of $\{7, 8\}$ \\
\hline
$E$ & $E$ & $C$ & $C$ & $6$ & $1,2,3,4, 7,8$ & $5$ \\
\hline
$E$ & $\beta$ & $C$ & $C$ & $6, 7$ & $1,2,3,4, 8$ & $5$ \\
\hline
$\alpha$ & $E$ & $C$ & $C$ & $1, 6$ & $2,3,4,7,8$ & $5$ \\
\hline
$\alpha$ & $\alpha$ & $C$ & $C$ &  $6$ & $1,7, 8$ & $5$,   $\text{One of } \{ 2, 3, 4\}$\\
\hline
$\alpha$ & $\beta$ & $C$ & $C$ & $1, 6, 7$ & $2,3,4,8$ & $5$  \\
\hline
$\beta$ & $\beta$ & $C$ & $C$ &  $6, 8$ & $1, 2,3,4,7$ & $5$    \\
\hline
$C$ & $C$ & $C$ & $E$ &  $3$ & $4,5,6,7,8$ & $1, 2$    \\
\hline
$C$ & $C$ & $C$ & $\alpha$ & $3, 4$ & $5,6,7,8$ & $1, 2$  \\
\hline
$C$ & $C$ & $C$ & $\beta$ &  $3$ & $4$  & $1, 2$, One of $\{5, 6, 7, 8\}$\\
\hline
$C$ & $C$ & $E$ & $C$ &  $2, 5$ & $3, 4, 6, 7, 8$ &  $1$\\
\hline
$C$ & $C$ & $\alpha$ & $C$ &  $2, 5$ &  $6, 7, 8$ & $1$, One of $\{3, 4\}$\\
\hline
$C$ & $C$ & $\beta$ & $C$ &  $2, 5$ & $3,4$ & $1$, One of $\{6, 7, 8\}$ \\
\hline
$C$ & $E$ & $C$ & $C$ &  $1, 6$ & $2, 3, 4,7,8$  & $5$\\
\hline
$C$ & $\alpha$ & $C$ & $C$ &  $1, 6$ & $7, 8$ & $5$, One of $\{2, 3, 4\}$\\
\hline
$C$ & $\beta$ & $C$ & $C$ &  $1, 6$ & $2, 3, 4$ &  $5$, One of $\{7, 8\}$ \\
\hline
$E$ & $C$ & $C$ & $C$ &  $7$ &  $1, 2, 3, 4, 8$ & $5,6$\\
\hline
$\alpha$ & $C$ & $C$ & $C$ &  $7$ & $8$ & $5, 6$, One of $\{1, 2, 3, 4\}$ \\
\hline
$\beta$ & $C$ & $C$ & $C$ &  $7, 8$ & $1, 2, 3, 4$  & $5, 6$\\
\hline
\end{tabular}%
}
\end{table}

\bibliography{BibFiles.bib}

\begin{thebibliography}{10}

\bibitem{kadam2020SPCOM}
S.~Kadam and G.~S. Kasbekar, ``{Node Cardinality Estimation Using a Mobile Base Station in a Heterogeneous Wireless Network Deployed Over a Large Region},'' in {\em 2020 International Conference on Signal Processing and Communications (SPCOM)}, pp.~1--5, 2020.

\bibitem{mozaffari2019tutorial}
M.~Mozaffari, W.~Saad, M.~Bennis, Y.-H. Nam, and M.~Debbah, ``{A Tutorial on UAVs for Wireless Networks: Applications, Challenges, and Open Problems},'' {\em IEEE Communications Surveys \& Tutorials}, vol.~21, no.~3, pp.~2334--2360, 2019.

\bibitem{kanistras2013survey}
K.~Kanistras, G.~Martins, M.~J. Rutherford, and K.~P. Valavanis, ``{A Survey of Unmanned Aerial Vehicles (UAVs) for Traffic Monitoring},'' in {\em 2013 International Conference on Unmanned Aircraft Systems (ICUAS)}, pp.~221--234, IEEE, 2013.

\bibitem{ke2016real}
R.~Ke, Z.~Li, S.~Kim, J.~Ash, Z.~Cui, and Y.~Wang, ``{Real-time Bidirectional Traffic Flow Parameter Estimation from Aerial Videos},'' {\em IEEE Transactions on Intelligent Transportation Systems}, vol.~18, no.~4, pp.~890--901, 2016.

\bibitem{giambene20195g}
G.~Giambene, E.~O. Addo, and S.~Kota, ``{5G Aerial Component for IoT Support in Remote Rural Areas},'' in {\em 2019 IEEE 2nd 5G World Forum (5GWF)}, pp.~572--577, IEEE, 2019.

\bibitem{dinh2019flying}
T.~D. Dinh, D.~T. Le, T.~T.~T. Tran, and R.~Kirichek, ``{Flying Ad-Hoc Network for Emergency Based on IEEE 802.11p Multichannel MAC Protocol},'' in {\em International Conference on Distributed Computer and Communication Networks}, pp.~479--494, Springer, 2019.

\bibitem{kadam2017fast}
S.~Kadam, C.~S. Raut, and G.~S. Kasbekar, ``{Fast Node Cardinality Estimation and Cognitive MAC Protocol Design for Heterogeneous M2M Networks},'' in {\em GLOBECOM 2017-2017 IEEE Global Communications Conference}, pp.~1--7, IEEE, 2017.

\bibitem{kadam2020fast}
S.~Kadam, C.~S. Raut, A.~D. Meena, and G.~S. Kasbekar, ``{Fast node cardinality estimation and cognitive MAC protocol design for heterogeneous machine-to-machine networks},'' {\em Wireless Networks}, vol.~26, no.~6, pp.~3929--3952, 2020.

\bibitem{liew2019probability}
J.~T. Liew, F.~Hashim, A.~Sali, M.~F.~A. Rasid, and A.~Jamalipour, ``{Probability-based Opportunity Dynamic Adaptation (PODA) of Contention Window for Home M2M Networks},'' {\em Journal of Network and Computer Applications}, vol.~144, pp.~1--12, 2019.

\bibitem{arjona2018TagID}
L.~Arjona, H.~Landaluce, A.~Perallos, and E.~Onieva, ``{Timing-Aware RFID Anti-Collision Protocol to Increase the Tag Identification Rate},'' {\em IEEE Access}, vol.~6, pp.~33529--33541, 2018.

\bibitem{liu2019TagSearch}
C.-G. Liu, I.-H. Liu, C.-D. Lin, and J.-S. Li, ``{A novel tag searching protocol with time efficiency and searching accuracy in RFID systems},'' {\em Computer Networks}, vol.~150, pp.~201--216, 2019.

\bibitem{liu2020TagSearch}
X.~Liu, J.~Yin, J.~Liu, S.~Zhang, and B.~Xiao, ``{Time Efficient Tag Searching in Large-Scale RFID Systems: A Compact Exclusive Validation Method},'' {\em IEEE Transactions on Mobile Computing}, vol.~21, no.~4, pp.~1476--1491, 2022.

\bibitem{yu2018TagSearch1}
J.~Yu, W.~Gong, J.~Liu, L.~Chen, and K.~Wang, ``{On Efficient Tree-Based Tag Search in Large-Scale RFID Systems},'' {\em IEEE/ACM Transactions on Networking}, vol.~27, no.~1, pp.~42--55, 2019.

\bibitem{yu2018TagSearch2}
J.~Yu, W.~Gong, J.~Liu, and L.~Chen, ``{Fast and Reliable Tag Search in Large-Scale RFID Systems: A Probabilistic Tree-based Approach},'' in {\em IEEE INFOCOM 2018 - IEEE Conference on Computer Communications}, pp.~1133--1141, 2018.

\bibitem{liu2022MissingTagID}
K.~Liu, L.~Chen, J.~Huang, S.~Liu, and J.~Yu, ``{Revisiting RFID Missing Tag Identification},'' in {\em IEEE INFOCOM 2022 - IEEE Conference on Computer Communications}, pp.~710--719, 2022.

\bibitem{fahim2018TagIDError}
A.~Fahim, T.~Elbatt, A.~Mohamed, and A.~Al-Ali, ``{Towards Extended Bit Tracking for Scalable and Robust RFID Tag Identification Systems},'' {\em IEEE Access}, vol.~6, pp.~27190--27204, 2018.

\bibitem{zhu2019MissingTagID}
W.~Zhu, X.~Meng, X.~Peng, J.~Cao, and M.~Raynal, ``{Collisions Are Preferred: RFID-Based Stocktaking with a High Missing Rate},'' {\em IEEE Transactions on Mobile Computing}, vol.~19, no.~7, pp.~1544--1554, 2020.

\bibitem{zang2018MissingTagID}
Y.~Zhang, S.~Chen, Y.~Zhou, and O.~Odegbile, ``{Missing-Tag Detection with Presence of Unknown Tags},'' in {\em {2018 15th Annual IEEE International Conference on Sensing, Communication, and Networking (SECON)}}, pp.~1--9, 2018.

\bibitem{liu2020TagIDContention}
J.~Liu, X.~Chen, S.~Chen, W.~Wang, D.~Jiang, and L.~Chen, ``{Retwork: Exploring Reader Network with {COTS} {RFID} Systems},'' in {\em 2020 USENIX Annual Technical Conference (USENIX ATC 20)}, pp.~889--896, USENIX Association, July 2020.

\bibitem{kadam2019rapid}
S.~Kadam, S.~V. Yenduri, P.~H. Prasad, R.~Kumar, and G.~S. Kasbekar, ``{Rapid Node Cardinality Estimation in Heterogeneous Machine-to-Machine Networks},'' {\em IEEE Transactions on Vehicular Technology}, vol.~70, no.~2, pp.~1836--1850, 2021.

\bibitem{3gpp}
P.~Panigrahi, ``{Service Requirements for Machine-Type Communications},'' {\em 3GPP}, release 16, 2014.

\bibitem{MACM2M}
A.~Rajandekar and B.~Sikdar, ``{A Survey of MAC Layer Issues and Protocols for Machine-to-Machine Communications},'' {\em IEEE Internet of Things Journal}, vol.~2, no.~2, pp.~175--186, 2015.

\bibitem{M2MmobileInternet}
G.~Wu, S.~Talwar, K.~Johnsson, N.~Himayat, and K.~D. Johnson, ``{M2M: From mobile to embedded internet},'' {\em IEEE Communications Magazine}, vol.~49, no.~4, pp.~36--43, 2011.

\bibitem{intelligentIOT}
O.~Bello and S.~Zeadally, ``{Intelligent Device-to-Device Communication in the Internet of Things},'' {\em IEEE Systems Journal}, vol.~10, no.~3, pp.~1172--1182, 2016.

\bibitem{qian2011cardinality}
C.~Qian, H.~Ngan, Y.~Liu, and L.~M. Ni, ``{Cardinality Estimation for Large-Scale RFID Systems},'' {\em IEEE Transactions on Parallel and Distributed Systems}, vol.~22, no.~9, pp.~1441--1454, 2011.

\bibitem{arjona2017scalable}
L.~Arjona, H.~Landaluce, A.~Perallos, and E.~Onieva, ``{Scalable RFID Tag Estimator with Enhanced Accuracy and Low Estimation Time},'' {\em IEEE Signal Processing Letters}, vol.~24, no.~7, pp.~982--986, 2017.

\bibitem{hou2015PLACE}
Y.~Hou, J.~Ou, Y.~Zheng, and M.~Li, ``{PLACE: Physical layer cardinality estimation for large-scale RFID systems},'' in {\em {2015 IEEE Conference on Computer Communications (INFOCOM)}}, pp.~1957--1965, 2015.

\bibitem{liu2019NCENetworked}
J.~Liu, Y.~Zhang, S.~Chen, M.~Chen, and L.~Chen, ``{Collision-resistant Communication Model for State-free Networked Tags},'' in {\em {2019 IEEE 39th International Conference on Distributed Computing Systems (ICDCS)}}, pp.~656--665, 2019.

\bibitem{lin2019NCETash}
Q.~Lin, L.~Yang, C.~Duan, and Z.~An, ``{Tash: Toward Selective Reading as Hash Primitives for Gen2 RFIDs},'' {\em IEEE/ACM Transactions on Networking}, vol.~27, no.~2, pp.~819--834, 2019.

\bibitem{zhou2018NCETimeVarying}
Z.~Zhou and B.~Chen, ``{RFID Counting over Time-Varying Channels},'' in {\em {IEEE INFOCOM 2018 - IEEE Conference on Computer Communications}}, pp.~1142--1150, 2018.

\bibitem{ng2020NCEBeacon}
P.~C. Ng, J.~She, P.~Spachos, and R.~Ran, ``{A Fast Item Identification and Counting in Ultra-dense Beacon Networks},'' in {\em {GLOBECOM 2020 - 2020 IEEE Global Communications Conference}}, pp.~1--6, 2020.

\bibitem{deng2015NCEAlohaFrameCollisonRate}
D.-J. Deng, C.-C. Lin, T.-H. Huang, and H.-C. Yen, ``{On Number of Tags Estimation in RFID Systems},'' {\em IEEE Systems Journal}, vol.~11, no.~3, pp.~1395--1402, 2017.

\bibitem{shahzad2014NCE(ART)}
M.~Shahzad and A.~X. Liu, ``{Fast and Accurate Estimation of RFID Tags},'' {\em IEEE/ACM Transactions on Networking}, vol.~23, no.~1, pp.~241--254, 2015.

\bibitem{shahzad2012NCE(ART)}
M.~Shahzad and A.~X. Liu, ``{Every Bit Counts: Fast and Scalable RFID Estimation},'' in {\em {Proceedings of the 18th Annual International Conference on Mobile Computing and Networking}}, MobiCom '12, (New York, NY, USA), p.~365–376, Association for Computing Machinery, 2012.

\bibitem{ferreira2019NCEAlohaFrame}
H.~P. Ferreira, F.~M. Assis, and A.~R. Serres, ``{Novel RFID method for faster convergence of tag estimation on dynamic frame size ALOHA algorithms},'' {\em IET Communications}, vol.~13, no.~9, pp.~1218--1224, 2019.

\bibitem{han2010NCEAnonymity}
H.~Han, B.~Sheng, C.~C. Tan, Q.~Li, W.~Mao, and S.~Lu, ``{Counting RFID Tags Efficiently and Anonymously},'' in {\em 2010 Proceedings IEEE INFOCOM}, pp.~1--9, 2010.

\bibitem{xi2020NCEInAndOut}
Z.~Xi, X.~Liu, J.~Luo, S.~Zhang, and S.~Guo, ``{Fast and Reliable Dynamic Tag Estimation in Large-Scale RFID Systems},'' {\em {IEEE Internet of Things Journal}}, vol.~8, no.~3, pp.~1651--1661, 2021.

\bibitem{nguyen2019NCEAlohaExpectationMaximization}
C.~T. Nguyen, V.-D. Nguyen, and A.~T. Pham, ``{Tag Cardinality Estimation Using Expectation-Maximization in ALOHA-Based RFID Systems With Capture Effect and Detection Error},'' {\em IEEE Wireless Communications Letters}, vol.~8, no.~2, pp.~636--639, 2019.

\bibitem{hasan2018NCEGaussianDistribution}
M.~M. Hasan, S.~Wei, and R.~Vaidyanathan, ``{Estimation of RFID Tag Population Size by Gaussian Estimator},'' in {\em 2018 IEEE International Conference on Communications (ICC)}, pp.~1--6, 2018.

\bibitem{zheng2014NCE(ZOE)}
Y.~Zheng and M.~Li, ``{Towards More Efficient Cardinality Estimation for Large-Scale RFID Systems},'' {\em IEEE/ACM Transactions on Networking}, vol.~22, no.~6, pp.~1886--1896, 2014.

\bibitem{zheng2013NCE(ZOE)}
Y.~Zheng and M.~Li, ``{ZOE: Fast cardinality estimation for large-scale RFID systems},'' in {\em 2013 Proceedings IEEE INFOCOM}, pp.~908--916, 2013.

\bibitem{li2010NCEEnergyMin}
T.~Li, S.~Wu, S.~Chen, and M.~Yang, ``{Energy Efficient Algorithms for the RFID Estimation Problem},'' in {\em {2010 Proceedings IEEE INFOCOM}}, pp.~1--9, 2010.

\bibitem{zheng2011NCE(PET)}
Y.~Zheng and M.~Li, ``{PET: Probabilistic Estimating Tree for Large-Scale RFID Estimation},'' {\em IEEE Transactions on Mobile Computing}, vol.~11, no.~11, pp.~1763--1774, 2012.

\bibitem{wang2022NCENoisyChannel}
B.~Wang and G.~Duan, ``{A reliable cardinality estimation for missing tags over a noisy channel},'' {\em Computer Communications}, vol.~188, pp.~125--132, 2022.

\bibitem{xiao2016NCEInAndOut}
Q.~Xiao, B.~Xiao, S.~Chen, and J.~Chen, ``{Collision-Aware Churn Estimation in Large-Scale Dynamic RFID Systems},'' {\em IEEE/ACM Transactions on Networking}, vol.~25, no.~1, pp.~392--405, 2017.

\bibitem{kodialam2007NCE(EZB)}
M.~Kodialam, T.~Nandagopal, and W.~C. Lau, ``{Anonymous Tracking Using RFID Tags},'' in {\em {IEEE INFOCOM 2007 - 26th IEEE International Conference on Computer Communications}}, pp.~1217--1225, 2007.

\bibitem{lodialam2006NCE}
M.~Kodialam and T.~Nandagopal, ``{Fast and Reliable Estimation Schemes in RFID Systems},'' in {\em Proceedings of the 12th Annual International Conference on Mobile Computing and Networking}, MobiCom '06, (New York, NY, USA), p.~322–333, ACM, 2006.

\bibitem{bui2017novel}
A.-T.~H. Bui, C.~T. Nguyen, T.~C. Thang, and A.~T. Pham, ``{A Novel Effective DQ-Based Access Protocol with Load Estimation for Massive M2M Communications},'' in {\em 2017 IEEE Globecom Workshops (GC Wkshps)}, pp.~1--7, IEEE, 2017.

\bibitem{lugo2021NCEM2MLoadEstimationMPRExpectationMaximization}
A.~G. Orozco-Lugo, M.~Lara, V.~Sandoval-Curmina, and G.~M. Galvan-Tejada, ``{Offered load estimation in random access multipacket perception systems using the expectation-maximization algorithm},'' {\em Signal Processing}, vol.~179, p.~107827, 2021.

\bibitem{zhou2014understanding}
Z.~Zhou, B.~Chen, and H.~Yu, ``{Understanding RFID Counting Protocols},'' {\em IEEE/ACM Transactions on Networking}, vol.~24, no.~1, pp.~312--327, 2014.

\bibitem{han2010counting}
H.~Han, B.~Sheng, C.~C. Tan, Q.~Li, W.~Mao, and S.~Lu, ``{Counting RFID Tags Efficiently and Anonymously},'' in {\em 2010 Proceedings IEEE INFOCOM}, pp.~1--9, IEEE, 2010.

\bibitem{xiao2019estimating}
Q.~Xiao, Y.~Zhang, S.~Chen, M.~Chen, J.~Liu, G.~Cheng, and J.~Luo, ``{Estimating Cardinality of Arbitrary Expression of Multiple Tag Sets in a Distributed RFID System},'' {\em IEEE/ACM Transactions on Networking}, vol.~27, no.~2, pp.~748--762, 2019.

\bibitem{lee2019NCEHeteroIDFree}
T.-K. Lee, C.-C. Chen, Y.~Ren, C.-K. Lin, and Y.-C. Tseng, ``{ID-Free Multigroup Cardinality Estimation for Massive RFID Tags in IoT},'' in {\em 2019 IEEE VTS Asia Pacific Wireless Communications Symposium (APWCS)}, pp.~1--6, 2019.

\bibitem{xiao2019NCEHeteroSnapshot}
Q.~Xiao, Y.~Zhang, S.~Chen, M.~Chen, J.~Liu, G.~Cheng, and J.~Luo, ``{Estimating Cardinality of Arbitrary Expression of Multiple Tag Sets in a Distributed RFID System},'' {\em IEEE/ACM Transactions on Networking}, vol.~27, no.~2, pp.~748--762, 2019.

\bibitem{zhang2020NCEHeteroSnapshotAnonymity}
Y.~Zhang, S.~Chen, Y.~Zhou, O.~O. Odegbile, and Y.~Fang, ``{Efficient Anonymous Temporal-Spatial Joint Estimation at Category Level Over Multiple Tag Sets With Unreliable Channels},'' {\em IEEE/ACM Transactions on Networking}, vol.~28, no.~5, pp.~2174--2187, 2020.

\bibitem{liu2016NCEHeteroFrameBinarySearch}
X.~Liu, K.~Li, A.~X. Liu, S.~Guo, M.~Shahzad, A.~L. Wang, and J.~Wu, ``{Multi-Category RFID Estimation},'' {\em IEEE/ACM Transactions on Networking}, vol.~25, no.~1, pp.~264--277, 2017.

\bibitem{sesha2019rapid}
S.~V. Yenduri, P.~H. Prasad, R.~Kumar, S.~Kadam, and G.~S. Kasbekar, ``{Rapid Node Cardinality Estimation in Heterogeneous Machine-to-Machine Networks},'' in {\em 2019 IEEE 89th Vehicular Technology Conference (VTC2019-Spring)}, pp.~1--7, 2019.

\bibitem{cohen2017CECompNetsOnlineML}
R.~Cohen and Y.~Nezri, ``{Cardinality Estimation in a Virtualized Network Device Using Online Machine Learning},'' {\em IEEE/ACM Transactions on Networking}, vol.~27, no.~5, pp.~2098--2110, 2019.

\bibitem{ullah2021CESwitchingnets}
A.~Ullah, P.~Reviriego, A.~Akram, and M.~N. Siraj, ``{Switch-Based High Cardinality Node Detection},'' {\em IEEE Embedded Systems Letters}, vol.~13, no.~4, pp.~190--193, 2021.

\bibitem{wang2022NCEDataNetsOnlineCE2}
H.~Wang, C.~Ma, S.~Chen, and Y.~Wang, ``{Fast and Accurate Cardinality Estimation by Self-Morphing Bitmaps},'' {\em IEEE/ACM Transactions on Networking}, vol.~30, no.~4, pp.~1674--1688, 2022.

\bibitem{wang2022NCEDataNetsOnlineCE1}
H.~Wang, C.~Ma, S.~Chen, and Y.~Wang, ``{Online Cardinality Estimation by Self-morphing Bitmaps},'' in {\em {2022 IEEE 38th International Conference on Data Engineering (ICDE)}}, pp.~1--13, 2022.

\bibitem{hou2018NCEPHY1}
Y.~Hou and Y.~Zheng, ``{PHY-Tree: Physical Layer Tree-Based RFID Identification},'' {\em IEEE/ACM Transactions on Networking}, vol.~26, no.~2, pp.~711--723, 2018.

\bibitem{lin2022compact}
K.~Lin, H.~Chen, N.~Yan, Z.~Ni, and Z.~Li, ``{Compact Unknown Tag Identification for Large-Scale RFID Systems},'' in {\em 2022 18th International Conference on Mobility, Sensing and Networking (MSN)}, pp.~700--707, IEEE, 2022.

\bibitem{liu2023efficient}
X.~Liu, Y.~Huang, Z.~Xi, J.~Luo, and S.~Zhang, ``{An Efficient RFID Tag Search Protocol Based on Historical Information Reasoning for Intelligent Farm Management},'' {\em ACM Transactions on Sensor Networks}, 2023.

\bibitem{liu2022revisiting}
K.~Liu, L.~Chen, J.~Huang, S.~Liu, and J.~Yu, ``{Revisiting RFID Missing Tag Identification},'' in {\em IEEE INFOCOM 2022-IEEE Conference on Computer Communications}, pp.~710--719, IEEE, 2022.

\bibitem{chu2022efficient}
C.~Chu, G.~Wen, and J.~Niu, ``{Efficient and Robust Missing Key Tag Identification for Large-scale RFID Systems},'' {\em Digital Communications and Networks}, 2022.

\bibitem{kim2022eco}
H.~Kim and J.~Ben-Othman, ``{Eco-friendly Low Resource Security Surveillance Framework toward Green AI Digital Twin},'' {\em IEEE Communications Letters}, vol.~27, no.~1, pp.~377--380, 2022.

\bibitem{kim2022intelligent}
H.~Kim, J.~Ben-Othman, K.-i. Hwang, and B.~Choi, ``{Intelligent Aerial-ground Surveillance and Epidemic Prevention with Discriminative Public and Private Services},'' {\em IEEE Network}, vol.~36, no.~3, pp.~40--46, 2022.

\bibitem{gitlink}
K.~Bhargao, ``{Node Cardinality Estimation Simulations}.'' \url{https://github.com/KaustubhBhargao/Node_Cardinality_Estimation-Simulatoins/tree/master}, 2023.

\bibitem{kleinberg2006algorithm}
J.~Kleinberg and E.~Tardos, {\em {Algorithm Design}}.
\newblock Pearson Education India, 2006.

\end{thebibliography}
\bibliographystyle{ieeetr}

\newpage
\bio{SachinKadam}
{Sachin Kadam} received the B.Eng. degree in electronics and communication engineering from the People’s Education Society Institute of Technology, Bengaluru, India, in 2007, the M.Tech. degree in electrical engineering from the Indian Institute of Technology (IIT) Kanpur, India, in 2012, and the Ph.D. degree from IIT Bombay, India, in 2020. He is currently working as an Engineering Manager (Networking domain) in Technology Innovation Hub Foundation for IoT and IoE, IIT Bombay campus, Mumbai, India. His research interests include the design and analysis of wireless and M2M networks, semantic communications, differential privacy, and learning. He was a recipient of Scholarship Foundation for Excellence, California, USA, during the B.Eng. degree.
\endbio

\bio{Kaustubh}
{Kaustubh S. Bhargao} is pursuing a Dual Degree (B.Tech and M.Tech) at the Department of Electrical Engineering, Indian Institute of Technology (IIT) Bombay, Mumbai, India. He will be joining as a Platform Architecture Engineer at Qualcomm, Bengaluru, India, in 2023. His research interests include communication networks, hardware security and instrumentation design.
\endbio

\bio{Gaurav}
{Gaurav S. Kasbekar} received the B.Tech. degree in Electrical Engineering from the Indian Institute of Technology (IIT) Bombay, Mumbai, India, in 2004, the M.Tech. degree in Electronics Design and Technology (EDT) from the Indian Institute of Science (IISc), Bangalore, India, in 2006, and the Ph.D. degree from the University of Pennsylvania, Philadelphia, PA, USA, in 2011. He is currently an Associate Professor with the Department of Electrical Engineering, IIT Bombay. His research interests include communication networking and network security. He received the CEDT Design Medal for being adjudged the best Masters student in EDT at IISc.  
\endbio
\end{document}